\documentclass[superscriptaddress,showpacs,floatfix,nofootinbib,nobibnotes,aps,prl,twocolumn,amsmath,amssymb,10pt]{revtex4-1}

\usepackage{lipsum}
\usepackage{bm}
\usepackage{lmodern}
\usepackage{tcolorbox}
\usepackage[tracking=true]{microtype}
\usepackage[english]{babel}
\usepackage{epstopdf}
\usepackage{placeins}
\usepackage{hyperref}
\usepackage{subcaption}
\usepackage[normalem]{ulem}

\usepackage{mathtools}
\usepackage{tikz}
\usetikzlibrary{arrows.meta,calc,arrows}
\usepackage{pgfplots}
\pgfplotsset{width=7cm,compat=1.3}

\newcommand{\hiddennote}[1]{}

\newcommand{\xmin}{x_{\text{min}}}
\newcommand{\pv}{\text{p-value}}

\newcommand{\bB}{\mathbf{B}}
\newcommand{\bL}{\mathbf{L}}
\newcommand{\bC}{\mathbf{C}}
\newcommand{\bI}{\mathbf{I}}

\newcommand{\bJ}{\mathbf{J}}
\newcommand{\bV}{\mathbf{V}}
\newcommand{\bphi}{\bm{\phi}^{(\ell)}}

\newcommand{\btheta}{\bm{\theta}}

\newcommand{\bX}{\mathbf{X}}
\newcommand{\bs}{\mathbf{s}}
\newcommand{\be}{\mathbf{e}}
\newcommand{\bzeros}{\mathbf{0}}

\newcommand{\bp}{\mathbf{p}}
\newcommand{\bg}{\mathbf{g}}
\newcommand{\bd}{\mathbf{d}}

\newcommand{\bmu}{\bm{\mu}}

\newcommand{\fbar}{\bar{f}}
\newcommand{\bfl}{\mathbf{f}}
\newcommand{\bfell}{\mathbf{f}^{(\ell)}}

\newcommand{\bfbar}{\mathbf{\bar{f}}}

\newcommand{\A}{\mathcal{A}}
\newcommand{\T}{\mathcal{T}}
\newcommand{\popd}{\text{pop}^{\text{distr}}}
\newcommand{\popn}{\text{pop}^{\text{node}}}

\newcommand{\dg}{d^{\text{germany}}}
\newcommand{\Gtc}{\mathcal{G}(t).{\text{changed}}}
\newcommand{\Gc}{\mathcal{G}.{\text{changed}}}
\newcommand{\Gconn}{\mathcal{G}.{\text{connected}}}

\DeclareMathOperator*{\argmax}{arg\,max}

\renewcommand{\P}{\mathcal{P}}

\newcommand{\R}{\mathbb{R}}
\newcommand{\N}{\mathcal{N}}
\newcommand{\LL}{\mathcal{L}}
\newcommand{\G}{\mathcal{G}}
\newcommand{\eps}{\epsilon}
\newcommand{\Prob}{\mathbb{P}}

\newcommand{\C}{\mathcal{C}}

\newcommand{\fell}{f^{(\ell)}}
\newcommand{\hatxmin}{\hat{x}_{\text{min}}}
\newcommand{\ntail}{n_{\text{tail}}}
\newcommand{\nnonzero}{n_{\neq 0}}
\newcommand{\niter}{n_{\text{iter}}}
\newcommand{\nfeas}{n_{\text{feasible}}}

\newtheorem{theorem}{Theorem}[section]

\newtheorem{assumption}[theorem]{Assumption}

\newtheorem{corollary}[theorem]{Corollary}

\newtheorem{definition}[theorem]{Definition}
\newtheorem{example}[theorem]{Example}

\newtheorem{lemma}[theorem]{Lemma}

\newtheorem{proposition}[theorem]{Proposition}
\newtheorem{remark}[theorem]{Remark}

\newenvironment{proof}[1][Proof]{\textbf{#1.} }{\ \rule{0.5em}{0.5em}}

\usepackage{threeparttable}
\usepackage{algorithm}
\usepackage{algorithmicx}
\usepackage{algpseudocode}
\algnewcommand{\Inputs}[1]{%
  \State \textbf{Inputs:}
  \Statex \hspace*{\algorithmicindent}\parbox[t]{.8\linewidth}{\raggedright #1}
}
\algnewcommand{\Initialize}[1]{%
  \State \textbf{Initialize:}
  \Statex \hspace*{\algorithmicindent}\parbox[t]{.8\linewidth}{\raggedright #1}
}

\makeatletter
\newcommand*{\rom}[1]{\expandafter\@slowromancap\romannumeral #1@}
\makeatother
\begin{document}

\title{Emergence of scale-free blackout sizes in power grids}
\author{Tommaso Nesti}
\affiliation{Centrum Wiskunde and Informatica, 1098 XG Amsterdam, Netherlands.}
\author{Fiona Sloothaak}
\affiliation{Eindhoven University of Technology, 5612 AZ Eindhoven,  Netherlands.}
\author{Bert Zwart}
\affiliation{Centrum Wiskunde and Informatica, 1098 XG Amsterdam, Netherlands.}
\affiliation{Eindhoven University of Technology, 5612 AZ Eindhoven,  Netherlands.}
\date{\today}

\begin{abstract}
We model power grids as graphs with heavy-tailed sinks, which represent demand from cities, and study cascading failures on such graphs. Our analysis links the scale-free nature of blackout sizes to the scale-free nature of city sizes, contrasting previous studies suggesting that this nature is governed by self-organized criticality.  Our results are based on a new mathematical framework combining the physics of power flow with rare event analysis for heavy-tailed
distributions, and are validated using various synthetic networks and the German transmission grid.
\end{abstract}

\maketitle

Securing a reliable power grid is of tremendous societal importance due to the highly disruptive repercussions of blackouts. Yet, the study of cascading failures in power grids is a notoriously challenging problem due to its sheer size, combinatorial nature, mixed continuous and discrete processes, physics and engineering specifications~\cite{Bienstockbook, Dorfler2013, Simpson2016, schafer2018_communications, Nesti2018}. Traditional epidemics models~\cite{Watts2002simplemodelof, morone2015,Hindes2016,Pastor2001} are unsuitable for its study, as the physics of power flow are responsible for a non-local propagation of failures~\cite{hines2017}. This challenge has created extensive interest from the engineering and physics communities~\cite{motter2004, witthaut2016,yang2017_PRL1,schafer2018_energy,
Witthaut2015,Crucitti2004,Huang2006}. Analytic models determining the blackout size ignore the
microscopic dynamics of power flow, while the analysis of more realistic networks typically does not go beyond simulation studies. Therefore, a fundamental understanding of blackouts is lacking.

The total blackout size, measured in terms of number of customers affected, is known to be {\em scale-free}~\cite{Carreras2004Chaos, Dobson2007, hines2009, carreras2016}, meaning there exist constants $C, \alpha>0$ such that
\begin{equation}
\label{eq:scaling}
P(S>x) \approx C x^{-\alpha},
\end{equation}
where $\approx$ means that the ratio of both quantities approaches $1$ as $x\rightarrow\infty$.
This law, also known as the Pareto law, occurs in many applications of science and engineering~\cite{Barabasi1999, suki1994, barabasi2005, Clauset2009, Simon1955}. Its significance in our context lies in the fact that big blackouts are substantially more likely than one would infer from more conventional statistical laws. As a result, mitigation policies cannot write off extremely large blackouts as virtually impossible events, and should focus on those in equal proportion to the small, frequent ones. Given the tremendous societal impact of large blackouts\hiddennote{ref on 100 billion dollars per year cost}, understanding why~\eqref{eq:scaling} occurs can lead to focused prevention and/or mitigation policies and is therefore of major significance.

Several attempts to explain~\eqref{eq:scaling} have appeared in the literature. Using simulations, previous studies suggest that~\eqref{eq:scaling} may occur as a consequence of self-organized criticality~\cite{Bak1988,Carreras2004Chaos, Dobson2007, Bienstockbook, Cascadebook}. Specifically,~\cite{Carreras2004Chaos} compares simulation traces of a model for blackouts with those of a model that is known to exhibit self-organized criticality, and shows that the autocorrelation functions are similar. Such indirect analogies of different observables do not provide direct explanations into the precise mechanism behind~\eqref{eq:scaling}. 

Other strands of literature model the cascading mechanism as a branching process with critical offspring distribution~\cite{kim2010}, without taking physical laws of electricity into consideration.
Such models lead to blackout sizes with infinite mean, corresponding to a value of $\alpha=0.5$. While a naive parametric estimation procedure using all data would lead to values of $\alpha$ in the range $(0,1)$, modern statistical techniques focusing on the tail end of the distribution clearly indicate a finite mean blackout size~\cite{hines2009, carreras2016}.

In this Letter, we propose a radically different and much simpler explanation than the aforementioned suggestions. Our central hypothesis is that~\eqref{eq:scaling} is inherited from a similar law for the distribution of {\em city sizes}~\cite{zipf49, Simon1955, Rosen1980, Batty}. We support this claim with a careful analysis of actual data, a new mathematical framework, and supporting simulations for additional insight and validation.

To develop intuition, we view the power grid as a connected graph where nodes represent cities, which are connected by edges modeling transmission lines. Initially, this is a single fully functioning network with balanced supply and demand. After several line failures, the network breaks into disconnected sub-networks, referred to as islands. The balance between supply and demand is not guaranteed to hold in each island, and at least one island is facing a power shortage. As the sum of total demand will be proportional to the total population in the island, the size of the power shortage is proportional to the total population, which is the sum of cities in that island. We now invoke a property of sums of Pareto distributed random variables, which informally says that the sum is dominated by the maximum. In other words, the size of the largest city in this island drives the scale-free nature of the blackout. In extreme value theory, this is known as the principle of a single big jump \cite{Resnick2007, heavytailbook}.

This line of reasoning implies that city sizes and blackout sizes both have Pareto distributions with similar tail behavior. For the case of the US blackout sizes (in terms of the number of customers affected) and city sizes (in terms of population), we confirm this with historical data as summarized in Fig.~\ref{fig:historical_data}, which shows that the parameters $\alpha$ for blackout and city sizes distributions are remarkably similar, each having a finite mean. We refer to Supplemental Material~\cite{NSZ19sm}, section~\rom{2}, for details.
\begin{figure}[h!]
  \begin{subfigure}[t]{0.49\columnwidth}
  \centering
          \includegraphics[width=\textwidth,height=0.9\textwidth]{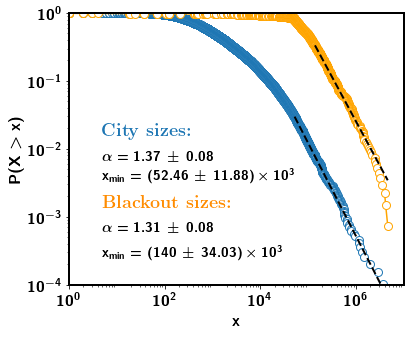}

          \label{fig:PLFIT}
\end{subfigure}\hfill
  \begin{subfigure}[t]{0.49\columnwidth}
  \centering
          \includegraphics[width=\textwidth,height=0.9\textwidth]{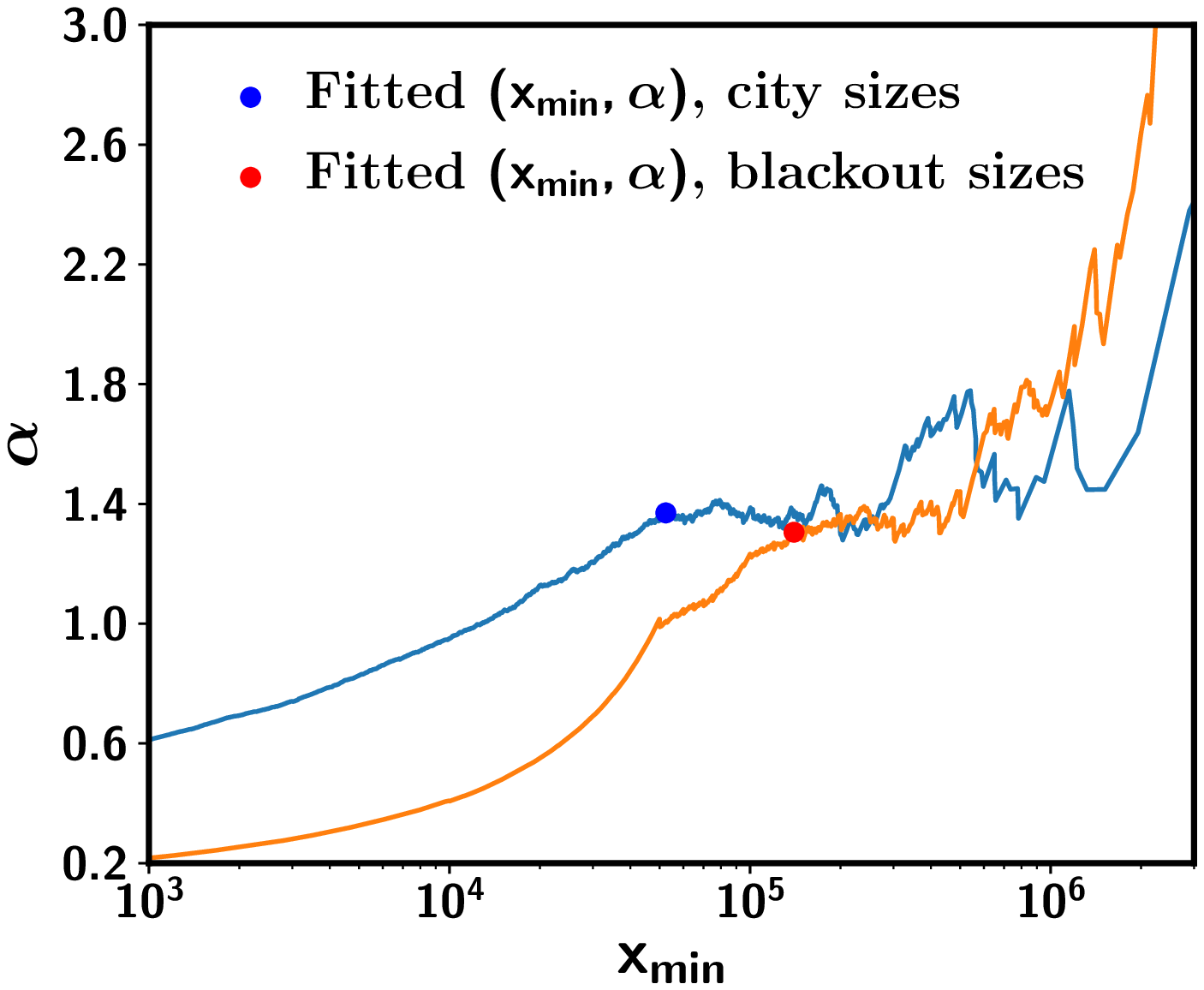}

      \label{fig:Hill}
      \end{subfigure}
      ~
      \caption{\footnotesize
      Left: Pareto tail behavior of US city~\cite{Clauset2009} and blackout sizes~\cite{OE417} in the region $x>\xmin$. Estimates
       are based on PLFIT~\cite{Clauset2009}.
        Points depict the empirical complementary cumulative distribution function (CCDF); Solid line depicts the CCDF of a Pareto distribution with parameters $\alpha,\xmin$.
     Right: Hill estimator $\xmin\to \alpha(\xmin)$, also known as the Hill plot~\cite{NSZ19sm}. The PLFIT estimates for city sizes (blue dot) and blackout sizes (red dot) lie within a relatively flat region of the graph, providing support for the Pareto fit.     
     }
       \label{fig:historical_data}
         \end{figure} 

In what follows, we make our claim rigorous by introducing
a new mathematical framework that captures the salient characteristics of actual power system dynamics~\cite{Bienstockbook} and sheds light on the connection between blackout and city sizes.  For a full account, see~\cite{NSZ19sm}, section~\rom{4}.

\begin{figure*}[ht!]
    \centering
    \begin{subfigure}[t]{\textwidth}
        \centering
\includegraphics[width=0.18\textwidth,height=0.22\textwidth]{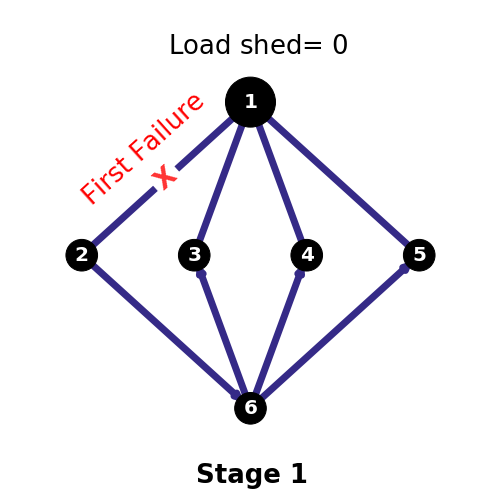}
\includegraphics[width=0.19\textwidth,height=0.22\textwidth]{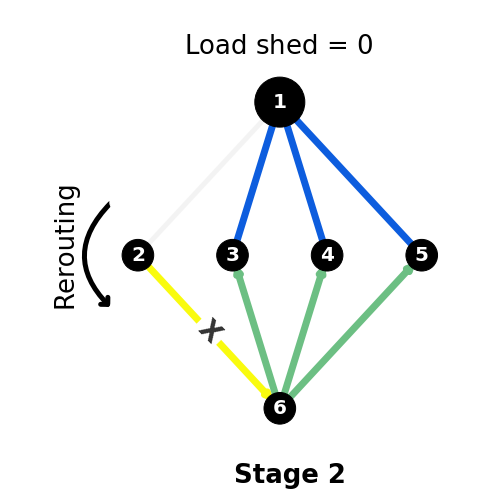}
      \includegraphics[width=0.20\textwidth,height=0.22\textwidth]{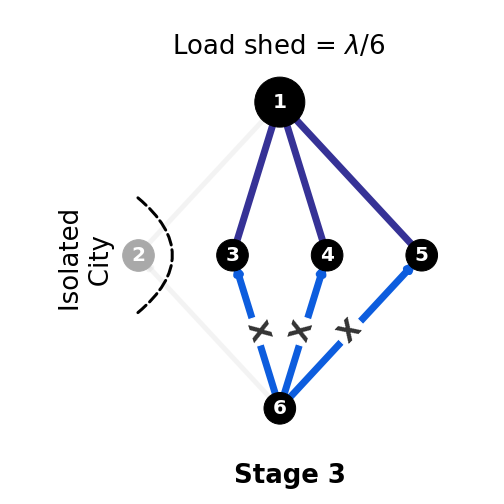}
      \includegraphics[width=0.24\textwidth,height=0.22\textwidth]{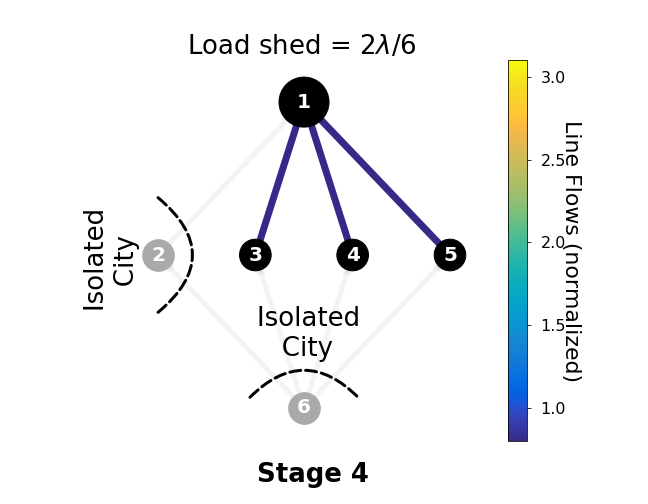}
  \end{subfigure}
%
\caption{
          \footnotesize{ Cascade in a 6-node network with $X_1=1, X_j=0$ for $j\ge 2$, $\lambda
          	> 3/4$. The four lower and upper line flows are $\lambda / 24$  and $5 \lambda / 24$, respectively, with corresponding emergency limits $1 / 24$ and $5 / 24$.  The failure of  an upper line causes the load on the adjacent lower line to surge to $\lambda / 6> 1/24$, causing this line to trip (Stage 2). This cutoff leads to the load on the three remaining lower lines to surge to $\lambda / 18$, causing them to trip as well (Stage 3). After isolating node~2 and~6, the cascade ends with $|A_1|=4$ and a total load shed of $2\lambda/6$ (Stage~4).}}
  \vspace{-3mm}
\label{fig:6nodes}
  \end{figure*}

We consider a network with $n$ nodes and $m$ lines. Node $i$ represents a city with  $X_i$ inhabitants.
We consider a static setting where each inhabitant demands one unit of energy. We assume that the $X_i$'s are independent and identically distributed Pareto random variables with $P(X>x) \approx Kx^{-\alpha}$ for constants $K, \alpha>0$. For convenience, we label the nodes such that $X_1$ represents the largest city.

For the electricity line flows, we adopt a linear DC power flow model.
This model approximates the more involved AC power flow equations, is widely used in high-voltage transmission system analysis~\cite{Purchala2005}, and accurately described the evolution of the 2011 San Diego blackout~\cite{Bernstein2014}.
Specifically, if $\bg=(g_1,...,g_n)$ and $\bX=(X_1,\ldots,X_n)$ represent the power generation and demand at each city, then the line flows $\bfl=(f_1,...,f_m)$ are given by  $\bfl = \bV (\bg-\bX)$,
where the matrix $\bV \in \mathbb{R}^{m\times n}$ is determined by the network topology and the line reactances.

Our framework consists of three stages called \textit{planning}, \textit{operational} and \textit{emergency}.
The first two stages determine the actual line limits and line flows.
We employ the widely used Direct Current Optimal Power Flow (DC-OPF) formulation with quadratic supply cost functions~\cite{Bienstockbook}:
\begin{equation}
\begin{aligned}
\label{opf}
\min_{\bg}\, &\frac 12 \sum_{i=1}^n g_{i}^2\\
\text{s.t. }&\sum_{i=1}^ng_i = \sum_{i=1}^n X_i,
\end{aligned}
\end{equation}
subject to the reliability constraint
\begin{equation}
\label{rel-constraint}
-\bfbar \leq \bV (\bg-\bX)\leq \bfbar.
\end{equation}
The planning stage concerns how the operational line limits $\bar{\bfl}$ are set. For this, we
solve~\eqref{opf} without~\eqref{rel-constraint}, yielding the uniform (across cities) solution $g^{(\text{pl})}_j=\frac{1}{n}\sum_{i=1}^n X_i$ for all $j\ge 1$, and $\bfl^{(\text{pl})}=-\bV\bX $ (see~\cite{NSZ19sm}, section~\rom{4}). Then, the operational line limits $\bfbar$ are set as
\begin{equation}
\label{eq:LineLimitsPlanning}
\bar{f}_\ell =\lambda |f_{\ell}^{(\text{pl})}|= \lambda |(\bV \bX)_\ell|, \hspace{1cm} \ell=1,...,m,
\end{equation}
where $\lambda \in (0,1]$ is a safety tuning parameter, referred to as loading factor.
In the operational stage, we solve~\eqref{opf} subject to~\eqref{rel-constraint}, yielding a different solution $\bg^{(\text{op})}$ which is not uniform due to the constraint~\eqref{rel-constraint}.
Eq.~\eqref{eq:LineLimitsPlanning} implies that line flows can have a heavy tail, which is consistent with impedance data~\cite{wang2010}. This property is essential, as it allows to create a subnetwork in which the mismatch between supply and demand is heavy-tailed.

This mismatch is established in the emergency stage, which is described next. We focus on cascades initiated by a single line failure, sampled uniformly across all lines. A line failure changes the topology of the grid and causes a global redistribution of network flows according to power flow physics. Consecutive failures occur whenever there are one or more lines for which the redistributed power flow exceeds its emergency line limit $F_\ell= \bar{f}_{\ell}/\lambda$.
Failures are assumed to occur subsequently, and take place at the line where the relative exceedance is largest. Whenever line failures create additional islands, we proportionally lower either generation or demand at all nodes to restore power balance. The cascade continues within each island until none of the remaining emergency line limits are exceeded anymore.

\normalsize
Our formulation may be extended to handle multiple initial failures, correlated city sizes, generator failures, simultaneous failures, generation limits, other strictly convex supply cost functions, and other load-shedding mechanisms. Such variations would affect the value of the pre-factor $C$, but not the exponent~$\alpha$: the tail of the blackout distribution is dominated by the scenario where there is a single city that has a large power demand, while the demand of the other cities is negligible. A formal version of this statement is that, for sufficiently small $\epsilon$,
{\small \begin{equation}
\label{keyformula}
P(S>x) = P (S > x ; X_1 > x, X_i \leq \epsilon x, i\geq 2) + {\rm o}(x^{-\alpha}).
\end{equation}
}
\normalsize
This is a mathematical description of the aforementioned principle of a single big jump.
After a normalization argument, it suffices to consider the case where $X_1=y>0$ and $X_j=0$ for $j\ge 2$.
Then, the solution of the operational DC-OPF can be computed in closed form: $g_1^{(\text{op})}=(1-\lambda (n-1)/n)y$ and $g_j^{(\text{op})}=(\lambda/n) y$ for $j\ge 2 $ (see~\cite{NSZ19sm}, Lemma~\rom{4.2}).
Let $A_1$ be the set of nodes that represents the island containing the largest city, after the cascade has stopped.
The islands that do not contain the largest city must lower their generation to zero after a disconnection, and hence immediately turn stable. Iterating, the blackout size in component~$A_1$ is given by
\begin{equation}
S = \sum_{i\in A_1} (X_i-g_i)  = \sum_{j\not\in A_1} (g_j-X_j)= \lambda\frac{n-|A_1|}{n} y.
\end{equation}
Integrating over realizations of $X_1=y, y\geq x$, and using the property of Pareto tails $\mathbb{P}(\,\max(X_1,\ldots,X_n)>x) \approx n \mathbb{P} (X>x)\approx n K x^{-\alpha}$~\cite{Resnick2007}, we find that~\eqref{eq:scaling} holds with
\begin{equation}\label{eq:constant}
C=nK \sum_{j=1}^{n-1} \mathbb{P}\left( |A_1|=j \right) \lambda^\alpha(1-j/n)^\alpha \in [0,\infty).
\end{equation}
The most delicate step, for which~\cite{NSZ19sm}, section~\rom{4.D} provides a rigorous proof, is to show that the cascade sequence does not change when performing the normalization argument in the limit $x\rightarrow\infty$, which is non trivial due to continuity issues.

In~\cite{NSZ19sm}, section~\rom{4}, we show that the pre-factor $C$ in~\eqref{eq:constant} is discontinuous at a discrete set of values of $\lambda$. At such points, the number of possible scenarios leading to a large blackout is increasing, and/or $|A_1|$ is decreasing in $\lambda$. We illustrate this in Fig.~\ref{fig:6nodes}, which also shows how the principle of a single big jump~\eqref{keyformula}, which links the total blackout size to the size of the largest city $X_1$, is realized by means of a few load shedding events, each of which is a fixed fraction of $X_1$ and corresponds to a network disconnection.


Our analysis illustrates how heavy-tailed city sizes cause heavy-tailed blackout sizes. Our modeling choices allow for a precise exploration of the cascade sequence, and inherently, an explicit formula for the blackout size tail. However, we emphasize that the essential elements that lead to heavy-tailed blackout sizes are that both the demands and the line limits are heavy-tailed. The small nodes together generate a non-negligible fraction of the demand of the large node. When the power grid satisfies these properties, then~\eqref{keyformula} continues to hold, leading to a heavy-tailed mismatch whenever there is a disconnection.
We illustrate this numerically by studying the effect of relaxing several assumptions in our framework. 

The choice of a quadratic cost function in the DC-OPF ensures that it is most efficient to divide the power generation as equally as possible among the cities, causing all cities to generate a non-negligible fraction of the total demand. Other strictly convex increasing cost functions would lead to a similar effect. Moreover, our result is robust to piecewise linear cost functions (see~\cite{NSZ19sm}, Section~\rom{6.C}), and to the inclusion of generation limits, as long as these limits are a non-negligible fraction of the total demand.

To illustrate the sensitivity of our result with respect to the chosen power flow model, we partially extend our framework to the AC power flow model. We tested its effect on multiple network topologies, and as illustrated in Fig.~\ref{fig:PLFITAC}, we conclude that city size tails still drive the blackout size tail even when the DC assumption is violated. Intuitively, the chosen power flow model determines the redistribution of flow after failures, and thus the cascade sequence. This effect is captured in the prefactor, but does not destroy the Pareto-tailed consequence in the blackout size.

An important remark is that our mathematical framework relies on the city sizes to be random variables. Naturally, city sizes are essentially fixed. 
The remaining source of randomness in our framework, namely the location of the first failure, can be interpreted as a mechanism to bootstrap linear combinations of city sizes.
It is well-known~\cite{Resnick2007} that bootstrap methods cannot recover heavy-tailed behavior if the data set is small. In order to recover a Pareto tail, the frozen network therefore needs to be sufficiently large, e.g. $10^4$ nodes.
To illustrate this, Fig.~\ref{fig:PLFITfrozen} shows simulation results for the SynGrid model, a random graph model designed to generate realistic power grid topologies~\cite{wang2010}.
Finally, Fig.~\ref{fig:PLFITfrozenUniform} reveals that Pareto-tailed city sizes is a crucial assumption in order to recover the same scale-free behavior for blackout sizes, as light-tailed city sizes do not lead to heavy-tailed blackout sizes. 
Additional supporting experiments are reported in~\cite{NSZ19sm}, section~\rom{6}.

\begin{figure}
\begin{subfigure}[t]{0.33\columnwidth}
  \centering
          \includegraphics[width=\textwidth]{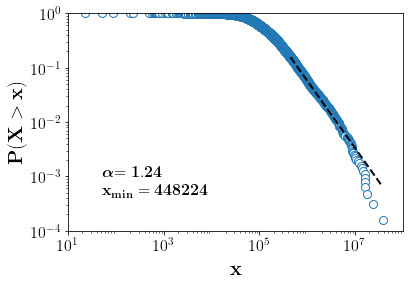}
\end{subfigure}\hfill
 \begin{subfigure}[t]{0.33\columnwidth}
  \centering
          \includegraphics[width=\textwidth]{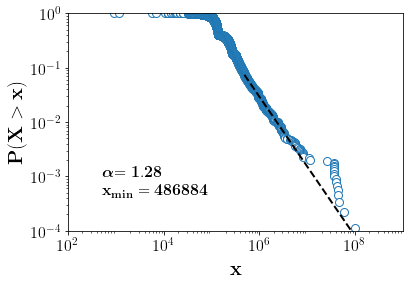} 
\end{subfigure}\hfill
 \begin{subfigure}[t]{0.33\columnwidth}
  \centering
          \includegraphics[width=\textwidth]{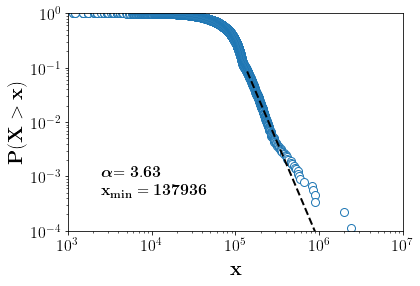}
\end{subfigure}\hfill

\begin{subfigure}[t]{0.33\columnwidth}
  \centering
          \includegraphics[width=\textwidth]{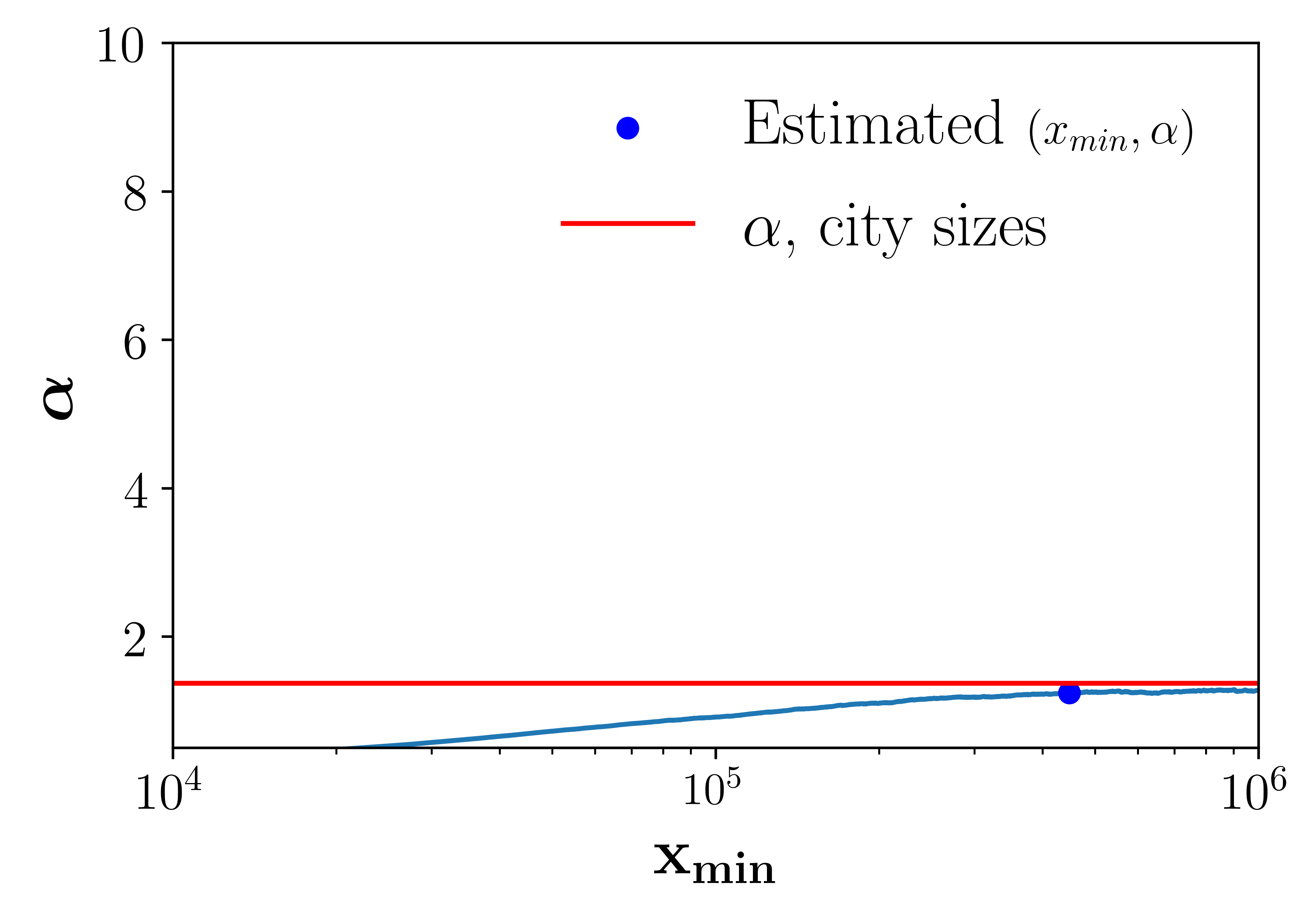}
          \caption{\footnotesize{IEEE 39-bus network, AC power flow model.}}
          \label{fig:PLFITAC}
\end{subfigure}\hfill
 \begin{subfigure}[t]{0.33\columnwidth}
  \centering
                   \includegraphics[width=\textwidth]{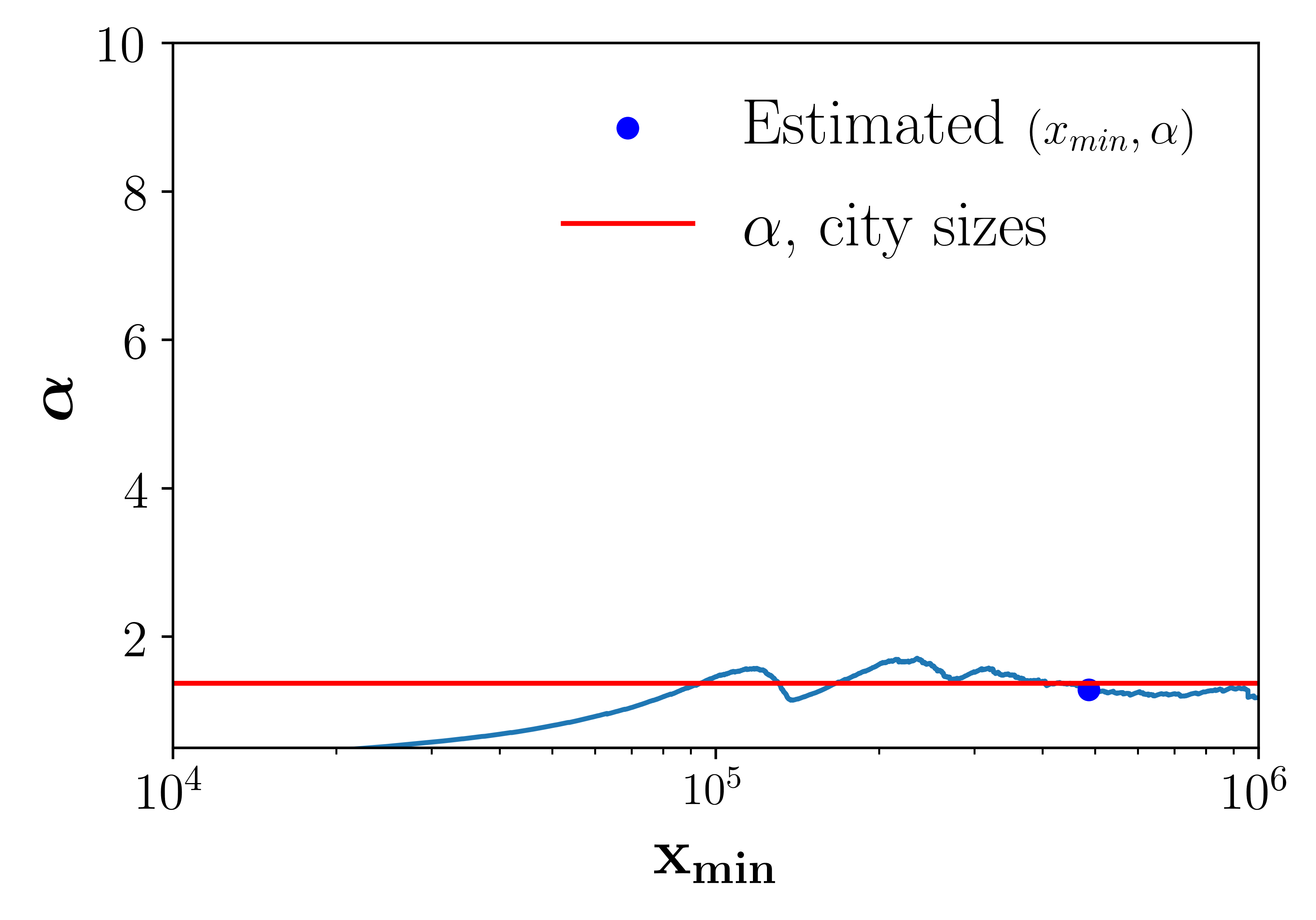}
        \caption{\footnotesize{SynGrid model, $n=10^4$, frozen city sizes.}}
           \label{fig:PLFITfrozen}
\end{subfigure}\hfill
 \begin{subfigure}[t]{0.33\columnwidth}
  \centering
                   \includegraphics[width=\textwidth]{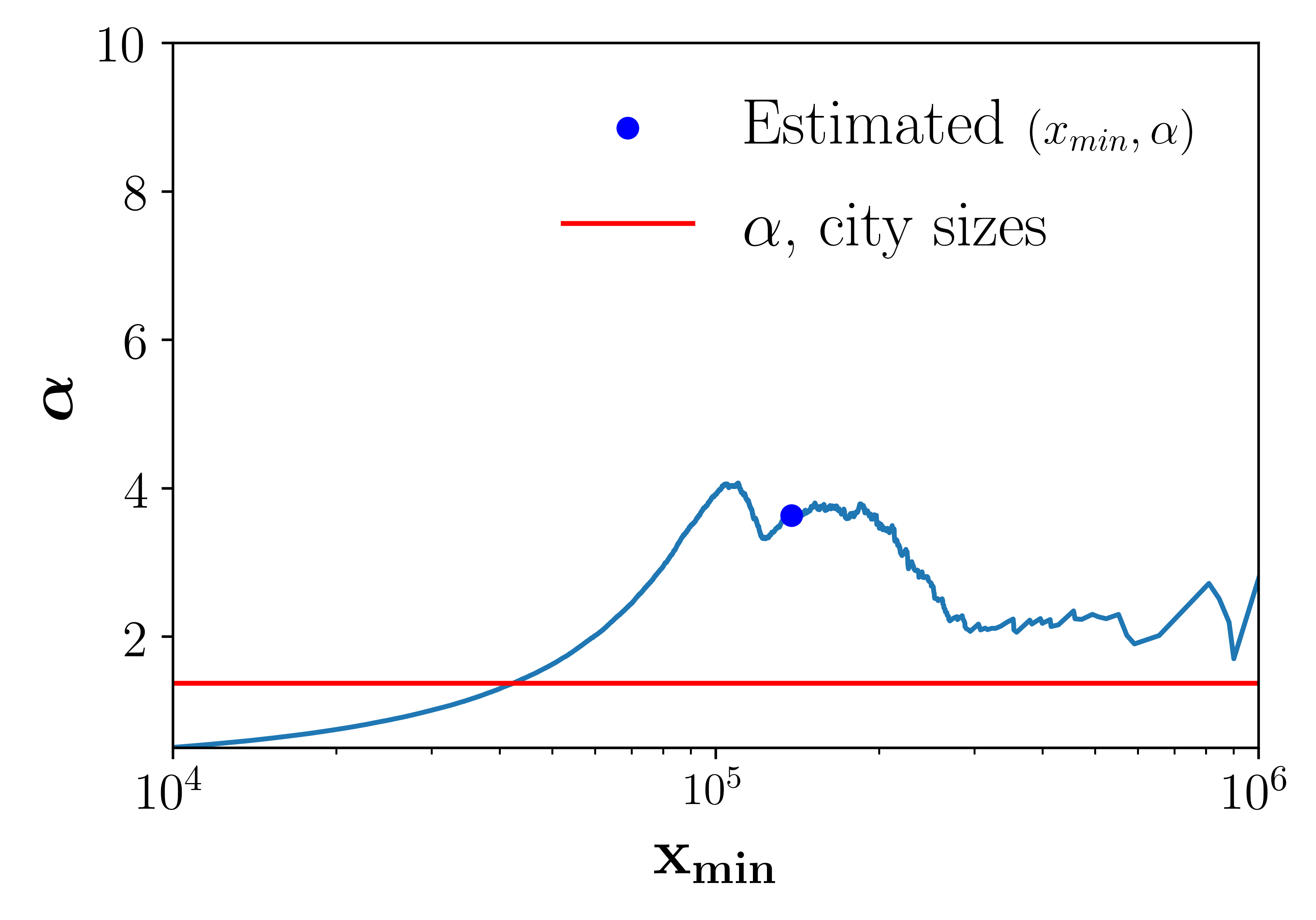}
        \caption{\footnotesize{SynGrid model, $n=10^4$, frozen city sizes sampled from uniform distribution.}}
                     \label{fig:PLFITfrozenUniform}
\end{subfigure}\hfill
\caption{\footnotesize Pareto tail behavior of simulated blackout sizes using the described cascade model with relaxed assumptions, for different topologies and loading factor $\lambda=0.9$. City sizes are sampled from a Pareto distribution with tail index $\alpha^{(\text{city})}=1.37$ in (a,b), and from a uniform distribution with the same mean in (c).
Top: points depict the empirical CCDF, dashed line depicts the CCDF of a Pareto distribution with parameters $\alpha,\xmin$, estimated via PLFIT~\cite{Clauset2009}. Bottom: Hill plots. Red line corresponds to the tail index $\alpha^{(\text{city})}$. 
A good fit is achieved when the PLFIT estimate (blue dot) lies in a flat region closely tracing the red line.
\label{fig:panel_results}
}
\end{figure}

We next present experimental results using the SciGRID network~\cite{SCIGRID00,PyPSA2017}, a model of the German transmission grid that includes generation limits and relaxes several assumptions.
 We simulate blackout realizations by considering one year's worth of hourly snapshots. 
 For each snapshot, we solve the operational DC-OPF and remove one line uniformly at random, initiating a cascade. To assign city sizes to nodes, we have cities correspond to German districts, and we assign a fraction of the population of each district to specific nodes based on a Voronoi tessellation procedure.
In this way, we account for the feature that a single city can encompass multiple nodes in a network. For more details, we refer to~\cite{NSZ19sm}, section~\rom{7}.

The German SciGRID network has a relatively small number of nodes (less than 600), and city sizes are frozen. Therefore, we do not recover Pareto-tailed blackout sizes. However, uniformly across different loading factors $\lambda$, we found that the preponderance of blackouts involves just a single load shedding event due to a network disconnection. For a moderate loading factor $\lambda = 0.7$, nearly 98\% of blackouts only involve a single disconnection. Even for a high loading factor $\lambda = 0.9$, 90\% of the blackouts involve a single disconnection, and the fraction of blackouts with four or more disconnections is below 4\%.
Fig.~\ref{fig:scigrid} depicts the largest observed
blackout, for different values of~$\lambda$.
Even in this massive blackouts, the bulk of the total load shed is the result of a few load shedding events.
These observations are typical properties that follow from our framework (see Fig.~\ref{fig:6nodes}), and sharply contrast the branching process approximations where many small jumps take place.

\begin{figure}[!htb]
\begin{subfigure}[t]{0.49\columnwidth}
     \centering
\includegraphics[width=1\textwidth]{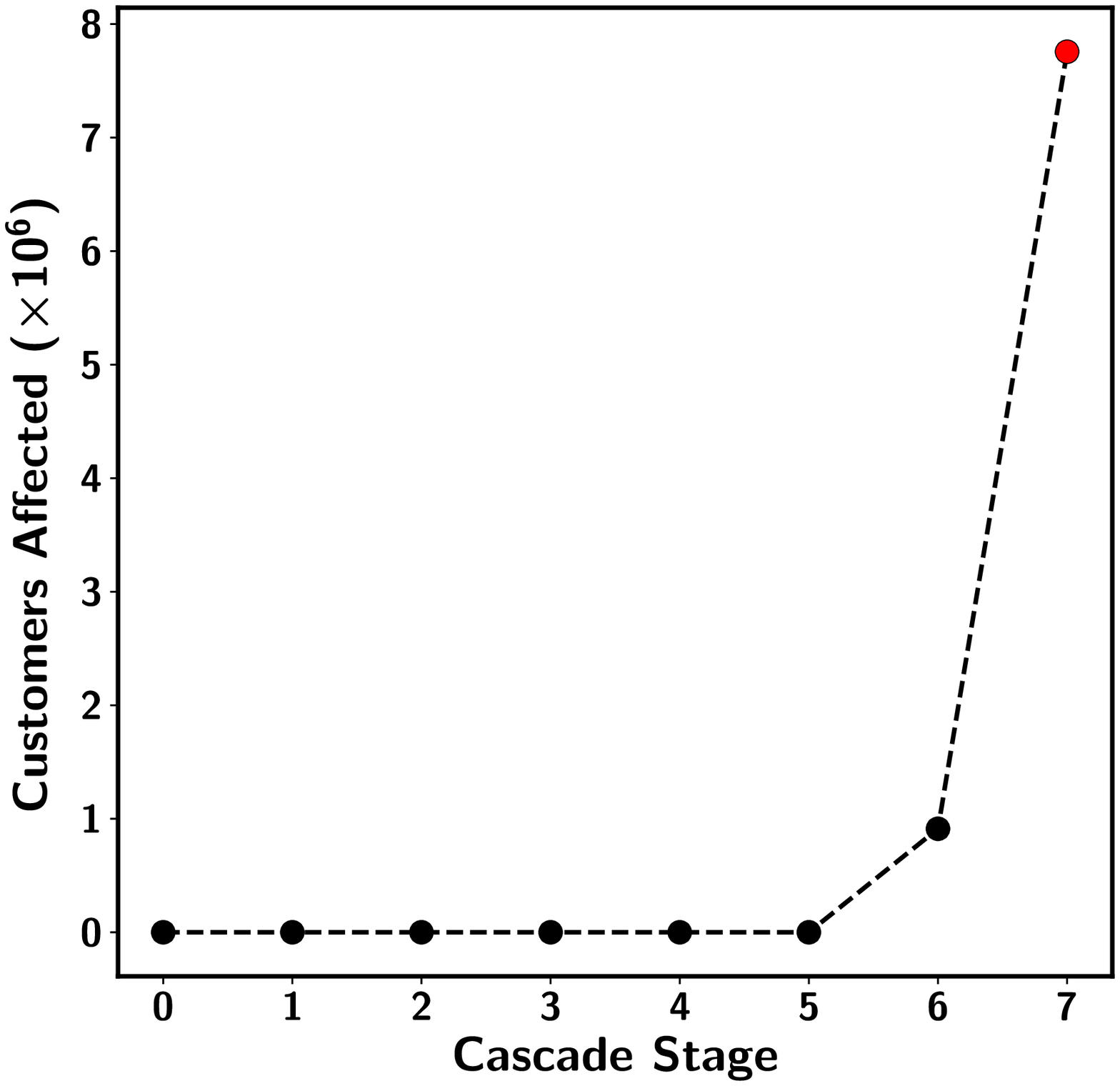}
\end{subfigure}
\hspace{-5mm}
\begin{subfigure}[t]{0.49\columnwidth}
     \centering
\includegraphics[width=1\textwidth]{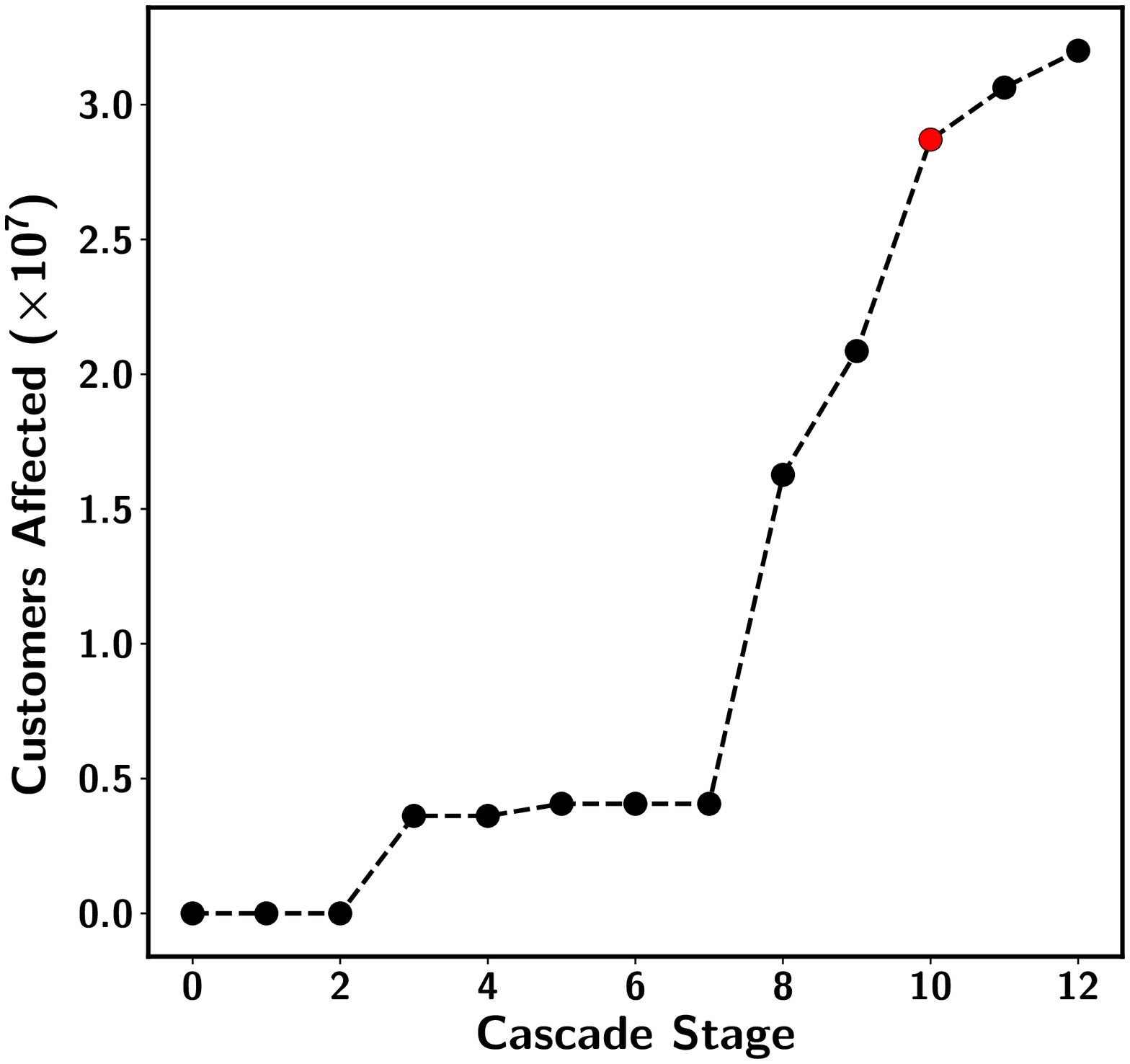}
\end{subfigure}
  \begin{subfigure}[t]{0.45\columnwidth}
     \centering
     \includegraphics[width=\textwidth]{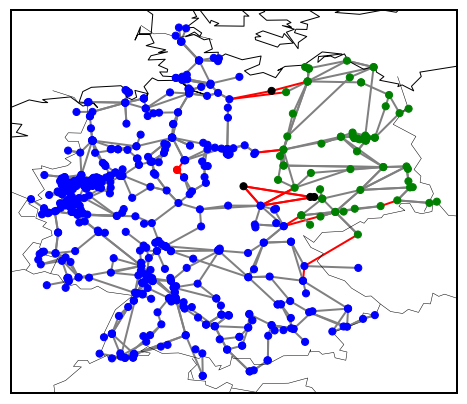}
  \end{subfigure}
\begin{subfigure}[t]{0.45\columnwidth}
     \centering
     \includegraphics[width=\textwidth]{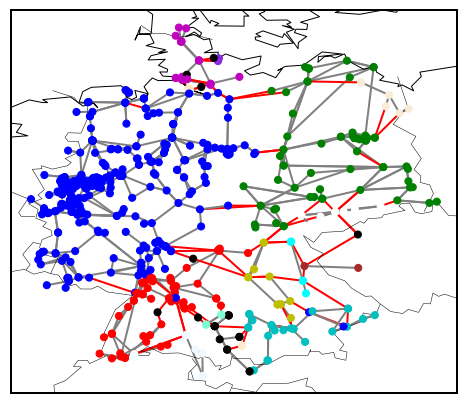}
     \end{subfigure}
          \caption{\footnotesize
Dissection of a massive blackout in the SciGRID network for loading
factors $\lambda = 0.7$ (left) and
$\lambda = 0.9$ (right) in terms of the cumulative number of affected
customers at each stage of the cascade, as displayed in the top charts with the selected stage colored red. The corresponding islanded components are visualized with different colors in the bottom pictures.}
   \label{fig:scigrid}
 \end{figure}

Using data analysis, probabilistic analysis, and simulations, we have illustrated how extreme variations in city sizes can cause the scale-free nature of blackouts.
Our explanation and refinement~\eqref{eq:constant} of the scaling law~\eqref{eq:scaling}
show that specific details such as network characteristics only appear in the pre-factor~\eqref{eq:constant}. The main parameter $\alpha$, which determines how fast the probability of a big blackout vanishes as its size grows, is completely determined by the city size distribution.
 Decreasing the constant~\eqref{eq:constant} by performing network upgrades (which in our framework is equivalent to decreasing $\lambda$) would only lead to a modest decrease in the likelihood of big blackouts. Consequently, it is questionable whether network upgrades, as considered in~\cite{Dobson2007, yang2017}, are the most effective way to mitigate the consequences of big blackouts.

Instead, it may be more effective to invest in responsive measures that enable consumers to react to big blackouts.
It is shown in~\cite{hines2009} that durations of blackouts have a tail which is decreasing much faster than~\eqref{eq:scaling}.
At the same time, production facilities often lack redundancy - even short blackouts can lead to huge costs, suggesting that the costs associated to a blackout are concave up to a certain duration.  Therefore, if the goal is to minimize the negative effects of a big blackout, it may be far more effective to invest in solutions (such as local generation and storage) that aim at surviving a blackout of a specific duration. This is consistent with recent studies on
the importance of resilient city design~\cite{Bai2018}.

Finally, our framework and insights suggest new ways of approaching scale-free phenomena in other transportation networks, such as highway traffic jams \cite{trafficjampaper}.
While transport network topologies are not scale-free, they may still exhibit scale-free behavior, caused by the scale-free nature of nodal sizes. 
\\

\noindent
{\bf Acknowledgements.}
We thank Sem Borst for useful discussions, and
the Isaac Newton Institute for support and hospitality during the program ``Mathematics of Energy Systems''.
The grants NWO 639.033.413 , NWO 024.002.003 and EPSRC EP/R014604/1 provided financial support.
 

\nocite{hill1975, 
Mieghem2010, 
VanMieghem2017, 
Tondel2003, 
Watts1998, 
Wang2018syngrid, 
Zimmerman2011, 
MolzahnHiskens, 
SCIGRID1, 
wpr, 
eurostat, 
eurostatboundaries, 
githubboundaries, 
SCIGRID2} 
\bibliography{bibliography}

\begin{thebibliography}{58}
\bibitem [{NSZ()}]{NSZ19sm}%
  \BibitemOpen
  \href@noop {} {\enquote {\bibinfo {title} {See {S}upplemental {M}aterial for
  the extended mathematical framework and the details on the statistical data
  analysis and simulations, which includes refs.~[45-58]},}\ }\BibitemShut
  {NoStop}%
\end{thebibliography}
\clearpage

\onecolumngrid
\begin{center}
\textbf{\large Supplemental Material for:\\ Emergence of scale-free blackout sizes in power grids\\}
\end{center}

\twocolumngrid
\section{Background on Pareto distribution and parameter estimation}\label{s:pareto}
A Pareto-distributed random variable $X$ with minimum value $\xmin>0$ and tail exponent $\alpha>0$ is described by its complementary cumulative distribution function (CCDF) $\bar{F}(x)$
\begin{align}
\bar{F}(x)=\mathbb{P}(X> x)=\Bigl(\frac{x}{\xmin}\Bigr)^{-\alpha},\,x\ge \xmin.
\end{align}
The expected value of $X$ is equal to $(\alpha\xmin)/(\alpha-1)$ if $\alpha>1$, and $\infty$ otherwise.

In order to analyze the power law behavior of city and blackout sizes, we use the PLFIT
method introduced in~\cite{Clauset2009} to fit a Pareto distribution to a given empirical dataset $\{x_i\}_{i=1}^N$.
The PLFIT method is based on a combination of the Hill estimator to find the
tail exponent $\alpha$, and on the Kolmogorov-Smirnov statistic to find $\xmin$, as outlined below.
For each possible choice of $\xmin$, the best-fitting tail index $\alpha$ is found via the Hill estimator~\cite{hill1975}
\[
\hat{\alpha}(\xmin)=n\Bigl[\sum_{x_i\ge \xmin} \ln \frac{x_i}{\xmin}\Bigr]^{-1}.
\] Then, the KS goodness-of-fit statistic $D(\xmin) = \max_{x\ge\xmin} |S(x) - P(x)|$ is calculated, where $S(x)$ is the empirical Cumulative Distribution Function (CDF) of the data and $P(x)$ is the CDF of the Pareto distribution with parameters $\xmin$ and $\hat{\alpha}(\xmin)$. Finally, the estimated $\hatxmin$ is the one that minimizes $D$ over all possible choices of $\xmin$.
Uncertainty in the estimated tail exponent $\hat{\alpha}(\hatxmin)$ and lower bound $\hatxmin$ is quantified via the nonparametric bootstrap method described in~\cite{Clauset2009}.
Finally, a goodness-of-fit test based on the KS statistic is used to generate a p-value that quantifies the plausibility of the power law hypothesis.
The authors in~\cite{Clauset2009} suggest to use the following (conservative) choice: the power law is \textit{ruled out} if $p \le 0.1$.

We remark that any automatic procedure for the estimation of the parameter $\xmin$ is imperfect and should be paired with additional, case-by-case analysis. For instance, it is not known whether the PLFIT estimator is consistent.
In this paper, we always couple the PLFIT procedure with the manual observations of the Hill plot, i.e. the graph of the mapping $\xmin\to\hat{\alpha}(\xmin)$, and report whether the PLFIT results are consistent with the visual analysis of this plot, i.e. whether $\hat{x}_{\text{min}}$ lies within a region where the values of $\alpha$ are relatively stable.


\section{Historical data analysis}

In this section, we analyze the scale-free behavior of US city and blackouts sizes.
The data for US city sizes, as per the 2000 US census, are available in~\cite{Clauset2009}. The data for US blackouts are extracted from the Electric Disturbance Events Annual Summaries, Form OE-417~\cite{OE417} of the US Department of Energy, which includes information on the date, area of interest and number of customers affected by outage events. Here, the size of a blackout is defined as the number of customers affected by it. The dataset covers the period 2002-2018.

Each record, or row, of the OE-417 dataset, contains information such as the date, area of interest and the number of customers affected in a single outage event.
The presence of missing or noisy records in the dataset
requires the following pre-processing actions: i) records for which the ``Number of customers affected'' entry is unknown are removed; ii) records for which the ``Number of customers affected'' consists of two or more values, corresponding to different US states, are modified by replacing the multiple values with their sum; iii) records for which the ``Number of customers affected'' entry is not purely numeric are removed. The only two exceptions to iii) are when \textit{both} the “cumulative” and “peak” number of customers affected are reported (in which case only the “cumulative” values is retained), and when the number of customers affected is described by a range of values (in which case the midpoint value is retained).

Table~\ref{tab:PLFIT_stats} reports the PLFIT estimated parameters, the corresponding standard deviations (calculated using the nonparametric bootstrap method in~\cite{Clauset2009}), as well as the KS p-values, which indicate a good fit.
The results for the $\alpha$-estimates (city sizes: $1.37\pm 0.06$; blackout sizes: $1.31\pm 0.08$ ) corroborate the claim that the scale-free behavior of blackout sizes is inherited from the power law distribution for city sizes.
\begin{table*}[ht]
\centering
\begin{tabular}{c|c|c|c|c|c}
\hline
\hline
Dataset & $N$ & $\ntail$ &$\hat{\alpha}(\hatxmin)$  & $\hatxmin$ & KS $\pv$  \\
\hline
US city sizes, 2000 Census ($\times 10^3$)
 & $19447$ &  $580$ & $1.37\pm 0.08$ & $52.5 \pm 11.6$ & $0.76$ \\
US blackout sizes, form OE-417 ($\times 10^3$) & $1341$ & $448$ & $1.31\pm 0.08$ &  $140 \pm 31.3$ & $0.32$\\
\hline
\hline
\end{tabular}
\caption{ \small PLFIT statistics for US city and blackout sizes. $\ntail$ is the number of data points $x_i\ge\hatxmin$. Standard deviations obtained via nonparametric bootstrap with $1000$ repetitions. }
\label{tab:PLFIT_stats}
\end{table*}

Fig.~\ref{fig:PLFIT} reports the CCDF and the PLFIT results, and Fig.~\ref{fig:Hill} the corresponding Hill plots. We observe that the estimated parameters lie in the flat portion of the Hill plots.

\begin{figure}[h!]
  \begin{subfigure}[t]{\columnwidth}
  \centering
                    \includegraphics[width=0.9\textwidth]{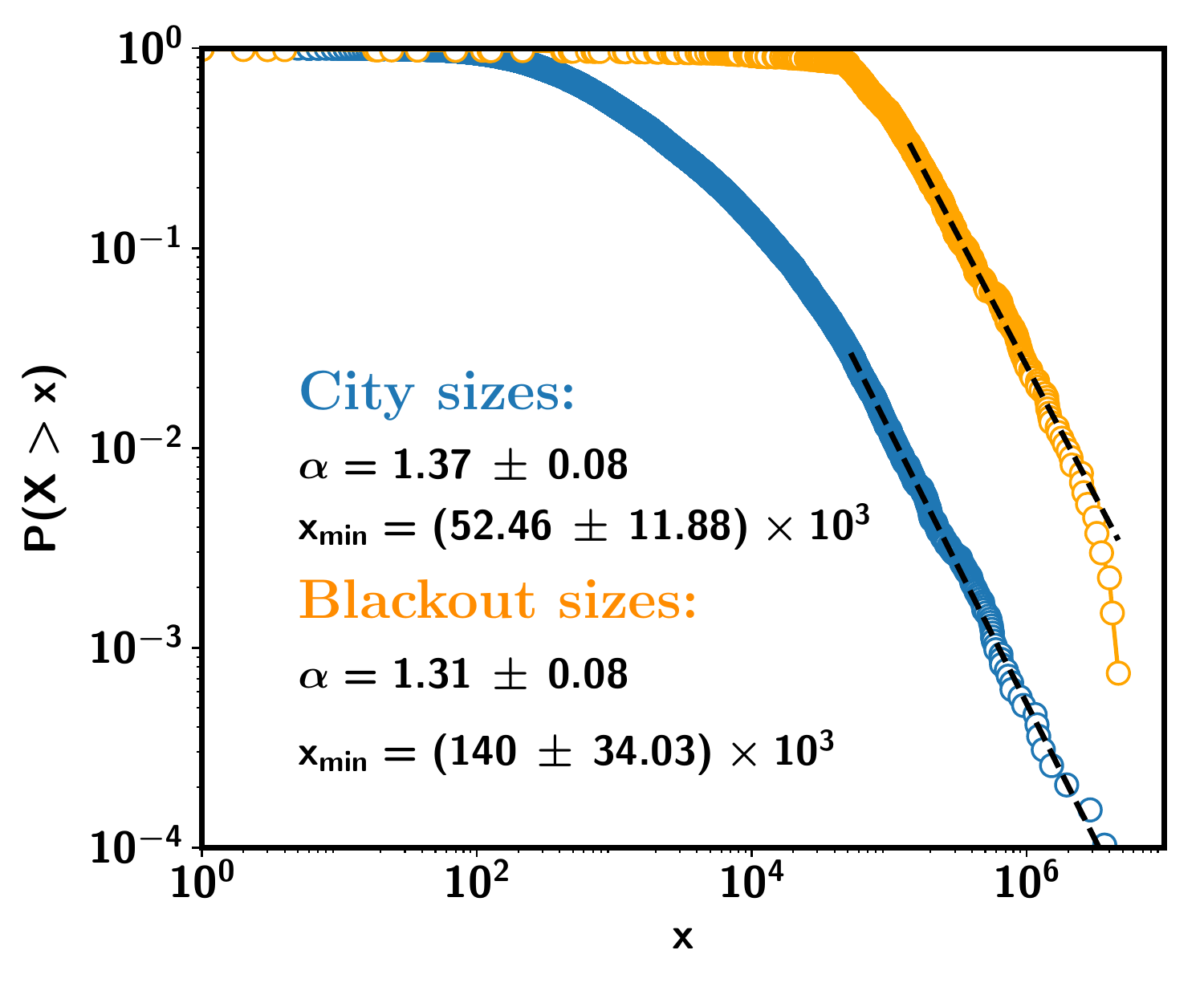}
          \caption{\small{Pareto tail behavior of US city and blackout sizes in the region $x>\xmin$. Estimates for $\alpha$ and $\xmin$, along with standard deviations, are based on PLFIT~\cite{Clauset2009}. Points represent the empirical complementary cumulative distribution function (CCDF); Solid line represents the CCDF of a Pareto distribution with parameters $\alpha,\xmin$.}}
          \label{fig:PLFIT}
\end{subfigure}\hfill
  \begin{subfigure}[t]{\columnwidth}
  \centering
          \includegraphics[width=0.9\textwidth]{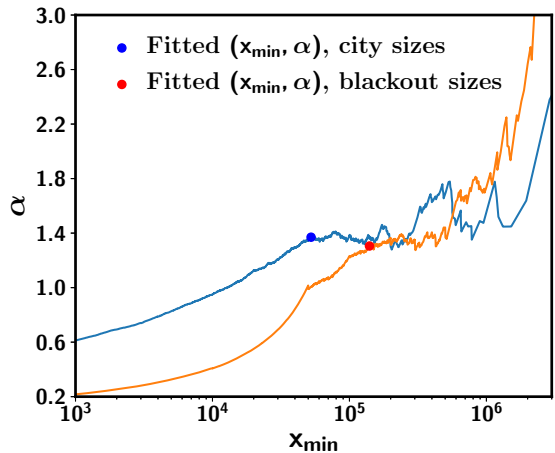}
        \caption{\small{Visualization of the estimated $\alpha$, obtained by only taking into consideration values in the region $x>\xmin$, as a function of $\xmin$. The PLFIT-estimated $\xmin$ for city sizes (blue dot) and blackout sizes (red dot) lie within a region where the values of $\alpha$ are relatively stable, substantiating the results of the PLFIT procedure.}}
      \label{fig:Hill}
      \end{subfigure}\hfill
      \caption{\small{Pareto tail behavior of US city and blackout sizes.}}
      \end{figure}

\section{Power flow}
We model the power grid as a connected graph $\mathcal{G}=\mathcal{G}(\N,\LL)$, where the set of nodes $\N$ represents the $n$ buses in the system, and the set of edges $\LL$ corresponds to the $m$ transmission lines. Let $\bg,\bd\in\R^{n}$ represent the nodal generation and load vectors, respectively, and $\bp=\bg-\bd$ be the net power injections vector. We make use of the \textit{DC approximation}, which is commonly used in high-voltage transmission system analysis~\cite{Purchala2005},
to model the relationship between active power injections $\bp$ and active line power flows $\bfl\in\R^m$, which is given by the linear mapping
\begin{equation}\label{eq:DC}
\bfl=\bV(\bg-\bd).
\end{equation}
The matrix $\bV\in\R^{m\times n}$ is known as the Power Transfer Distribution Factors (PTDF) matrix and is constructed as outlined below.

\subsection{PTDF matrix}\label{ss:properties_PTDF}
Choosing an arbitrary but fixed orientation of the transmission lines, the network structure is described by the \textit{edge-vertex incidence matrix} $\bC\in\R^{m\times n}$ defined as

\begin{equation*}
 	C_{\ell, i}=\begin{cases}
    	\phantom{-}1		&\text{if } \ell=(i,j),\\
    	-1		&\text{if } \ell=(j,i),\\
    	\phantom{-}0    	&\text{otherwise}.
    \end{cases}
\end{equation*}
Denote by $\beta_\ell>0$ the weight of edge $\ell\in\LL$, corresponding to the \textit{susceptance} of that transmission line. Note that $\beta_{\ell}=x_{\ell}^{-1}$, where $x_{\ell}$ is the reactance of line $\ell$. Denote by $\bB$ the $m \times m$ diagonal matrix defined as $\bB=\mathrm{diag}(\beta_1,\dots, \beta_m)$.
The network topology and weights are simultaneously encoded in the \textit{weighted Laplacian matrix} of the graph $G$, defined as $\bL = \bC^\top \bB \bC$ or entry-wise as
\begin{equation*}
	L_{i,j} =\begin{cases}
		-\beta_{i,j}							& \text{if } i \neq j,\\
		\sum_{k\neq j} \beta_{i,k} 	& \text{if } i=j.
	\end{cases}
\end{equation*}
All the rows of $\bL$ sum up to zero and thus the matrix $\bL$ is singular.

According to the \textit{DC approximation}, the relation between any zero-sum vector of power injections $\bp \in \R^n$ and the phase angles $\btheta \in \R^n$ can be written in matrix form as $\bp = \bL \btheta$. Defining $\bL^+ \in \R^{n \times n}$ as the \textit{Moore-Penrose} pseudo-inverse of $\bL$, we can rewrite this as
\begin{equation}
\label{eq:dcapprox}
	\btheta = \bL^+ \bp.
\end{equation}

\noindent
The line power flows $\bfl$ are related to the phase angles $\btheta$ via the linear relation $\bfl = \bB\bC\btheta$. In view of~\eqref{eq:dcapprox}, the line power flows $\bfl$ can be written as a linear
transformation of the power injections $\bp$, i.e.
\begin{equation}
\label{eq:PTDF}
	\bfl = \bV \bp,
\end{equation}
where $\bV:=\bB\bC\bL^+$ is the PTDF matrix.

The following lemma is based on a well-known result in graph theory (see, for example,~\cite{Mieghem2010}).
\begin{lemma}\label{lm:rk}
If $G$ is a connected graph, $\text{rk }(\bV)=\text{rk }(\bC)=\text{rk }(\bL)=\text{rk }(\bL^+)=n-1$, and the null space of $\bV$ is the one-dimensional subspace generated by $\be=(1,\ldots,1)\in\R^n$, i.e.
\begin{align*}\text{Ker }(\bV)=\text{Ker }(\bC)=\text{Ker }(\bL)=\text{Ker }(\bL^+)=<\be>.
\end{align*}
\end{lemma}

The following lemmas are technical results which will be needed in Section~\ref{s:CascadeModel}.
\begin{lemma}\label{lm:swap}
Changing the orientation of a subset of lines $\LL'\subset \LL $ has the effect of swapping the sign of the corresponding rows of the PTDF matrix $\bV$.  In particular, it is always possible to choose the orientation such that $\bV\be_1\ge 0$.
\end{lemma}

\begin{proof}
Changing the orientation of a line from $l_k=(i,j)$ to $\tilde{l}_k=(j,i)$, by definition, amounts to swapping the sign of the $k$-th row of matrix $\bC$, yielding a modified matrix $\tilde{\bC}=\bI^{(k)} \bC$, where $\bI^{(k)}$ is a diagonal matrix with $I_{ii}^{(k)}=1$ if $i\neq k$ and $I_{kk}^{(k)}=-1$.  Since $\tilde{\bL}=\tilde{\bC}^{\top}\tilde{\bC}=\bC^{\top} \bI^{(k)}\bI^{(k)}\bC=\bC^{\top}\bC=\bL$, the matrices $\bL$ and $\bL^+$ are not affected by the change.
As a consequence, the modified PTDF matrix $\overline{\bV}=\bB\tilde{\bC}\bL^+=\bI^{(k)}\bV$ differs from $\bV$ only by the swapped signs on the $k$-th row.
\end{proof}

\begin{lemma}\label{lm:C,V}
Let $G$ be assigned the orientation such that the set of edges incident to node $1$ is $\LL_1=\{(1,j)\,|\, j \text{ is adjacent to }1\}$, i.e.~$C_{\ell,1}=1=-C_{\ell,j}$ for all $\ell=(1,j)\in\LL_1$. Then, $V_{\ell,1}\ge 0$ for every $\ell\in\LL_1$. The converse is also true.
%

\end{lemma}
\begin{proof}
First, note that largest element in each row of $L^+$ is its diagonal entry (Corollary 1 in~\cite{VanMieghem2017}), i.e.~ $L^+_{1,1}-L^+_{1,j}\ge 0$ for every $\ell=(1,j)\in\LL_1$.
For any line $\ell=(1,j)\in\LL_1$, we have $V_{\ell,1}=(\bC\bL^+)_{\ell,1}=C_{\ell,1}L^+_{1,1}+C_{\ell,j}L^+_{1,j}$, where $C_{\ell,1}=-C_{\ell,j}=\pm 1$ depending on the orientation of line $\ell$. Thus, $V_{\ell,1}\ge 0$ if and only if $C_{\ell,1}=1=-C_{\ell,j}$.
\end{proof}

\subsection{Optimal Power Flow}\label{ss:OPF}
The Optimal Power Flow (OPF) program is
an optimization problem that determines the generation
schedule minimizing the total system generation cost
while satisfying demand/supply balance and network
physical constraints.
In its full generality, the OPF is a nonlinear, nonconvex optimization problem. For the purpose of this paper,
we will focus on a tractable approximation based on the DC power flow equations referred to as DC-OPF, which can be formulated as the following optimization problem:

\begin{alignat}{6}
& \underset{g\in\R^n}{\min} & & \sum_{i=1}^n C_i(g_i) \label{eq:objective}\\
& \text{s.t.}        & & \sum_{i=1}^n g_i=\sum_{i=1}^n d_i, &    \label{eq:balance} \\
&					  & & \underline{g}_i\le g_i \le \bar{g}_i,  &\ i \in \N,   \label{eq:gen}\\
&					  &	 -& \fbar_{\ell} \le V_{\ell}(g-d)\le \fbar_{\ell}, &\ \ell \in \LL,   \label{eq:lines}
\end{alignat}

$C_i(\cdot):\R\to\R$ denotes the cost function of generation at node $i$, $\underline{\bg},\bar{\bg}\in\R^n$ denote, respectively, the vector of nodal minimum and maximum generation capacities, and $\bfbar\in\R^m$ denotes the vector of line limits. We assume that $C_i(\cdot)$ is an increasing quadratic function. Specifically, we assume $C_i(g_i)=g_i^2/2$, $i=1,...,n$. For the purpose of this work, we do not consider generator limits, i.e. $\underline{g}_i=-\infty,\bar{g}_i=\infty$, $i=1,...,n$.

\subsection{Power flow redistribution}\label{ss:PowerFlowRedistribution}

In the event of the failure of a subset of transmission lines $\LL^{\prime}\subset \LL$, and provided that the power injections remain unchanged, the power flows will redistribute among the remaining lines according to power flow physics, provided that the altered graph $\widetilde{G}=(\mathcal{N},\LL\setminus\LL^{\prime})$ remained connected. The way the power flows redistribute is governed by the new PTDF matrix $\widetilde{\bV}$, which can be constructed analogously to $\bV$, mapping the (unchanged) power injections to the new power flows. We assume that the redistribution occurs instantaneously, without any transient effects.

As an illustration, we show how the redistributed power flows can be calculated in the special case of an isolated failure $\LL^{\prime}=\{\ell\}$. In this case, it is enough to calculate the vector $\bphi \in \R^{m-1}$ of redistribution coefficients, known as \textit{line outage distribution factors}. The quantity $\phi^{(\ell)}_j$ takes values in $[-1,1]$, and $\smash{|\phi^{(\ell)}_j|}$ represents the percentage of power flowing in line $\ell$ that is redirected to line $j$ after the failure of the former.
In particular, the new power flow configuration after the failure of line $\ell=(i,j)$, denoted by $\bfell \in\R^{m-1}$, is given by
\begin{equation}\label{eq:redistribution}
\fell_k=f_k+  f^{(\ell)}_\ell \phi^{(\ell)}_{k}, \,\forall \ell\neq k,
\end{equation}
where, for $k=(a,b)$ and $\ell =(i,j)$, the coefficient $\phi_{k,\ell} \in \R$ can be computed as
\begin{equation}
\label{eq:phi}
	\phi_{k,\ell}=\phi_{(i,j),(a,b)}= \beta_{\ell}^{-1} \cdot \frac{R_{a,j}-R_{a,i}+R_{b,i}-R_{b,j}}  {2(1-x_{i,j}^{-1} R_{i,j})},
\end{equation}
where $R_{i,j}$ is the \textit{effective resistance} between nodes $i$ and $j$, given by
 \begin{equation*}
	R_{i,j} = (\mathbf{e}_i - \mathbf{e}_j)^T L^+ (\mathbf{e}_i - \mathbf{e}_j) = (L^+)_{i,i}+(L^+)_{j,j}-2(L^+)_{i,j}.
\end{equation*}.

\section{Cascading failure model}\label{s:CascadeModel}
In view of the DC-OPF, in order to obtain a fundamental understanding of the correlation between blackout sizes and city sizes using the DC approximation model, we require a framework that adequately sets the power demand, the transmission line limits, generation limits, and the cost function for any fixed topology $\mathcal{G}= (\mathcal{N} , \mathcal{L})$. In addition, we need to specify a mechanism that causes the initial line failure, as well as which lines possibly fail next after the power flow redistribution. For this purpose, we consider a framework that consists of three problems: the \textit{planning} problem, the \textit{operational} problem, and the \textit{emergency} problem. Next, we explain our framework in more detail, followed by listing some vital properties.

\subsection{Description of the framework}\label{ss:description}
The planning problem refers to how the generation limits, the line limits and power demand are determined with respect to the city sizes $X_1,...,X_n$. We assume that each node represents a city with size $X_i$ inhabitants. For ease of presentation, we consider a framework with a static setting where each inhabitant demands one unit of energy, i.e.~$d_i=X_i$ for every $i=1,...,n$. We assume that the cost function is an increasing quadratic function and that generator limits do not pose an effective constraint in the DC-OPF. In other words, $C_i(g_i)=g_i^2/2$, $i=1,...,n$, and $\underline{g}_i=-\infty,\bar{g}_i=\infty$, $i=1,...,n$. The line limits are set as a fraction of the absolute power flow in a setting where also the line limits pose no effective constraint. More specifically, in the absence of any generator and transmission line limits, it is easy to see that the solution of the DC-OPF is $\bg^{(\text{planning})}=\frac{1}{n}\sum_{i=1}^n X_i \be$. The associated flow vector is given by
$\bfl^{(\text{planning})}=\bV (\bg^{(\text{planning})}-\bX)=-\bV\bX$, where we used that $\bV \bg^*=\bzeros$ (Lemma~\ref{lm:rk}). For a safety loading factor $\lambda \in [0,1]$, referred as loading factor in the rest, the \textit{operational} line limits are set as
\begin{equation}
\label{linelimits}
\bfbar_j = \lambda  \left\vert(\bV \bX \right)_j\vert, \hspace{1cm} j=1,...,m.
\end{equation}

In the operational problem, we solve the DC-OPF for an increasing quadratic cost function and line limits as in~\eqref{linelimits} to obtain the generation vector $\bg$. That is, we solve
\begin{alignat}{6}
& \underset{g\in\R^n}{\min} & & \sum_{i=1}^n g_i^2/2 \\
& \text{s.t.}        & & \sum_{i=1}^n g_i=\sum_{i=1}^n X_i, &     \\
&					  &	 & \bV\bX - \lambda \left\vert\bV \bX \right\vert \leq \bV \bg \leq  & \bV\bX + \lambda \left\vert\bV \bX \right\vert,
\end{alignat}
where $|\bV\bX|$ denotes the vector with elements $(|\bV\bX|)_j = |(\bV\bX)_j|$, $j=1,...,m$.

Finally, in the emergency problem, we focus on the failure process after an initial disturbance. We assume that the initial failure is caused by a single line failure, chosen uniformly at random over all lines. We point out that our framework can be extended to multiple initial line failures, or adapted to deal with generator failures. The initial failure may cause a cascading effect that leads to multiple line failures that disintegrate the network. A consecutive line failure occurs whenever there is at least one line such that its \textit{emergency} line limit is exceeded. That is, instead of considering the conservatively chosen operational line limits $\fbar_{\ell}$, we take the line limits to be $F_{\ell}=\lambda^*\fbar_{\ell}$ for some constant $\lambda^*>1$. A canonical choice is $\lambda^*=1/\lambda$. We assume that line failures occur subsequently, and occurs at the line where its relative exceedance is largest.

Whenever line failures cause the network to disconnect in multiple islands, we assume that the energy balance is restored by proportionally lowering either generation or demand at all nodes. Naturally, this alters the line power flows. More specifically, before the initial disturbance occurs, the network flows are given by $\bV(\bg-\bX)$, where $\bg$ is the solution of the DC-OPF in the operational problem. After any line failure, we check whether this causes the network to disconnect, and if so, we proportionally lower the generation in one component and the demand in the other component such that demand and generation are balanced in the two disconnected components. The network flows are updated according to the laws of physics in every component. That is, the removal of one or more lines yields a modified matrix $\widetilde{\bV}$ (see Section~\ref{ss:PowerFlowRedistribution}), and possibly modified generation $\tilde{\bg}$ and demand $\tilde{\bd}$. The line flows are given by $\widetilde{\bV}(\tilde{\bg}-\tilde{\bd})$. This cascading failure process continues until the line limits $F_{\ell}$ of all surviving lines are sufficient to carry the power flows.

This iterative process leads to a network having disconnected sets. We make the convention that $A_1$ is the set of nodes that contains the city with the largest demand after the cascade has taken place. We point out that the set $A_1$ is random, and in particular, $A_1= \{1,...,n\}$ if the cascade stops without causing network disconnections.

Whenever the network disintegrates in multiple components, we alter the generation and demand to restore the power balance in every component.
We approximate the total of load that is shed or equivalently, the number of customers affected by the blackout, by the mismatch between generation and demand in the component containing the city with highest power demand, defined as
\begin{equation}\label{eq:S_def}
S = \left\vert \sum_{i\in A_1} (X_i-g_i) \right\vert.
\end{equation}
Due to properties of the Pareto distribution, this turns out to be a good approximation as it yields exactly the same limiting behavior. We study this notion in more detail in the next sections.

\subsection{Principle of a single city with large demand}
A vital property in our framework is that the only likely way to have a large blackout is when there is a single city that has a large power demand. To formalize this notion, write $d_1=\max\{d_1,...,d_n\}$ with $d_i, i=1,...,n$ independent and identically Pareto distributed power demands. Note that for every $\epsilon>0$,

\begin{align*}
&\Prob\left( S > x \right) = \\&\Prob\left( S > x ; \sum_{i=2}^n d_i < \epsilon d_1 \right) + \Prob\left( S > x ; \sum_{i=2}^n d_i \geq \epsilon d_1 \right).
\end{align*}

It turns out that we can show that (in certain settings) the first term on the right-hand side has a Pareto tail, and the second term is negligible. More specifically, the following result can be shown.

\begin{lemma}\label{lem:SingleLargeDemand}
Suppose $d_i$, $i=1,...,n$ are independent and identically Pareto distributed with tail exponent $\alpha>0$, and write $d_1=\max\{d_1,...,d_n\}$. For every $\epsilon >0$, as $x \rightarrow \infty$,
\begin{align}
\Prob\left( S > x ; \sum_{i=2}^n d_i \geq \epsilon d_1 \right) = O\left(x^{-2\alpha}\right).
\end{align}
\end{lemma}

\begin{proof}
We observe that the total mismatch can never exceed the sum of all demands, and hence
\begin{align*}
S \leq \sum_{i=1}^n d_i \leq n d_1.
\end{align*}
Therefore,
\begin{align*}
&\Prob\left( S > x ; d_i > \sum_{i=2}^n d_i \geq \epsilon d_1 \right)\\ \leq &\Prob\left( d_1 > \frac{x}{n} ; d_i > \epsilon \frac{d_1}{n} \textrm{ for some } i=2,...,n  \right) \\
\leq  &\Prob\left( d_i > \epsilon \frac{x}{n^2} \textrm{ for some } i=2,...,n \right)
\end{align*}
Write $I(y)=\vert \{i : d_i >y \}\vert$. Since for every $\eta >0$,
\begin{align*}
\Prob\left( I(\eta x) \geq 2 \right) = O\left(x^{-2\alpha}\right)
\end{align*}
as  $x \rightarrow \infty$, the result follows.
\end{proof}

In other words, Lemma~\ref{lem:SingleLargeDemand} implies that if for some $\epsilon>0$ sufficiently small,
\begin{align*}
\Prob\left( S > x ; \sum_{i=2}^n d_i < \epsilon d_1 \right) \sim C x^{-\alpha}
\end{align*}
holds for some constant $C \in (0,\infty)$, then the only likely way to have a large blackout is when there is a single city that has a large demand.

\subsection{Closed-form solution for the operational OPF in the case $\bd=\be_1$.}\label{ss:preliminary}
Note that without loss of generality, we can always normalize our framework by dividing all parameters (e.g. generation, line limits, etc.) by the sum of all power demands. This yields an equivalent setting where the total power demand equals one. In view of Lemma~\ref{lem:SingleLargeDemand}, it is sensible to consider the special case where $\bd=\be_1$. That is, node~$1$, henceforth referred to as the sink node, has unit demand, while all other nodes have zero demand. For this special case, a closed-form solution exists for the generation vector in the operational OPF.

First, we consider the planning problem. As stated in the model description, in the absence of any generator and transmission line limits, the solution of the planning OPF is $\bg^*=\frac{1}{n}\be$, with associated flow vector $\bfl^* =\bV (\bg^*-\be_1)=-\bV\be_1$, where we used that $\bV \bg^*=\bzeros$ (Lemma~\ref{lm:rk}). Therefore, the operational problem~\eqref{eq:objective}-\eqref{eq:lines} reduces to
\begin{alignat}{6}\label{eq:Plambda}
& \underset{g\in\R^n}{\min} & & \sum_{i=1}^n g_i^2/2 \\
& \text{s.t.}        & & \be^{\top}\bg=\be^{\top}\be_1=1, &    \\
&					  &	 -&\lambda |\bV \be_1| \le \bV(\bg-\bd)\le \lambda |\bV\be_1|, &
\end{alignat}
which we will denote by $P(\lambda)$. Lemma~\ref{lm:g_prelimit} shows that the solution of $P(\lambda)$ is of closed form.

\begin{lemma}\label{lm:g_prelimit}
Let $\lambda\in(0,1)$. Let $G$ be assigned the orientation such that $\bV \be_1\ge\bzeros$. Then, the solution of $P(\lambda)$ is given by
\[\bg(\lambda) =  \lambda \frac{1}{n}\be + (1-\lambda) \be_1,\]
i.e. $g_1(\lambda) = 1 - \lambda \frac{n - 1}{n}$ and
$g_i(\lambda) = \lambda \frac{1}{n}$ for all $i = 2, \dots, n$.
The corresponding line flows are at capacity and are given by $\bfl(\lambda)=-\lambda \bV\be_1$.
\end{lemma}

\begin{proof}
First, we note that the selected orientation on $G$ implies that the set of edges incident to node $1$ is $\LL_1=\{(1,j)\,|\, j \text{ is adjacent to }1\}$ (i.e. the edges in $\LL_1$ \textit{exit} node $1$), or, in terms of the edge-node incidence matrix $\bC$, that $C_{\ell,1}=1=-C_{\ell,j}$ for all $\ell=(1,j)\in\LL_1$. This is proved in Lemma~\ref{lm:C,V} in Section~\ref{ss:properties_PTDF}.

Due to the chosen orientation, $\bfbar(\lambda)= \lambda\,|\bV\be_1|=\lambda\bV\be_1
$ and the line limit constraints in $P(\lambda)$ can be rewritten as
\[
(1-\lambda) \bV\be_1\le \bV \bg \le (1+\lambda) \bV\be_1.
\]
The problem $P(\lambda)$ is a strictly convex optimization problem with linear equality and inequality constraints. Therefore, in order to show that $\bg(\lambda)$ is the unique optimal solution, it is sufficient to show that it satisfies the KKT conditions for $P(\lambda)$, which read
\begin{align}
&\bg+\bV^{\top}(\bmu^+-\bmu^-)+\gamma \be=0 \label{eq:kkt1},\\
&\bmu^+\ge\bzeros,\bmu^-\ge \bzeros ,\label{eq:kkt2},\\
&\mu_l^+ (\bV\bg-(1+\lambda)\bV\be_1)_\ell=0 \,\forall \ell\in \LL, \label{eq:kkt3a}\\
&\mu_l^- (-\bV\bg+(1-\lambda)\bV\be_1)_\ell=0 \,\forall \ell\in \LL,\label{eq:kkt3b}\\
&\be^{\top}\bg=1,\label{eq:kkt4}\\
&(1-\lambda) \bV\be_1\le \bV \bg \le (1+\lambda) \bV\be_1,\label{eq:kkt5}
\end{align}
where $\gamma$ is the Lagrange multipliers for the equality constraint and $\bmu^+,\bmu^-\in\R^m$ are the Lagrange multipliers for the inequality constraints.

Since $\bV\bg(\lambda)=(1-\lambda) \bV\be_1$ and $\be^{\top} \bg(\lambda)=\be^{\top}\be_1=1$, the candidate solution $\bg(\lambda)$ clearly satisfies the feasibility conditions~\eqref{eq:kkt4},\eqref{eq:kkt5} and the complementary slackness condition~\eqref{eq:kkt3b}. Moreover, condition~\eqref{eq:kkt3a} is satisfied if we choose $\bmu^+=\bzeros$.

Using the facts that $\bV\bg(\lambda)=(1-\lambda)\bV\be_1$ and $\text{Ker}(\bV)=<\be>$, pre-multiplying equation~\eqref{eq:kkt1} by $\bV$ yields $(1-\lambda)\be_1+\bV^{\top}\bmu\in \text{Ker}(\bV)$. This is equivalent to
\[ (1-\lambda)\be_1+\bV^{\top}\bmu =\frac{(1-\lambda)}{n}\be,
\]
where in the last equality we used again the property that $\bV\be=\bzeros$.
To conclude the proof, it remains to be shown that that there exist a nonnegative solution $\bmu^-\ge\bzeros$ of the matrix equation
\begin{equation}\label{eq:kkt_lagr}
\bV^{\top}(-\bmu^-)=(1-\lambda)(\be/n-\be_1).
\end{equation}

We construct a non-negative solution $\bmu^-$ as follows:
\[\mu_{\ell}^-:=(1-\lambda)\be_{\LL_1}=
\begin{cases}
(1-\lambda)\quad &l\in\LL_1\\
0\quad &l\notin\LL_1,
\end{cases},\]
where $\be_{\LL_1}$ is a $m$- dimensional vector containing ones in positions given by $\LL_1$, and $0$ elsewhere.
Invoking Lemma~\ref{lm:C,V} we see that $\bC\be_1=\be_{\LL_1}$, yielding
 $\bmu^-=(1-\lambda) \bC\be_1$.
 Using the definition of $\bV=\bC\bL^+,\bL=\bC^{\top}\bC$, and the property $\bL^+\bL=(\bI-\bJ/n)$ (see~\cite{VanMieghem2017}), we observe that Eq.~\eqref{eq:kkt_lagr} is indeed satisfied:
\begin{align*}\bV^{\top}\bmu^-&=-(1-\lambda)\bV^{\top}\bmu^-=-(1-\lambda)\bV^{\top}\bC\be_1\\&=-(1-\lambda)(\bL^+\bL)\be_1=(1-\lambda)(\be/n-\be_1).
\end{align*}
 Setting $\gamma=-1/n$ completes the proof.
\end{proof}

Finally, we solve the emergency problem. Observe that whenever there is a network disconnection, the component that does not contain node~$1$ has no power demand, and hence the generation at every node in that component is reduced to zero. Evidently, no consecutive failures occur in this component. On the other hand, the demand at node~1 is reduced by the number of nodes that disconnect from this component times $\lambda/n$. Therefore, the total amount of load that is shed is exactly equal to the total amount of reduced power generation at node~1 (power imbalance), which is given by
\begin{align}
S= \sum_{i \not\in A_1} \frac{\lambda}{n} =\frac{\lambda(n-|A_1|)}{n}.
\end{align}

Naturally, the way the failure process cascades through the network after the initial disturbance is highly dependent on the network topology. The redistribution of power flow takes place as described in Section~\ref{ss:PowerFlowRedistribution}, and we stress that this is a deterministic process. In this special case, the only sources of randomness come from the choice of the initial line failure, and possibly the choice of subsequent line failure whenever the redistribution of power flow causes the relative exceedance to be the same at multiple lines. Therefore, given a network topology and the line that initially fails, we can determine exactly how the failure process propagates through the network.

It may be apparent from Lemma~\ref{lem:SingleLargeDemand} that this special case where $\bd = \be_1$ describes some form of limiting behavior. That is, as Lemma~\ref{lem:SingleLargeDemand} holds for every $\epsilon>0$, we observe that the normalized demand vector $\bd$ converges to the unit vector $\be_1$ as $\epsilon \downarrow 0$. Next, we show that for almost all values of $\lambda$, for all demand vectors $\bd$ for which $\bd \rightarrow \be_1$ as $\epsilon \downarrow 0$, the order at which line failures occur converges to the sequence of line failures as if the demand vector would have been $\bd=\be_1$.

\subsection{Convergence of cascade sequence}

The operational OPF
\begin{alignat}{6}
& \underset{g\in\R^n}{\min} & & \frac{1}{2}\bg^\top\bg \label{eq:Plambda1}\\
& \text{s.t.}        & & \be^{\top}\bg=\be^{\top}\bd, &    \label{eq:Plambda2}\\
&					  &	 | & \bV(\bg-\bd)|\le \lambda \left\vert\bV\bd
\right\vert, & \label{eq:Plambda3}
\end{alignat}
is a strictly convex optimization problem, and since $\bg=\lambda \bar{\bd}\be + (1-\lambda) \bd$ is a feasible point, the feasible set of this optimization problem is nonempty. Therefore, for each demand vector $\bd$, there exists a unique optimal solution $\bg^*(\bd)$.

If we view $\bd$ as a \textit{parameter} of the problem, then~\eqref{eq:Plambda1}-\eqref{eq:Plambda3} is an instance of a \textit{multi-parametric quadratic programming} (mp-QP) problem with a strictly convex objective function, for which it is known that the optimal solution $\bg^*(\bd)$ is a continuous function of the parameter vector $\bd$ (Theorem $1$,~\cite{Tondel2003}). This continuity property will be used extensively in the rest of this section.

We assume in our framework that line failures occur subsequently, i.e. a next line failure occurs at the line where the line limit is relatively most exceeded. Recall that $F_j$ denotes the \textit{emergency} line limit of line $j \in \mathcal{L}$, and is given by (taking $\lambda^*=\frac{1}{\lambda}$)
\begin{align*}
F_j=\lambda^*\lambda|(\bV\bd)_j|=|(\bV\bd)_j|.
\end{align*}
We write $f^{(m)}_j$ as the flow on line $j$ after the failure of the first $m-1$ lines and after the load/generation shedding took place, where we use the convention that $f_j^{(1)}$ denotes the flow on line $j$ when no initial disturbance has occurred yet, and $f_j^{(m)}=0$ if line $j$ has already failed before the $m$-th step of the cascading failure process. The cascade is initiated by the random failure of line $\ell=\ell^{(1)}$. The $m$-th line to fail, for $m\ge 2$, is given by
\begin{align}\label{eq:next_failure}
\ell^{(m)}=\argmax_{j\in \A^{(m)}}
\Bigl\{\frac{|f_j^{(m)}|-F_j}{F_j}\Bigr\}=\argmax_{j\in \A^{(m)}}\Bigl\{\frac{|f_j^{(m)}|}{F_j}\Bigr\},
\end{align}
where $\A^{(m)}=\{j \,:\, |f_j^{(m)}|\ge F_j\}$ is the set of lines that exceed the limit.

\begin{remark}
Note that the line limits and line flows depend on $\bd$ and $\lambda$ through the operational OPF, so that the sequence of subsequent failure depends on $\bd,\lambda$, and on the initial failure $\ell=\ell^{(1)}$. That is,
\begin{align*}
&F_j=F_j(\bd),f^{(m)}_j=f^{(m)}_j(\bd,\lambda),\A^{(m)}=\A^{(m)}(\ell,\lambda,\bd),\\&
\ell^{(m)}=\ell^{(m)}(\ell,\lambda,\bd).\end{align*}
For the sake of exposition, we do not write the dependency on $\bd$, $\lambda$ and $\ell$.
\end{remark}

Let $\C=\{\ell^{(1)},\ldots,\ell^{(T)}\}$ be a cascade sequence, where $\ell^{(T)}$ is the last failure before the cascade stops. Such a sequence is uniquely determined by the first failure $\ell^{(1)}$ and by the demand vector $\bd$ and by $\lambda$, i.e. $\C=\C(\bd,\lambda,\ell)$. In view of Lemma~\ref{lem:SingleLargeDemand} and the normalization property, the goal of this section is to show that if $\bd\to\be_1$, then the cascade sequence does not depend on $\bd$ anymore, i.e.
\begin{align*}
\C(\bd,\lambda,\ell)=\C(\be_1,\lambda,\ell)\qquad\text{if }\bd \rightarrow \be_1.
\end{align*}

We observe that if $|\A^{(m)}|=0$, no more line failures occur. Technically, it is also possible that $|\A^{(m)}|>1$ and hence the subsequent line failure next needs to be chosen out of a set of multiple lines. We exclude the cases that do not yield unique maximizers from our framework.
\begin{assumption}\label{ass:ass2}
For all lines $j$, the ratios between redistributed flows and line limits
\[\frac{|f_j^{(m)}(\be_1)|}{F_j(\be_1)}=\frac{\lambda |(\bV^{(m)}\be_1)_j|}{|(\bV\be_1)_j|}\] are all different for all $m\geq 2$, where $\bV^{(m)}$ denotes the PTDF matrix for the remaining network after $m-1$ failures have taken place. This assumption is needed to ensure the uniqueness of the maximizer in~\eqref{eq:next_failure}.
\end{assumption}
This assumption ensures that whenever the first line failure $\ell$ and the parameter $\lambda$ is known, the cascade sequence is unique and deterministic for demand vector $d=\be_1$. This assumption is for technical convenience, and we stress that our results hold more generally. In particular, this assumption rules out certain network topologies with some form of symmetry, but we can slightly adapt the framework to deal with these cases as well.

 That is, suppose that $|\A^{(m)}|>1$ for some $m \in \mathbb{N}$ and the set $\A^{(m)}$ consists only of lines that are indistinguishable from one another (lines that are `symmetric'). Since nodal demands are independent and identically distributed, this implies that each of these lines has an equal probability of being the line that fails next. By the symmetry of the network topology, regardless of which line is chosen to fail next, the resulting networks after the cascade are indistinguishable. We illustrate this notion for the 6-node example in the next section.

To analyze the power imbalance in this framework, we need to introduce some notation as well as formally define the shedding rule and the redistribution of power flows.

\begin{definition}[Uniform shedding rule]\label{def:shedding}
Let $\bg^{(1)}=\bg^*,\bd^{(1)}=\bd$ be the initial generation and demand vectors.
Assume that the removal of lines $\ell^{(1)},\ldots,\ell^{(m)}$, $m\ge 1$, disconnects the network in components $\G_i^{(m)}=(\N_i^{(m)},\LL_i^{(m)})$, $i=1,\ldots,h_m$.
Define the power imbalance in component $\G_i^{(m)}$ as
 \begin{align*}
Y_{\G_i^{(m)}} = \sum_{k \in \N_i^{(m)}} (g^{(m)}_k-d_k^{(m)}).
\end{align*}
In order to re-achieve power balance, generation and demand in each component are modified iteratively according to the following uniform shedding rule, for $k\in \N_i^{(m)}$:
\begin{align*}
 & d_k^{(m+1)} = \begin{cases}
 \left( 1- \frac{Y_{\G_i^{(m)}}}{\sum_{l \in \N_i^{(m)}} d_l^{(m)}}\right)d_k^{(m)}\quad &\text{if } Y_{\G_i^{(m)}}<0
 \\
  d_k^{(m)} & \text{if } Y_{\G_i^{(m)}}\ge 0
 \end{cases},
 \end{align*}

  \begin{align*}
 g_k^{(m+1)} = \begin{cases}
 g_k^{(m)} \quad &\text{if } Y_{\G_i^{(m)}}<0\\
\left( 1- \frac{Y_{\G_i^{(m)}}}{\sum_{l \in \N_i^{(m)}} g_l^{(m)}}\right)g^{(m)}_k & \text{if } Y_{\G_i^{(m)}}\ge 0\quad 
 \end{cases}
 \end{align*}
\end{definition}

\begin{definition}[Power flow redistribution]\label{def:flows}
Assume that the removal of lines $\ell^{(1)},\ldots,\ell^{(m)}$, $m\ge 1$, disconnects the network in components $\G_i^{(m)}=(\N_i^{(m)},\LL_i^{(m)})$, $i=1,\ldots,h_m$.
Then, the line flows in component $\G_i^{(m)}$ are given by
\begin{align*}
f_{\LL_i}^{(m+1)}=\bV^{(m+1,\G_i)}(\bg_{\N_i}^{(m+1)}-\bd_{\N_i}^{(m+1)}),
\end{align*}
where $\bV^{(m+1,G_i)}$ is the PTDF matrix for the subgraph $\G_i^{(m)}$, and
$g_{\N_i}^{(m+1)},d_{\N_i}^{(m+1)}$ are defined as in Definition~\ref{def:shedding}.
\end{definition}

A second assumption we require to show the convergence of the cascade sequence involves the following.
\begin{assumption}\label{ass:ass1}
\vspace{0.1cm}
 For all lines $j$ and $m\ge 2$,
\[|f_j^{(m)}(\be_1)|-F_j(\be_1)\neq 0.\]
 That is,
for $\bd=\be_1$ it is not possible for a line flow $|f_j^{(m)}|$ to be exactly equal to its limit. In terms of PTDF matrices and $\lambda$, this assumption reads
\begin{align*}
\lambda |(\bV^{(m,\G_i)}\be_1)_j|\neq |(\bV\be_1)_j|, \hspace{1cm} m\geq 2.
\end{align*}
This assumption means that we exclude finitely many $\lambda$-s from our analysis, which correspond to phase-transitions.
\end{assumption}

Assumption~\ref{ass:ass1} states that none of the line flows equal its emergency line limit in the cascade sequence if $d=\be_1$. In order to prove the convergence of the cascade sequence, we also need a continuity property of the line flows at every stage with respect to the demand vector.

\begin{lemma}[Continuity of $f_j^{(m)}$ with respect to $\bd$]\label{lm:continuity}
At each stage $m$ of the cascade, the redistributed power flows $f_j^{(m)}$ are continuous in the initial demand vector $\bd$ for all $j=1,...,m$.
\end{lemma}

\begin{proof}
Assume that the removal of lines $\ell^{(1)},\ldots,\ell^{(m)}$, $m\ge 1$, disconnects the network in components $\G_i^{(m)}=(\N_i^{(m)},\LL_i^{(m)})$, $i=1,\ldots,h_m$. According to Definition~\ref{def:flows},
\[f_{\LL_i}^{(m+1)}=\bV^{(m+1,\G_i)}(\bg_{\N_i}^{(m+1)}-\bd_{\N_i}^{(m+1)}),\]
for each connected component $\G_i^{(m)}$, so $\bfl^{(m+1)}$ is continuous in $\bg^{(m+1)},\bd^{(m+1)}$. Moreover, according to Definition~\ref{def:shedding}, $\bg^{(m+1)},\bd^{(m+1)}$ are continuous functions of $\bg^{(m)},\bd^{(m)}$. By unfolding the recursion, and using that $\bg^*(\bd)$ is continuous in $\bd$, we see that $\bfl^{(m+1)}$ is continuous in $\bd$.
\end{proof}

Finally, we can show the main result of this section.
\begin{proposition}\label{prop:MainResultConvergence}
Assume that Assumptions~\ref{ass:ass2} and~\ref{ass:ass1} hold, and let $\C(\bd,\lambda,\ell)=\{\ell^{(1)},\ldots,\ell^{(T)}\}$ be a cascade sequence initiated by $\ell=\ell^{(1)}$. Then, there exists a $\eps>0$ such that
\[d_1=1, d_j<\eps,\, \forall j\ge 2\implies \C(\bd,\lambda,\ell)=\C(\be_1,\lambda,\ell).\]

\end{proposition}
\begin{proof}

Let $\ell^{(1)}$ be the first failure, and consider
\[\ell^{(2)}=\argmax_{j\in \A^{(2)}}\Bigl\{\frac{|f_j^{(2)}|}{F_j}\Bigr\},\]
where $\A^{(2)}=\A^{(2)}(\bd,\lambda)=\{j \,:\, |f_j^{(2)}|\ge F_j\}$.
Lemma~\ref{lm:continuity} shows that $f_j^{(2)}(\bd)
\to f_j^{(2)}(\be_1)$ as $\bd\to\be_1$,
so by continuity and Assumption~\ref{ass:ass1}
\[|f_j^{(2)}(\bd)|-F_j(\bd)
\to |f_j^{(2)}(\be_1)|-F_j(\be_1)\neq 0. \]  Consequently, there exist $\eps>0$ such that, if $d_j<\eps$ for all $j\ge 2$, then
\[
|f_j^{(2)}(\bd)|>F_j(\bd)\,\iff\, |f_j^{(2)}(\be_1)|>F_j(\be_1).
\]
In other words, a line limits is exceeded for $\bd=\be_1$ (which, due to our assumption, implies that it is strictly exceeded) if and only if it is also (strictly) exceeded when $\bd$ is close enough to $\be_1$, implying that $\A^{(2)}(\bd)=\A^{(2)}(\be_1)$ .

Moreover, there exists a $\eps^{(1)}\le\eps$ such that, if $d_k<\eps$ for $k\ge 2$, then
\begin{align*}\max_{j\in\A^{(2)}(\bd_1,\lambda)}
\frac{|f_j^{(2)}(\bd,\lambda)|}{F_j(\bd)}&=
\max_{j\in\A^{(2)}(\be_1,\lambda)}
\frac{|f_j^{(2)}(\bd,\lambda)|}{F_j(\bd)}\\&=
\max_{j\in\A^{(2)}(\be_1,\lambda)}
\frac{|f_j^{(2)}(\be_1,\lambda)|}{F_j(\be_1)},
\end{align*}
where in the second equality we used that
$\A^{(2)}(\bd)=\A^{(2)}(\be_1)$, and in the third equality we used again continuity.
Finally, Assumption~\ref{ass:ass2} allows us to conclude that the max is unique and that the (unique) second failure
$\ell^{(2)}(\bd,\lambda)=\ell^{(2)}(\be_1,\lambda)$ does not depend on $\bd$ if $d_k<\eps^{(1)}$, $k\ge2$.

As Lemma~\ref{lm:continuity} holds for every stage of the cascade, we can repeat the steps above to construct a sequence $\eps^{(T)}\le \ldots,\eps^{(2)}\le\eps^{(1)}$ such that the cascade sequence $\C$ is well defined and does not depend on $\bd$ if $d_j<\eps^{(T)}$ for all $j\ge2$.
\end{proof}

\begin{example}
To illustrate how one can easily derive the phase-transition values, we consider the 4-node cycle topology. With the standard clock-wise orientation,  we have
\[\bV^{(\text{clock})}=\frac{1}{8}
\begin{bmatrix}
3 & -3 & -1 & 1\\
1 & 3 & -3 & -1\\
-1 & 1 & 3 & -3\\
-3 & -1 & 1 & 3.
\end{bmatrix}.\]
For $d_1 \gg d_2,d_3,d_4$, we can change the orientation such that $\bV \bd\ge \bzeros$, which is given by the edgelist $\{(1,2),(2,3),(4,3),(1,4)\}$. Then the matrix $\bV$ reads
\[
\bV=\text{diag}(\bs)\bV^{(\text{clock})}=\frac{1}{8}\begin{bmatrix}
3 & -3 & -1 & 1\\
1 & 3 & -3 & -1\\
1 & -1 & -3 & 3\\
3 & 1 & -1 & -3,
\end{bmatrix},\]
where $\bs=\text{sign} (\bV^{(\text{clock})}\be_1)=[1,1,-1,-1]$.
In this case,
\[\phi_{k,\ell}=-s_k=
\begin{cases}
-1 \qquad &\text{if } k\in\{(1,2),(2,3)\}\\
1 \qquad &\text{if }  k\in\{(4,3),(1,4)\}
\end{cases}.\]
Assume that the first failure is $\ell=(1,2)$, so that the power flow redistribution is
\[\lambda\left((\bV\be_1)_k+\phi_{k,\ell}(\bV\be_1)_\ell\right)=\lambda \cdot
\begin{cases}
1/4 \qquad &\text{if } k=(2,3)\\
-1/2 \qquad &\text{if }  k=(4,3)\\
-3/4 \qquad &\text{if }  k=(1,4)\}
\end{cases}.\]

\noindent
Then, the critical values of $\lambda$ are given by
\begin{equation}\label{eq:lambda1}
\lambda= \pm\frac{(\bV\be_1)_k}{(\bV\be_1)_k+\phi_{k,\ell}(\bV\be_1)_\ell},\,k\neq \ell,
\end{equation}
%
and we find that they are $\lambda=\frac{1}{4},\frac{1}{2}$.

Moreover, if $\lambda<\frac{1}{4}$ then the cascade stops immediately after the failure of the first line. If $\frac{1}{4} < \lambda <\frac{1}{2}$, then line $(4,3)$ fails afterwards
 and if $\lambda>\frac{1}{2}$ lines $(2,3),(4,3),(1,4)$ fails afterwards.
 Therefore, $\lambda=\frac{1}{4},\frac{1}{2}$ can be seen as phase-transition points.
\end{example}

\subsection{Asymptotic behavior of power imbalance}
In the previous sections, we showed that the only likely way to have a large blackout is when there is a single city that has a significantly larger demand than all other cities. Under certain assumptions, given the position of this city (i.e. labeling this as city~$1$) and the first line failure, the cascade sequence is deterministic and the same to the one as if the demand vector would have been $d=\be_1$. We exploit these properties to derive the tail behavior of $S$, or equivalently, the amount of load that is shed/the number of affected customers.

We point out that the demands are independent and identically distributed, so the probability that a city has the largest demand equals $1/n$. To obtain the tail behavior of $S$, we need that Assumptions~\ref{ass:ass2} and~\ref{ass:ass1} to hold regardless of which city has the largest power demand.

\begin{assumption}\label{ass:GeneralAss}
Assumptions~\ref{ass:ass2} and~\ref{ass:ass1} hold for any relabeling of the vertices.
\end{assumption}

Note that since the number of cities $n$ is finite, and inherently also the number of the possible lines where the first failure occurs, Assumption~\ref{ass:GeneralAss} excludes only a finite number of possible values of $\lambda$ from our framework. The main theorem follows.

\begin{theorem}
Suppose there is a fixed topology $\G=(\N,\LL)$ and a fixed $\lambda \in (0,1)$, for which Assumption~\ref{ass:GeneralAss} holds. Write $Z(i,\ell)$, $i=1,...,n$, $\ell=1,...,m$ as the number of cities that are not in the same component as city $i$ after the cascade under demand vector $d=\be_i$ and first line failure $\ell$. If $Z(i,\ell)=0$ for all $i=1,...,n$ and $\ell=1,...,m$, then as $x \rightarrow \infty$,
\begin{align}\label{eq:NoDisconnectionsAsympt}
\Prob(S>x) = O(x^{-2\alpha}).
\end{align}
Otherwise, as $x \rightarrow \infty$, there exists a $C \in (0,\infty)$ such that
\begin{align}\label{eq:WithDisconnectionsAsympt}
\Prob(S>x) \sim C x^{-\alpha}.
\end{align}
\end{theorem}

\begin{proof}
First, since the demands are independent and identically distributed, we observe that each city has an equal probability of being the city with the largest demand. That is, if $B$ denotes the city that has the largest demand, then
\begin{align*}
\Prob(B=i)=1/n, \hspace{1cm} i=1,...,n.
\end{align*}
By the law of total probability,
\begin{align*}
\Prob\left( S > x \right) = \sum_{i=1}^n \frac{1}{n} \Prob\left( S > x \,\big|\, B=i \right).
\end{align*}
Fix some $\epsilon>0$ (sufficiently small), and note that for all $i=1,...n$,
\begin{align*}
\Prob\left( S > x \,\big|\, B=i \right) &= \Prob\left( S > x ; \sum_{j\neq i}^n d_j < \epsilon d_i \,\big|\, B=i \right) \\&+ \Prob\left( S > x ; \sum_{j \neq i}^n d_j \geq \epsilon d_i \,\big|\, B=i \right).
\end{align*}
Due to Lemma~\ref{lem:SingleLargeDemand}, we observe that the second is of order $O(x^{-2\alpha})$ for all $i=1,...,n$, and hence
\begin{align*}
\sum_{i=1}^n \frac{1}{n}  \Prob\left( S > x ; \sum_{j \neq i}^n d_j \geq \epsilon d_i \,\big|\, B=i \right) = O(x^{-2\alpha}).
\end{align*}
For the first term, note that Assumption~\ref{ass:GeneralAss} ensures that $Z(i,\ell)$ is well-defined for all $i=1,...,n$ and $\ell=1,...,m$. Since we choose our first failure uniformly at random among all lines, we observe that by law of total probability, for all $i=1,...,n$,
\begin{align*}
&\Prob\left( S > x ; \sum_{j\neq i}^n d_j < \epsilon d_i \,\big|\, B=i \right) \\=& \sum_{l=1}^{m} \frac{1}{m} \Prob\left( S > x ; \sum_{j\neq i}^n d_j < \epsilon d_i \,\big|\, \ell^{(1)}=l, B=i \right).
\end{align*}

In case that $Z(i,\ell)=0$ for all $i=1,...,n$ and $\ell=1,...,m$, it follows from Proposition~\ref{prop:MainResultConvergence} that for all $\epsilon >0$ sufficiently small, the cascade sequence causes no disconnections for every city with largest demand and first line failure $\ell^{(1)}$. That is, for all $i=1,...,n$, $l=1,...,m$, $x>0$ and $\epsilon>0$ sufficiently small,
\begin{align*}
\Prob\left( S > x ; \sum_{j\neq i}^n d_j < \epsilon d_i \,\big|\, \ell^{(1)}=l, B=i \right) = 0.
\end{align*}
Therefore, if $Z(i,\ell)=0$ for all $i=1,...,n$ and $\ell=1,...,m$, then for all $\epsilon>0$ sufficiently small,
\begin{align*}
\sum_{i=1}^n \frac{1}{n} \Prob\left( S > x ; \sum_{j \neq i}^n d_j < \epsilon d_i \,\big|\, B=i \right)  = 0,
\end{align*}
and we conclude that~\eqref{eq:NoDisconnectionsAsympt} holds.

Next, suppose that $Z(i,\ell)\neq 0$ for at least some $i\in \{1,...,n\}$ and $\ell\in \{1,...,m\}$. It follows from Proposition~\ref{prop:MainResultConvergence} that for all $i\in \{1,...,n\}$ and $\ell\in \{1,...,m\}$ for which $Z(i,\ell)\neq 0$, it holds for all $\epsilon >0$ sufficiently small that the cascade sequence is the same as the one when the demand vector would have been $\bd=\be_i$. In particular,  whenever $\sum_{j \neq i}^n d_j < \epsilon d_i$, it holds for all $\epsilon >0$ sufficiently small that the set $A_1$ is deterministic and is the same set of nodes as if demand would have been $\bd=\be_i$, and the number of cities disconnected from city $i$ equals $Z(i,\ell)$. Recall Lemma~\ref{lm:g_prelimit} and the property that the generator vector $\bg$ is a continuous function of $\bd$. Consequently, for all $i\in \{1,...,n\}$ and $\ell\in \{1,...,m\}$ for which holds that $Z(i,\ell)\neq 0$, and for all $\epsilon >0$ sufficiently small,
{\small
\begin{align*}
&\Prob \left( S > x , \sum_{j\neq i}^n d_j < \epsilon d_i \,\big|\, \ell^{(1)}=\ell, B=i \right)
\\=& \Prob\left( \sum_{i \not\in A_1} (g_i-d_i)  > x , \sum_{j\neq i}^n d_j < \epsilon d_i \,\big|\, \ell^{(1)}=\ell, B=i \right)\\
\leq& \Prob\left( Z(i,\ell)\left(\frac{\lambda}{n} +c_1(\epsilon) \right)d_i > x , \sum_{j\neq i}^n d_j < \epsilon d_i \,\big|\, \ell^{(1)}=\ell, B=i \right),
\end{align*}
}
where $c_1(\epsilon)$ is a strictly positive function with $c_1(\epsilon) \rightarrow 0$ as $\epsilon \downarrow 0$. For independent identically Pareto-distributed random variables $X_1,...,X_n$, it holds that as $x \rightarrow \infty$,
\begin{align*}
\Prob\left(\max\{X_1,...,X_n\} \geq x \right) \sim n \Prob(X_i >x) =  n K x^{-\alpha}.
\end{align*}
Therefore, for all $i\in \{1,...,n\}$ and $\ell\in \{1,...,m\}$ for which $Z(i,\ell)\neq 0$,
\begin{align*}
&\lim_{\epsilon \downarrow 0} \lim_{x \rightarrow \infty} x^{\alpha} \Prob\left( S > x , \sum_{j\neq i}^n d_j < \epsilon d_i \,\big|\, \ell^{(1)}=\ell, B=i \right) \\\leq &\lim_{\epsilon \downarrow 0} n \left(Z(i,\ell)\left(\frac{\lambda}{n} +c_1(\epsilon)\right) \right)^\alpha \\
=& nK \left(Z(i,\ell)\frac{\lambda}{n}  \right)^\alpha .
\end{align*}
Similarly, we can obtain the same lower bound, i.e.
\begin{align*}
&\lim_{\epsilon \downarrow 0} \lim_{x \rightarrow \infty} x^{\alpha} \Prob\left( S > x , \sum_{j\neq i}^n d_j < \epsilon d_i \,\big|\, \ell^{(1)}=\ell, B=i \right) \\\geq & nK \left(Z(i,\ell)\frac{\lambda}{n}  \right)^\alpha .
\end{align*}
We conclude that as $x \rightarrow \infty$,
\begin{align*}
\Prob\left( S > x \right) = \sum_{i=1}^n \sum_{l=1}^{m} \frac{K}{m}  \left(Z(i,\ell)\frac{\lambda}{n}  \right)^\alpha x^{-\alpha}.
\end{align*}
Note that term in front of $x^{-\alpha}$ is a double sum of finitely many terms, and hence we can also conclude that~\eqref{eq:WithDisconnectionsAsympt} holds.
\end{proof}

\section{Cascade analysis for 6-node topology}
To illustrate how to derive the asymptotic behavior of the amount of load that is shed using our framework, we consider a network topology that consists of six nodes and eight lines as illustrated in the main paper. It follows from our results that in order to understand the behavior for large blackouts, it suffices to consider the behavior under unit demand vectors. Due to the highly symmetric structure of the network topology in this example, there are only two relevant options for the position of the city with the highest demand, as is illustrated in Fig.~\ref{fig:LineLimitsExample}. The red node represents the city that has unit demand (largest), while the other nodes have zero demand. Note that case A and B occur with probability $1/3$ and $2/3$, respectively. In each case, one can solve the operational problem to determine the \textit{emergency} line limits, which are also depicted in Fig.~\ref{fig:LineLimitsExample}. We illustrate how the cascading failure processes evolve in these cases next.

\begin{figure}[h!]
\centering
  \begin{subfigure}[t]{0.45\columnwidth}
  \centering
          \includegraphics[width=\textwidth]{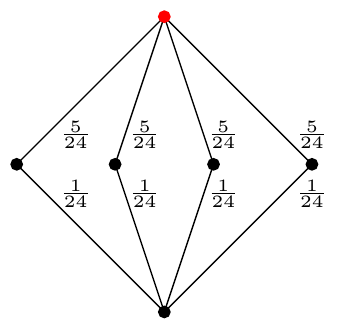}
          \caption{\small{Case A: the city with the largest demand is one of two cities that connect to four other cities.}}
\end{subfigure}\;\;\;
  \begin{subfigure}[t]{0.45\columnwidth}
  \centering
          \includegraphics[width=\textwidth]{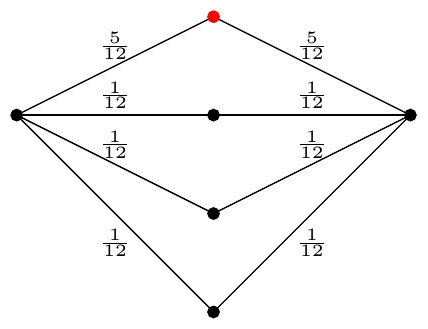}
        \caption{\small{Case B: the city with the largest demand is one of four cities that connect to two other cities.}}
      \label{fig:LineLimitsB}
      \end{subfigure}
      \caption{Line limits in the 6-node example.}
      \label{fig:LineLimitsExample}
      \end{figure}

\subsection{Case A}
The first line failure is chosen uniformly at random among all eight lines. Again, due to the symmetries of the network topology, we only need to consider two possibilities: when the first failure occurs at a top lines, or when it occurs at one of the bottom lines. Fig.~\ref{fig:CascadeExampleATop} illustrates the possible cascade when the initial failure is at a top line. The possible subsequent line failure occurs at the line for which the ratio of flow and line limit is largest, and therefore these values are depicted next to all (remaining) lines in Fig.~\ref{fig:CascadeExampleATop}. Only if the maximum ratio is strictly larger than one, another line fails. If the ratio is less than one, no consecutive line failure occurs, and if it equals one, then it corresponds to a phase-transition case.

\begin{figure}[h!]
  \begin{subfigure}[t]{0.45\columnwidth}
  \centering
          \includegraphics[width=\textwidth]{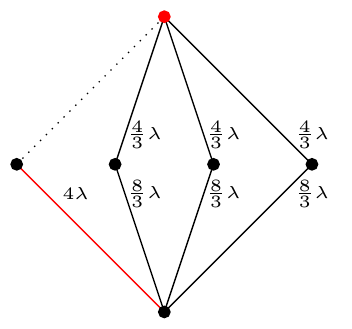}
          \caption{\small{Initial failure.}}
\end{subfigure}\hfill
  \begin{subfigure}[t]{0.45\columnwidth}
  \centering
          \includegraphics[width=\textwidth]{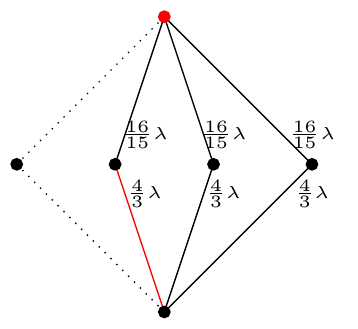}
        \caption{\small{Step 2.}}
      \end{subfigure}\hfill
       \begin{subfigure}[t]{0.45\columnwidth}
  \centering
          \includegraphics[width=\textwidth]{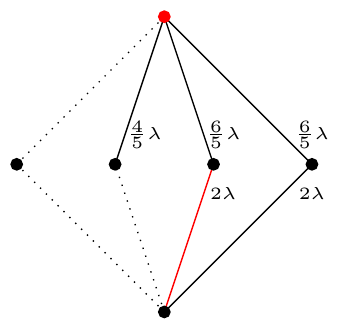}
        \caption{\small{Step 3.}}
      \end{subfigure}\hfill
             \begin{subfigure}[t]{0.45\columnwidth}
  \centering
          \includegraphics[width=\textwidth]{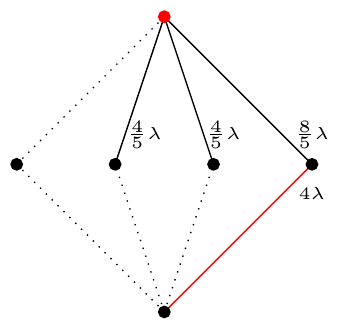}
        \caption{\small{Step 4.}}
      \end{subfigure}
             \begin{subfigure}[t]{0.45\columnwidth}
  \centering
          \includegraphics[width=\textwidth]{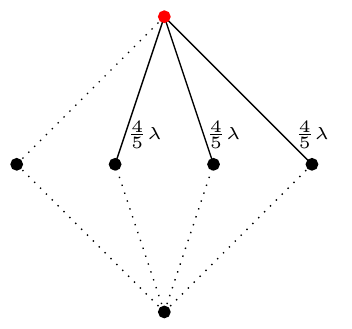}
        \caption{\small{Step 5.}}
      \end{subfigure}
      \caption{First line failure occurs at a top line.}
      \label{fig:CascadeExampleATop}
      \end{figure}

In Fig.~\ref{fig:CascadeExampleATop}, we observe that after the initial failure at a top line, the ratio of flow and line limit is highest at the corresponding bottom line. The ratio is $4\lambda$, and hence $\lambda=1/4$ is the first phase-transition value that we observe. If $\lambda<1/4$ the cascade ends immediately, otherwise this line fails. In step~2 (which is relevant for all value of $\lambda>1/4$), we observe that the bottom lines all have the same maximum ratio of $4/3\lambda$, which implies that Assumption~\ref{ass:ass2} is violated in this example. Yet, demands are independent and identically distributed, and the symmetric structure of this network topology ensures that each of the remaining bottom lines have equal probability to fail next. This explains why Assumption~\ref{ass:ass2} is too restrictive, and why our framework can deal with network topologies that have these types of symmetries as well. That is, regardless of the actual line that is chosen to fail next, the networks that appear in the next stages of the cascade are indistinguishable from one another. The second step also yields the second phase-transition value: if $1/4 < \lambda < 3/4$, then no consecutive failure occur, and if $\lambda >3/4$ another bottom line fails. In the latter case, also step~3 and step~4 are observed, where the network turns stable at step~5.

\begin{figure}[h!]
\centering
  \begin{subfigure}[t]{0.45\columnwidth}
  \centering
          \includegraphics[width=\textwidth]{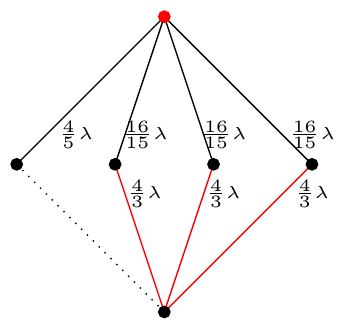}
          \caption{\small{First steps.}}
\end{subfigure}
  \begin{subfigure}[t]{0.45\columnwidth}
  \centering
          \includegraphics[width=\textwidth]{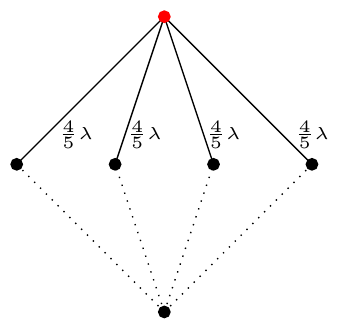}
        \caption{\small{Network after cascade.}}
      \end{subfigure}
      \caption{First line failure at a bottom line}
      \label{fig:CascadeExampleABottom}
      \end{figure}

\begin{figure}[h!]
\centering
  \begin{subfigure}[t]{0.45\columnwidth}
  \centering
          \includegraphics[width=\textwidth]{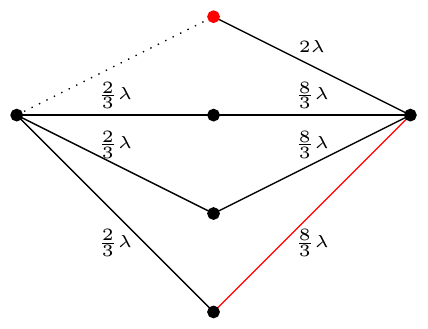}
          \caption{\small{Initial failure.}}
\end{subfigure}
  \begin{subfigure}[t]{0.45\columnwidth}
  \centering
          \includegraphics[width=\textwidth]{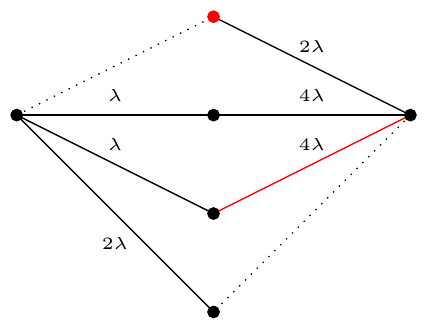}
        \caption{\small{Step 2.}}
      \end{subfigure} \\
        \begin{subfigure}[t]{0.45\columnwidth}
  \centering
          \includegraphics[width=\textwidth]{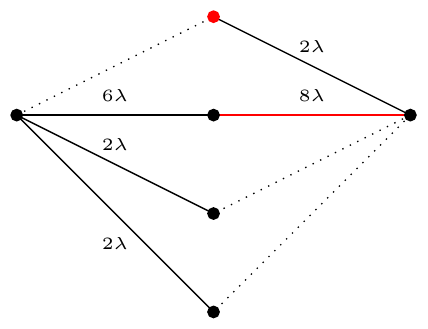}
          \caption{\small{Step 3.}}
\end{subfigure}
  \begin{subfigure}[t]{0.45\columnwidth}
  \centering
          \includegraphics[width=\textwidth]{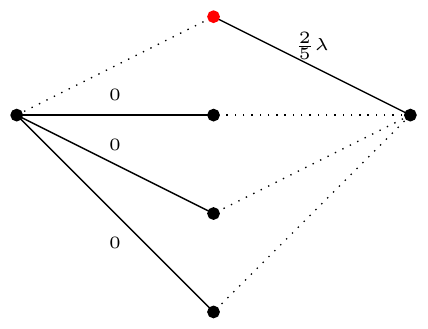}
        \caption{\small{Step 4.}}
      \end{subfigure} \\
      \caption{First line failure at a top line.}
      \label{fig:CascadeExampleBTop}
      \end{figure}
If the first line failure occurs at a bottom line, then the corresponding top line is stable, and a possible subsequent failure occurs at one of the three remaining bottom lines. Again, consecutive line failures occur when $\lambda>3/4$, and the cascade stops if $\lambda<3/4$. Using an analogue analysis as in the previous case, we would end up with a network where all bottom lines have failed, and all top lines are intact when $\lambda>3/4$. An illustration of this case is given in Fig.~\ref{fig:CascadeExampleABottom}.

\subsection{Case B} We can perform a similar analysis in this case. Again, due to the symmetries, there are only two truly different possibilities for the first line failure to occur: one of the two top lines, or one of the six other lines in Fig.~\ref{fig:LineLimitsB}.

In Fig.~\ref{fig:CascadeExampleBTop} we illustrate the possible cascades if the initial failure occurs at one of the two top lines. In this case, the cascade only continues if $\lambda >3/8$, and stops if $\lambda<3/8$. That is, we observe another phase-transition value, i.e.~$\lambda=3/8$. If $\lambda > 3/8$, then after the initial line failure three more line failure occurs, which after the cascade stops. In Fig.~\ref{fig:CascadeExampleBBottom} we illustrate the possible cascades if the initial failure occurs at one of the six bottom lines. Then, there is only a possible second line failure when $\lambda>1/2$, after which the cascade stops. If $\lambda<1/2$, the cascade stops after immediately after the initial line failure. Consequently, we obtain $\lambda=1/2$ as a fourth phase-transition value.

\begin{figure}[h!]
\centering
  \begin{subfigure}[t]{0.45\columnwidth}
  \centering
          \includegraphics[width=\textwidth]{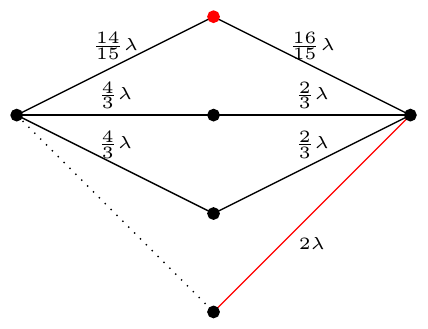}
          \caption{\small{Initial failure.}}
\end{subfigure}
  \begin{subfigure}[t]{0.45\columnwidth}
  \centering
          \includegraphics[width=\textwidth]{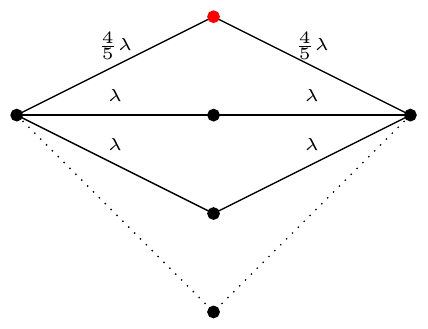}
        \caption{\small{Step 2.}}
      \end{subfigure}
      \caption{First line failure at a bottom line.}
      \label{fig:CascadeExampleBBottom}
      \end{figure}

  \subsection{Tail behavior of blackout size} To derive the tail behavior of the total amount of load shed, we need to determine the number of combinations that leads to $j$ cities disconnected from the city with largest demand $j=0,1,...,n-1$. That is, we count the number of tuples $(i,\ell)$ such that $Z(i,\ell)=j$, $j=0,1,...,n-1$, where $i$ denotes the city label and $\ell$ the first line failure. Since the network consists of six nodes and eight lines, there are a total of~$48$ possible tuples. It follows from the previous sections that in this example, there are four phase-transition values of $\lambda$, namely $1/4$, $3/8$, $1/2$ and $3/4$. Therefore we need to distinguish between five possible intervals of $\lambda$. In Table~\ref{tab:OverviewExampleNumberOfOccurances} we provide an overview, which follows directly from the results in the previous sections. A direct consequence is the following corollary.

\begin{table}[h!]
\centering
  \begin{tabular}{| l | c c c c c c |}
  	\hline
	$\#\{(i,\ell) : Z(i,\ell) = j\}$ & $j=0$ & $j=1$ & $j=2$ & $j=3$ & $j=4$ & $j=5$ \\
	\hline
	$0 < \lambda < 1/4$ & 48 & 0 & 0 & 0 & 0 & 0 \\
	$1/4 < \lambda < 3/8$ & 40 & 8 & 0 & 0 & 0 & 0 \\
	$3/8 < \lambda < 1/2$ & 32 & 8 & 0 & 0 & 8 & 0 \\
	$1/2 < \lambda < 3/4$ & 8 & 32 & 0 & 0 & 8 & 0 \\
	$3/4 < \lambda < 1$ & 0 & 32 & 8 & 0 & 8 & 0 \\
	\hline
  \end{tabular}
  \caption{Overview of number of tuples that lead to the disconnection of $j$ cities from the city with largest demand.}
  \label{tab:OverviewExampleNumberOfOccurances}
\end{table}

\begin{corollary}
Consider the 6-node network topology. If $\lambda \in (0,1/4)$, then as $x \rightarrow \infty$,
\begin{align*}
\Prob(S > x) \sim O(x^{-2\alpha}).
\end{align*}
Otherwise, as $x \rightarrow \infty$,
\begin{align*}
\Prob(S > x) \sim C(\lambda) K x^{-\alpha},
\end{align*}
where
\begin{align*}
C(\lambda) = \left\{ \begin{array}{ll}
\left(\lambda/6 \right)^{ \alpha} & \textrm{if } 1/4 < \lambda < 3/8,\\
\left(\lambda/6 \right)^{ \alpha} + \left(2\lambda/3 \right)^{ \alpha}& \textrm{if } 3/8 < \lambda < 1/2,\\
4\left(\lambda/6 \right)^{ \alpha} + \left(2\lambda/3 \right)^{ \alpha}& \textrm{if } 1/2 < \lambda < 3/4,\\
4\left(\lambda/6 \right)^{ \alpha} + \left(\lambda/3 \right)^{ \alpha} + \left(2\lambda/3 \right)^{ \alpha}& \textrm{if } 3/4 < \lambda < 1.
\end{array}\right.
\end{align*}

\end{corollary}


\section{Numerical validation on synthetic networks}

\begin{algorithm}[H]
\begin{small}
    \caption{Monte-Carlo simulation for synthetic networks}
    \label{alg:montecarlo}
    \begin{algorithmic}[1]
    	\Inputs{Parameters for sampling network and city sizes\\
    	Line limit scaling parameter $\lambda \in(0,1)$ \\
    	Number of blackout realizations  $\niter$ \\\vspace{0.5cm}}
        \Initialize{
        \State Sample network topology $\G$ and city sizes (C)\\
         \State Solve OPF without line limits; let $\bfl^*$ be the resulting power flows (B, D)
		\State Solve OPF with line limits $\bar{\bfl}=\lambda |\bfl^*|$ (B, D)
		\State Choose a random subset $\LL^\prime$ of lines, with cardinality $\niter$ \vspace{0.5cm}
        }
        \Procedure{}{}
		\For{ $\ell$ in $\LL^\prime$: }
					\State Sample city sizes (A, C)
					\State Solve OPF without line limits; let $\bfl^*$ be the resulting power flows  (B, D)
					\State Set $\Gconn = $ True
					\State Remove line $\ell$ from $\G$
        			     \State Set $\Gc = $True
        				\While{$\Gc = $True \textbf{and} $\Gconn =$ True}
        					
        					\State Recompute power flows (B)
        					\State Remove from $\G$ the line with the largest relative overload wrt. the original line limits $|\bfl^*|$ , if any
        					\State If a line was removed, let $\Gc =$ True; otherwise, set $\Gc = $ False.
        					\State If $\G$ is still connected, let $\Gconn =$ True; otherwise, set $\Gconn = $ False.
        		\EndWhile
        		\State Shed load/generation within each component of $\G$ to achieve power balance (\textit{load shedding event}) (B)
        		\State Store blackout realization
        	\EndFor
        \EndProcedure
    \end{algorithmic}
    \end{small}
\end{algorithm}

In this section, we numerically validate our theory by simulating synthetic blackouts using the three-stages mathematical model described in Section~\ref{s:CascadeModel}.
We use different random graph topologies, the Watts-Strogatz model~\cite{Watts1998} and the more recent SynGrid model developed in~\cite{Wang2018syngrid}, as well as IEEE test cases.
The Watts-Strogatz model produces graphs with small world topologies, while the SynGrid model produces random graphs with realistic and statistically correct power grid topologies. With the Watts-Strogatz model, we aim to analyze the impact of network topology via changing the rewiring probability $p$, while with the SynGrid model we study the impact of the line limit loading factor $\lambda$.

In Sections~\ref{sec:fixed}-~\ref{sec:pwl}, we relax various assumptions upon which our model is based, in order to test the sensitivity of the theory.
The general algorithm is summarized in Algorithm~\ref{alg:montecarlo}. Steps marked with (A), (B), (C) and (D) differ according to the particular assumption that is relaxed. In particular, these steps control:
(A) whether or not to resample city sizes at each iteration;
(B) whether to use a DC or AC power flow model;
(C) which random graph and city size distribution to use;
(D) whether to use a convex quadratic or a convex piecewise linear cost function in the OPF.

 City sizes are sampled from both a Pareto distribution with parameters $\alpha^{(d)}=1.37$ and $\xmin^{(d)}=5 \cdot 10^4$ (according to the results in Table~\ref{tab:PLFIT_stats}), and a uniform distribution with the same mean as the Pareto distribution.
The two different city sizes distributions are used to corroborate our theory from different angles. On the one hand, we show that when city sizes are heavy-tailed so are blackout sizes and the tail indexes are similar, as predicted by our asymptotic theory. On the other hand, when city sizes are not heavy-tailed, we show that the blackout size distribution is also not heavy-tailed.
Each iteration of the simulation stops when there are no more overloaded lines, or the graph got disconnected in two islands.

The results show that our framework is robust with respect to violations of the above-mentioned key assumptions, different topologies, line limit loading factors, and city size distributions.
Later, in Section~\ref{sc:scigrid}, we perform simulations on a model for the German transmission grid, where all of the simplifying assumptions, except for the DC flow model, are relaxed in favour of realistic parameters.

\subsection{Frozen city sizes}\label{sec:fixed}
Our mathematical framework described in Section~\ref{s:CascadeModel} models city sizes $X_1,\ldots,X_n$ as Pareto distributed random variables, while in the real world the sizes of cities served by a given power grid are essentially fixed. In this section, we show that our results still hold for a version of our model where the city sizes are kept fixed (i.e., they are not resampled at each iteration), provided that the network is large enough to avoid finite size effects.

With respect to Algorithm~\ref{alg:montecarlo}, step A is skipped, so that city sizes are kept frozen throughout the simulation, and the only source of randomness is the first outage event. The rest of the framework is unchanged, i.e. a DC flow model with a convex quadratic objective function is used (steps B, D), and both a Pareto and a uniform distribution are used for city sizes (step C).

The results are reported in Tables~\ref{tab:synthetic_WS},~\ref{tab:synthetic_syngrid} and Figures~\ref{fig:synthetic_WS},~\ref{fig:synthetic_syngrid}.
 We observe that, in the Watts-Strogatz case, the tail index estimates $\alpha$ are within one standard deviation apart from the city sizes index $\alpha^{(d)}=1.37$,
 consistently across different values of the rewiring probability $p$. The result is corroborated by the analysis of the Hill plots in Fig.~\ref{fig:synthetic_WS}, where we observe that the flat region of the graph $\xmin \to \alpha(\xmin)$ is close to $\alpha^{(\text{city})}=1.37$. Conversely, there is no indication of a heavy tail for the blackout size distribution in the case of uniform city sizes, as can be inferred from the Hill plots in Fig.~\ref{fig:synthetic_WS}.

For the SynGrid model, the tail index estimates $\alpha$ are within two standard deviations apart from the city sizes index, consistently across different values of $\lambda$. The fit is more accurate for larger values or $\lambda$, as can be observed from the Hill plots in Fig.~\ref{fig:synthetic_syngrid}, while the outliers at the far end of the tail could be attributed to finite size effects. A worse fit is observed in the case of smaller networks, as can be seen from Fig.~\ref{fig:different_sizes}. In particular, it appears that the estimated tail index of blackout sizes decreases monotonically to that of city sizes, and that convergence is achieved only for $n \ge 10000$.
Once more, in the case of light-tailed city sizes the heavy-tail behavior of blackout size is not observed.

We conclude that, provided that the network is large enough,
the Pareto law of blackout sizes is inherited from that of city sizes as predicted by our model, even in the case where city sizes are fixed and a realistic power grid topology is used. Moreover, when city sizes are light-tailed, the blackout size are not heavy-tailed, providing further support to our theory.


\begin{figure}[hbt!]
  \centering
          \includegraphics[width=\columnwidth]{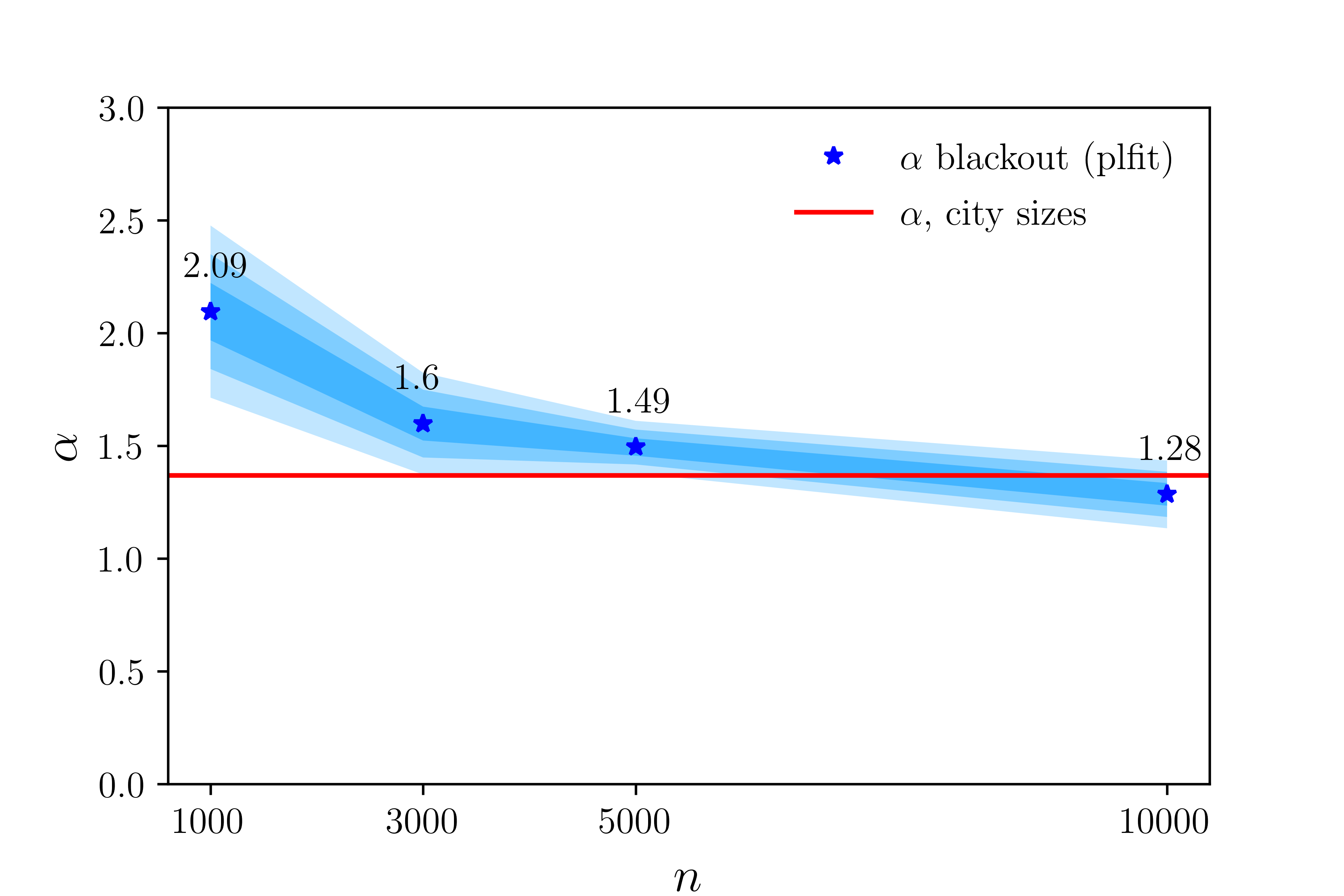}
       \caption{Visualization of PLFIT-based tail index for synthetically generated blackout keeping city sizes fixed, as the size $n$ of the network increases. The power grid topology is based on the SynGrid random graph model~\cite{Wang2018syngrid} with $m \sim 2.5 \, n$ lines, $\lambda=0.9$, and city sizes are sampled form a Pareto distribution. The shaded bands correspond to $1,2$ and $3$ standard deviations.
}
         \label{fig:different_sizes}
          \end{figure}

\begin{table}[H]
\centering
\begin{tabular}{c|c|c|c|c|c}
\hline
\hline
\small{$p$} & \small{$\niter$} &
\small{$\nnonzero$}& \small{$\ntail$}
 &\small{$\alpha$}   & \small{city sizes}   \\
\hline
$0.3$ & $10000$ & $5066$ & $1336$ & $1.56 \pm 0.04$   & pareto\\
$0.5$ & $10000$ & $5189$ & $366$  & $1.31 \pm 0.07$  & pareto\\
$0.7$ & $10000$ & $5257$ & $504$  & $1.41 \pm 0.06$  & pareto\\
$0.3$ & $10000$ & $5225$ & $1681$ & $4.31 \pm 0.11$  & uniform\\
$0.5$ & $10000$ & $5331$ & $1692$ & $4.47 \pm 0.11$  & uniform\\
$0.7$ & $10000$ & $5332$ & $1570$ & $4.80 \pm 0.12$  & uniform\\
\hline
\hline
\end{tabular}
\caption{\small PLFIT statistics for synthetically generated blackout data keeping city sizes fixed, using a Watts-Strogatz random graph model for the power grid topology with $n=10000$ nodes, $ m=20000$ lines, rewiring probabilities $p\in\{0.3,0.5,0.7\}$, mean degree $K=4$, line limit scaling factor $\lambda=0.7$, and different city sizes distribution.  $\nnonzero$ is the number of nonzero realizations, and $\ntail$ is the number of realizations $x_i\ge\hatxmin$.}
  \label{tab:synthetic_WS}
\end{table}

\begin{table}[H]
\centering
\begin{tabular}{c|c|c|c|c|c}
\hline
\hline
\small{$\lambda$} & \small{$\niter$} &
\small{$\nnonzero$}& \small{$\ntail$}
 &\small{$\alpha$}  & city sizes  \\
\hline
$0.9$ & $10000$ & $8883$ & $657$ & $1.28\pm 0.05 $ & pareto \\
$0.7$ & $10000$ & $4308$ & $1125$ & $1.35\pm 0.04 $ & pareto \\
$0.5$ & $10000$ & $1189$ & $435$ & $1.23\pm 0.06 $  &  pareto\\
$0.9$ & $10000$ & $8919$ & $763$ & $3.63\pm 0.13 $ & uniform \\
$0.7$ & $10000$ & $4764$ & $404$ & $4.00\pm 0.20 $  & uniform \\
$0.5$ & $10000$ & $1615$ & $79$ & $3.07\pm 0.35 $ &  uniform\\
\hline
\hline
\end{tabular}
\caption{\small PLFIT statistics for synthetically generated blackout data keeping city sizes fixed, using the SynGrid random graph model~\cite{Wang2018syngrid} for the power grid topology with $n=10000$ nodes, $m \sim 2.5 \, n$ lines, for different values of $\lambda$ and city sizes distribution.  $\nnonzero$ is the number of nonzero realizations, and $\ntail$ is the number of realizations $x_i\ge\hatxmin$.}
  \label{tab:synthetic_syngrid}
\end{table}


\begin{figure}[hbt!]
 \begin{subfigure}[t]{0.45\columnwidth}
  \centering
          \includegraphics[width=\textwidth]{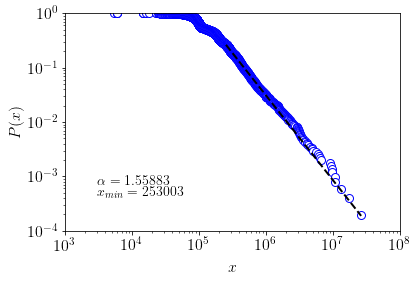}
          \caption{\small{PLFIT, Pareto city sizes, $p=0.3$}}
\end{subfigure}\hfill
 \begin{subfigure}[t]{0.45\columnwidth}
  \centering
      \includegraphics[width=\textwidth]{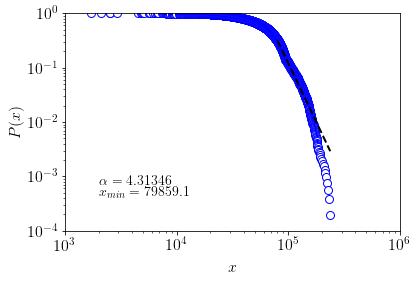}
          \caption{\small{PLFIT, uniform city sizes, $p=0.3$}}
\end{subfigure}\hfill
\begin{subfigure}[t]{0.45\columnwidth}
  \centering
          \includegraphics[width=\textwidth]{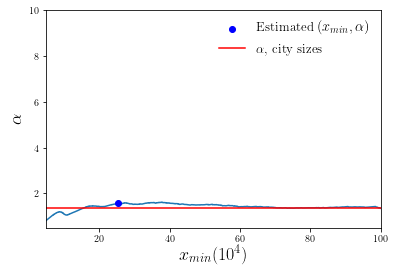}
          \caption{\small{Hill plot, Pareto city sizes, $p=0.3$.}}
\end{subfigure}\hfill
 \begin{subfigure}[t]{0.45\columnwidth}
  \centering
      \includegraphics[width=\textwidth]{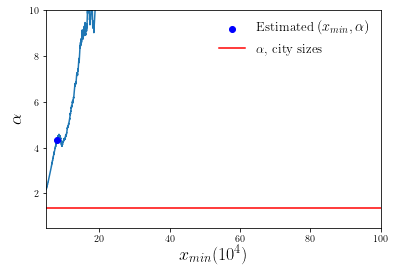}
          \caption{\small{Hill plot, uniform city sizes, $p=0.3$.}}
\end{subfigure}\hfill
\begin{subfigure}[t]{0.45\columnwidth}
  \centering
          \includegraphics[width=\textwidth]{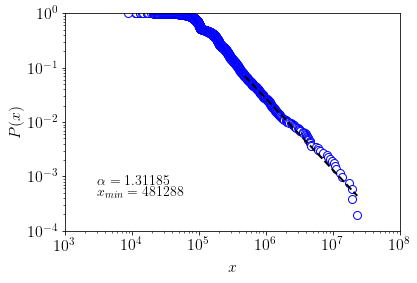}
          \caption{\small{PLFIT, Pareto city sizes, $p=0.5$.}}
\end{subfigure}\hfill
 \begin{subfigure}[t]{0.45\columnwidth}
  \centering
      \includegraphics[width=\textwidth]{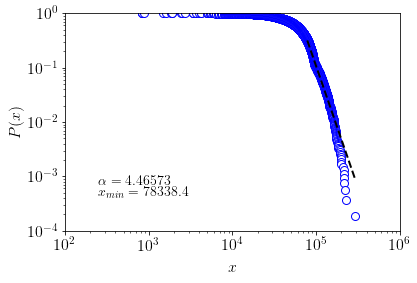}
          \caption{\small{PLFIT, uniform city sizes, $p=0.5$.}}
\end{subfigure}\hfill
\begin{subfigure}[t]{0.45\columnwidth}
  \centering
          \includegraphics[width=\textwidth]{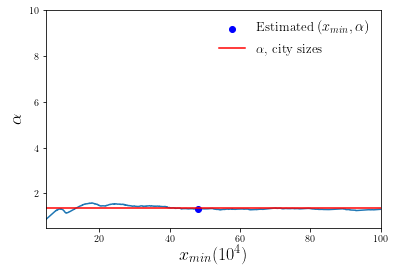}
          \caption{\small{Hill plot, Pareto city sizes, $p=0.5$.}}
\end{subfigure}\hfill
 \begin{subfigure}[t]{0.45\columnwidth}
  \centering
      \includegraphics[width=\textwidth]{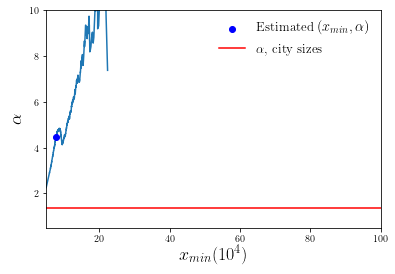}
          \caption{\small{Hill plot, uniform city sizes, $p=0.5$.}}
\end{subfigure}\hfill
 \begin{subfigure}[t]{0.45\columnwidth}
  \centering
          \includegraphics[width=\textwidth]{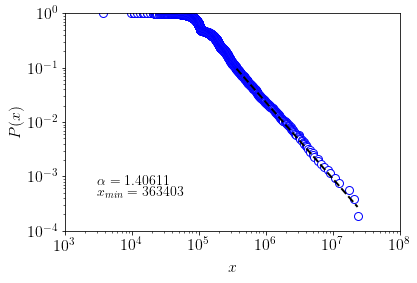}
          \caption{\small{PLFIT, Pareto city sizes, $p=0.7$.}}
\end{subfigure}\hfill
 \begin{subfigure}[t]{0.45\columnwidth}
  \centering
      \includegraphics[width=\textwidth]{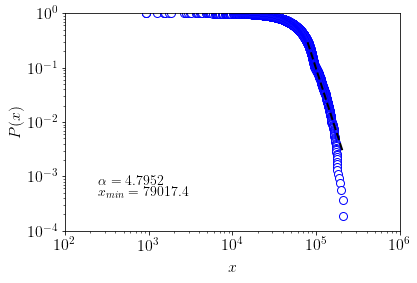}
          \caption{\small{PLFIT, uniform city sizes, $p=0.7$.}}
\end{subfigure}\hfill
\begin{subfigure}[t]{0.45\columnwidth}
  \centering
          \includegraphics[width=\textwidth]{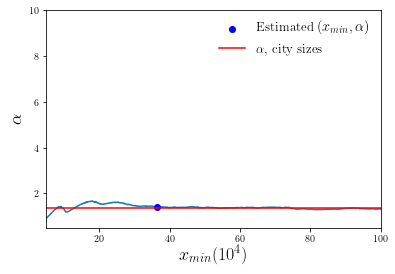}
          \caption{\small{Hill plot, Pareto city sizes, $p=0.7$.}}
\end{subfigure}\hfill
 \begin{subfigure}[t]{0.45\columnwidth}
  \centering
      \includegraphics[width=\textwidth]{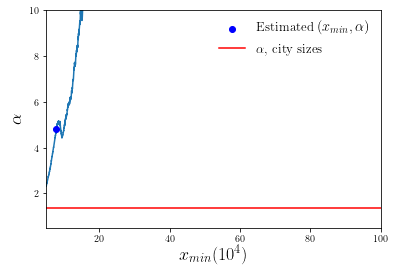}
          \caption{\small{Hill plot, uniform city sizes, $p=0.3$.}}
\end{subfigure}
\caption{\small Results for synthetically generated blackout data, using a Watts-Strogatz model for the power grid topology and keeping city sizes fixed, for different rewiring probabilities $p$ and different city sizes distributions. a,b,e,f,i,j): PLFIT results and log-log plot of CCDF; c,d,g,h,k,l) Hill plot: red line corresponds to the city sizes tail index $\alpha^{(\text{city})}=1.37$.\label{fig:synthetic_WS}}
      \end{figure}

\begin{figure}[h!]
 \begin{subfigure}[t]{0.45\columnwidth}
  \centering
          \includegraphics[width=\textwidth]{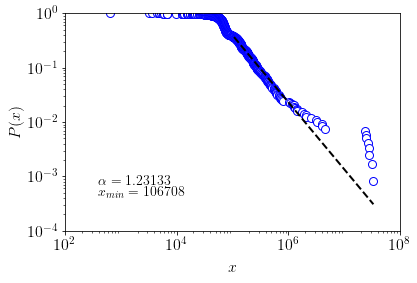}
          \caption{\small{PLFIT, Pareto city sizes, $\lambda=0.5$.}}
\end{subfigure}\hfill
 \begin{subfigure}[t]{0.45\columnwidth}
  \centering
          \includegraphics[width=\textwidth]{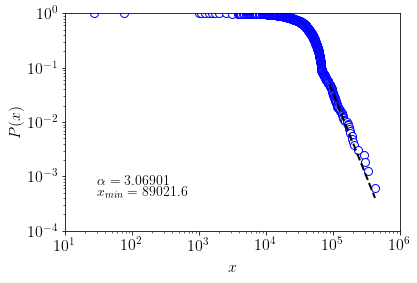}
          \caption{\small{PLFIT, uniform city sizes, $\lambda=0.5$.}}
\end{subfigure}\hfill
\begin{subfigure}[t]{0.45\columnwidth}
  \centering
          \includegraphics[width=\textwidth]{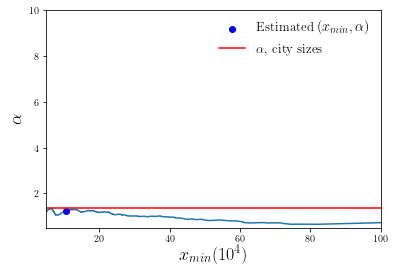}
          \caption{\small{Hill plot, Pareto city sizes, $\lambda=0.5$.}}
\end{subfigure}\hfill
 \begin{subfigure}[t]{0.45\columnwidth}
  \centering
          \includegraphics[width=\textwidth]{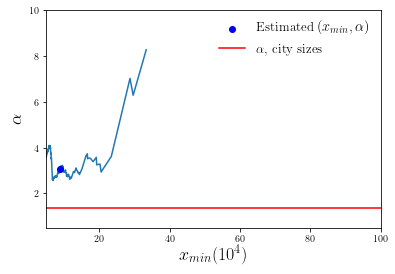}
          \caption{\small{Hill plot, uniform city sizes, $\lambda=0.5$.}}
\end{subfigure}
 \begin{subfigure}[t]{0.45\columnwidth}
  \centering
          \includegraphics[width=\textwidth]{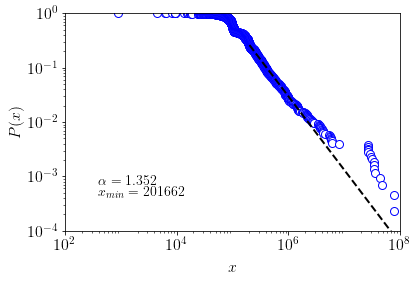}
          \caption{\small{PLFIT, Pareto city sizes, $\lambda=0.7$.}}
\end{subfigure}\hfill
 \begin{subfigure}[t]{0.45\columnwidth}
  \centering
          \includegraphics[width=\textwidth]{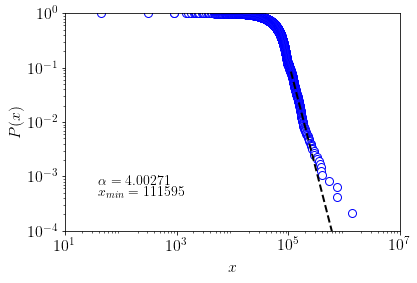}
          \caption{\small{PLFIT, uniform city sizes, $\lambda=0.7$.}}
\end{subfigure}\hfill
\begin{subfigure}[t]{0.45\columnwidth}
  \centering
          \includegraphics[width=\textwidth]{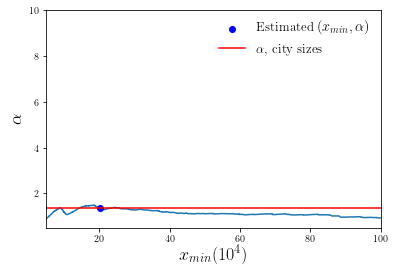}
          \caption{\small{Hill plot, Pareto city sizes, $\lambda=0.7$.}}
\end{subfigure}\hfill
 \begin{subfigure}[t]{0.45\columnwidth}
  \centering
          \includegraphics[width=\textwidth]{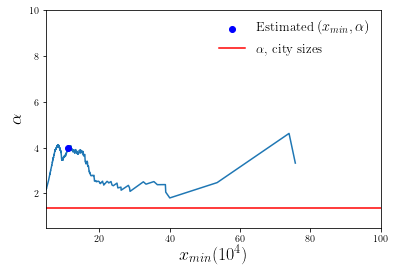}
          \caption{\small{Hill plot, uniform city sizes, $\lambda=0.7$.}}
\end{subfigure}
 \begin{subfigure}[t]{0.45\columnwidth}
  \centering
          \includegraphics[width=\textwidth]{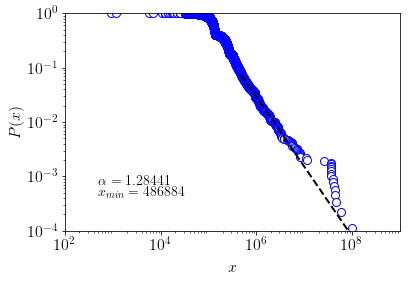}
          \caption{\small{PLFIT, Pareto city sizes, $\lambda=0.9$.}}
\end{subfigure}\hfill
 \begin{subfigure}[t]{0.45\columnwidth}
  \centering
          \includegraphics[width=\textwidth]{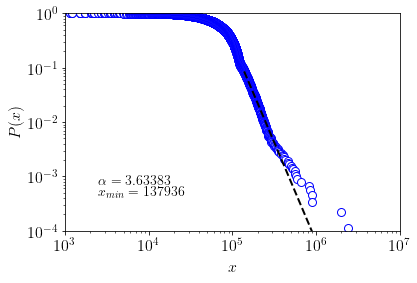}
          \caption{\small{PLFIT, uniform city sizes, $\lambda=0.9$.}}
\end{subfigure}\hfill
\begin{subfigure}[t]{0.45\columnwidth}
  \centering
          \includegraphics[width=\textwidth]{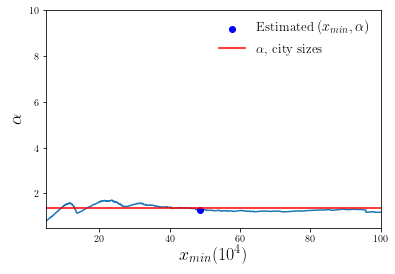}
          \caption{\small{Hill plot, Pareto city sizes, $\lambda=0.9$.}}
\end{subfigure}\hfill
 \begin{subfigure}[t]{0.45\columnwidth}
  \centering
          \includegraphics[width=\textwidth]{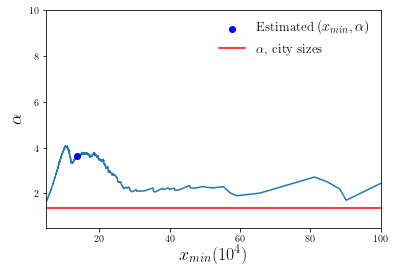}
          \caption{\small{Hill plot, uniform city sizes, $\lambda=0.9$.}}
\end{subfigure}
\caption{ \small Results for synthetically generated blackout data, using the SynGrid model in~\cite{Wang2018syngrid} for the power grid topology and keeping city sizes fixed, for line limit loading factors $\lambda=0.9$, and different city sizes distributions. a,b,e,f,i,j): PLFIT results and log-log plot of CCDF; c,d,g,h,k,l) Hill plot: red line corresponds to the city sizes tail index $\alpha^{(\text{city})}=1.37$.\label{fig:synthetic_syngrid}}
      \end{figure}

\subsection{Voltage limits and AC power flow model}\label{sec:AC}

In this section, we illustrate an extension of our framework to the AC power flow model that takes into account voltage limits and losses.  The experiments are performed using MATPOWER~\cite{Zimmerman2011}. As our three stages framework described in Section~\ref{s:CascadeModel} is devised with a DC power flow model in mind, we introduce the following modifications.
%
First, the planning and operational OPFs, as well as the calculation of the power flow (PF) redistribution after line failures in the emergency phase, are performed using the AC model~\cite{Bienstockbook}.
The required parameters to perform these calculations, such as voltage limits, line resistances and reactances are pulled from the MATPOWER test cases (as opposed to the original DC framework, which assumed unitary reactances and zero resistances), and the OPF and PF calculations are performed using MATPOWER's built-in routines.

Second, due to line losses, the active power injected into line $\ell=\{i,j\}$ at its sending end, denoted by $f_{ij}$, will differ from the one at the receiving end, $f_{ji}$.
Thus, we express the line limit constraints in terms of the maximum of active power flows at the two ends of the line. In particular, the line limit constraints in the operational OPF are given by
\begin{equation}
|f_{\ell}| < \fbar_{\ell} := \lambda   \max (|f_{ij}|, |f_{ji}|), \ell=\{i,j\}\in\LL,
\end{equation}
for a loading factor $\lambda\in(0,1)$.  Reactive power flow constraints are not considered.

Finally, the blackout size differs from the definition in Eq.~\eqref{eq:S_def} to take into account losses.
Specifically, in the island with a shortage demand is shed until total demand matches total generation minus network losses.
More precisely, after the first disconnection happens, we define $A_1$ be the island such that
$\sum_{i\in A_1} g_i-d_i < \sum_{i\notin A_1} g_i-d_i$, where $g_i$ is the generation at node $i$ as given by the operational OPF, and $d_i$ the corresponding demand.
Thus, the blackout size is defined as
\begin{equation}
S = \left\vert \sum_{i\in A_1} (g_i-d_i) - \eta \right\vert,
\end{equation}
where $\eta$ is a loss component defined as the sum of active power losses on the lines belonging to $A_1$, obtained after solving a new OPF in the subnetwork $A_1$. 

The rest of the framework is the same as in Section~\ref{s:CascadeModel}. In particular, generation limit are not considered, and all generators are assumed to be identical with cost functions $C_i(g_i) = g_i^2/2$. We perform experiments on the IEEE 14-bus, 30-bus and 39-bus networks from~\cite{Zimmerman2011}, which are modified accordingly to the above description. For each simulation we run $10000$ iterations using a loading factor $\lambda=0.9$, and we discard iterations resulting in a non-feasible AC-OPF.

With respect to Algorithm~\ref{alg:montecarlo}, the above changes affect steps $B$ by using an AC model for the OPF and PF computations rather than a DC model, and by adding losses to the blackout size realizations.
The rest of the framework remains unchanged, i.e. city sizes are resampled at each iteration (step A), and both a Pareto and a uniform distribution are used for city sizes (step C).
%

The results are reported in Table~\ref{tab:AC} and Fig.~\ref{fig:AC}. We observe that, in the case of Pareto distributed city sizes, there is indication of heavy tail for the blackout size distribution. The tail is lighter than that of city sizes for the 14-bus and 30-bus networks, while a much better fit is observed for the 39-bus test case.
At the same time, there is no indication of heavy tail for the blackout size distribution in the case of uniform city sizes, as can be inferred from the Hill plots comparison in Fig.~\ref{fig:AC}. We conclude that, even when the DC flow model assumption is violated, the city size distribution still plays an important role in affecting the distribution of blackout sizes.

We remark that these results are based on a partial adaptation of our DC framework to the AC case. In order to develop a more thorough mathematical theory for the AC model, one would have to modify and extend our framework considerably, especially the design and operational phase. In order to do such an extension, it would be necessary to develop a mechanism for assigning line resistances and voltage limits in our design and operational phases, for instance, which is beyond the scope of this study.
A direction for rigorous follow-up work would be to focus on special networks topologies (like rings), or to use a load flow model which is intermediate between DC and AC, such as lossless AC~\cite{MolzahnHiskens}.

\begin{table}
\centering
\begin{tabular}{c|c|c|c|c|c}
\hline
\hline
\small{test case} & \small{$\niter$} &
\small{$\nnonzero$}& \small{$\ntail$}
 &\small{$\alpha$}   & city sizes  \\
\hline
14-bus & $10000$ & $9145$ & $2227$ & $1.43 \pm 0.03$ &   pareto\\
30-bus & $10000$ & $7691$ & $781$ & $1.97 \pm 0.07$  &   pareto\\
39-bus & $10000$ & $6348$ &  $990$ & $1.24\pm 0.04$ &   pareto\\
14-bus & $10000$ & $9158$ & $171$ & $7.29 \pm 0.56$ &   uniform\\
30-bus &$10000$  & $9017$ & $149$ & $6.52 \pm 0.53$ &   uniform\\
39-bus &$10000$  & $8621$& $127$ & $7.20\pm 0.64$ &  uniform\\
\hline
\hline
\end{tabular}
\caption{\small PLFIT statistics for synthetically generated blackout data using an AC power flow model for different MATPOWER testcases, line limit loading factor $\lambda=0.9$, and different city sizes distributions. $\nnonzero$ is the number of nonzero realizations, and $\ntail$ is the number of realizations $x_i\ge\hatxmin$.}
  \label{tab:AC}
\end{table}

\begin{figure}[ht!]
 \begin{subfigure}[t]{0.45\columnwidth}
  \centering
          \includegraphics[width=\textwidth]{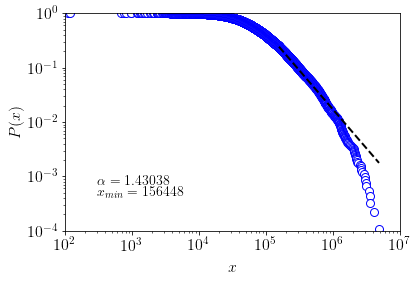}
          \caption{\small{PLFIT, Pareto city sizes, 14-bus network.}}
\end{subfigure}\hfill
 \begin{subfigure}[t]{0.45\columnwidth}
  \centering
          \includegraphics[width=\textwidth]{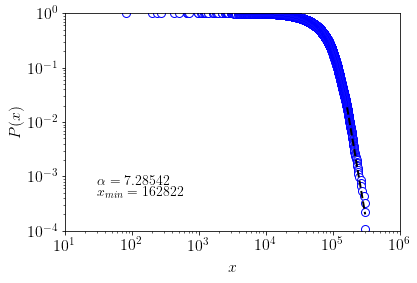}
          \caption{\small{PLFIT, uniform city sizes, 14-bus network.}}
\end{subfigure}\hfill
 \begin{subfigure}[t]{0.45\columnwidth}
  \centering
                   \includegraphics[width=\textwidth]{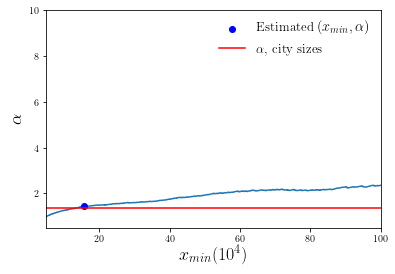}
          \caption{\small{Hill plot, Pareto city sizes, 14-bus network.}}
\end{subfigure}\hfill
 \begin{subfigure}[t]{0.45\columnwidth}
  \centering
                   \includegraphics[width=\textwidth]{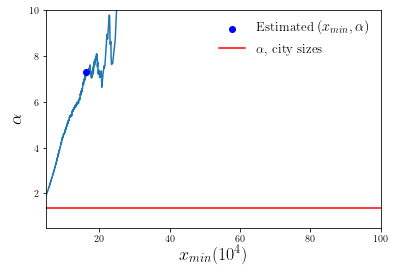}
          \caption{\small{Hill plot, uniform city sizes, 14-bus network.}}
\end{subfigure}\hfill
 \begin{subfigure}[t]{0.45\columnwidth}
  \centering
          \includegraphics[width=\textwidth]{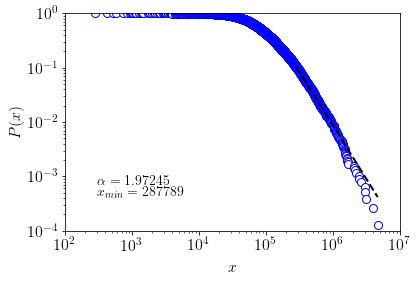}
          \caption{\small{PLFIT, Pareto city sizes, 30-bus network.}}
\end{subfigure}\hfill
 \begin{subfigure}[t]{0.45\columnwidth}
  \centering
          \includegraphics[width=\textwidth]{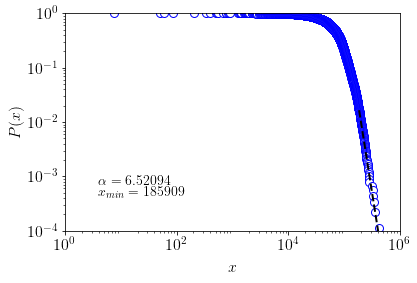}
          \caption{\small{PLFIT, uniform city sizes, 30-bus network.}}
\end{subfigure}\hfill
 \begin{subfigure}[t]{0.45\columnwidth}
  \centering
                   \includegraphics[width=\textwidth]{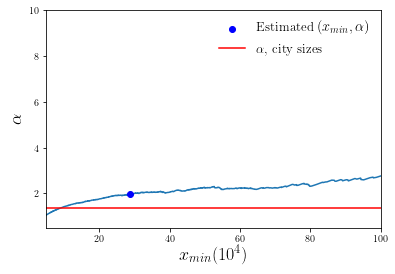}
          \caption{\small{Hill plot, Pareto city sizes, 30-bus network.}}
\end{subfigure}\hfill
 \begin{subfigure}[t]{0.45\columnwidth}
  \centering
                   \includegraphics[width=\textwidth]{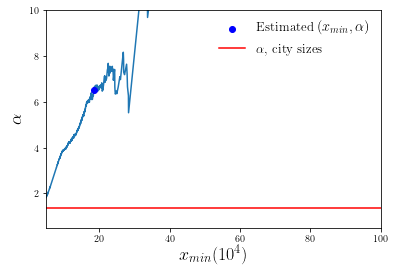}
          \caption{\small{Hill plot, uniform city sizes, 30-bus network.}}
\end{subfigure}\hfill
%
 \begin{subfigure}[t]{0.45\columnwidth}
  \centering
          \includegraphics[width=\textwidth]{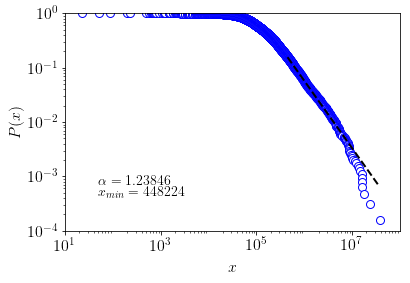}
          \caption{\small{PLFIT, Pareto city sizes, 39-bus network.}}
\end{subfigure}\hfill
 \begin{subfigure}[t]{0.45\columnwidth}
  \centering
          \includegraphics[width=\textwidth]{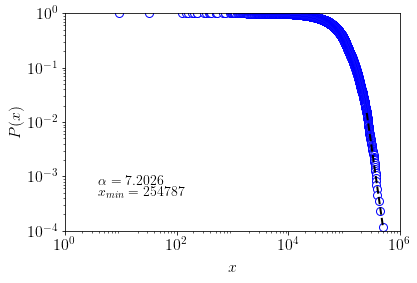}
          \caption{\small{PLFIT, uniform city sizes, 39-bus network.}}
\end{subfigure}\hfill
 \begin{subfigure}[t]{0.45\columnwidth}
  \centering
                   \includegraphics[width=\textwidth]{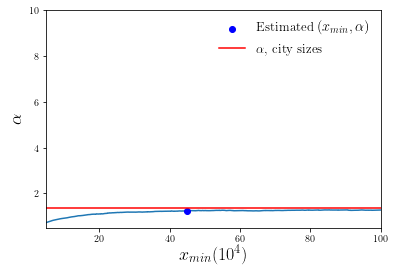}
          \caption{\small{Hill plot, Pareto city sizes, 39-bus network.}}
\end{subfigure}\hfill
 \begin{subfigure}[t]{0.45\columnwidth}
  \centering
                   \includegraphics[width=\textwidth]{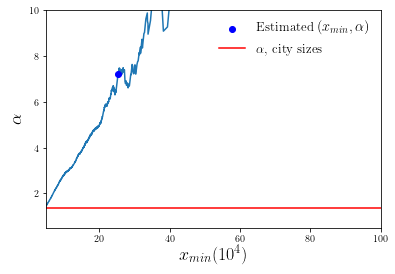}
          \caption{\small{Hill plot, uniform city sizes, 39-bus network.}}
\end{subfigure}\hfill
\caption{\small Results for synthetically generated blackout data using the AC power flow model for different MATPOWER testcases, line limit loading factor $\lambda=0.9$, and different city sizes distributions. a,b,e,f,i,j): PLFIT results and log-log plot of CCDF; c,d,g,h,k,l) Hill plot: red line corresponds to the city sizes tail index $\alpha^{(\text{city})}=1.37$.\label{fig:AC}}
%
  \end{figure}

\subsection{Convex piecewise-linear generator cost function}\label{sec:pwl}

Our theoretical framework assumes that the cost function is a convex quadratic function (in particular, the cost of generator $i$ is of the form $C_i(g_i)=g_i^2/2$).
In this section, we investigate the sensitivity of our results with respect to violation of this assumption, by simulating synthetic blackouts using a convex piecewise linear cost function instead and comparing the results.

With respect to Algorithm~\ref{alg:montecarlo}, the change affects steps D by modifying the objective function in the OPF.
We consider two MATPOWER testcases, \textit{case30} and \textit{case30pwl}, which only differs in the choice of the cost function (respectively, convex quadratic and convex piecewise linear). The testcases are modified according to our framework in Section~\ref{s:CascadeModel}. In particular, the cost function of generators of \textit{case30} is set to $C_i(g_i)=g_i^2/2$, while that of the generators of \textit{case30pwl} is taken from the testcase itself and set to the convex pwl function defined by the break-points $(0,0), (12, 144), (36, 1008), (60, 2832)$, expressed in (MW, \$/h).
The rest of the framework remains unchanged, i.e. city sizes are resampled at each iteration (step A), a DC model is used (step B), and a Pareto distribution is used for city sizes (step C).

The results, visualized in Fig.~\ref{fig:pwl}, show that the Pareto fits for the quadratic and piecewise linear case are very similar. In particular, the Hill plots show a remarkable fit for both cases, demonstrating the robustness of our theory to the form of the cost function.

This can be explained by observing that a key argument in our original framework is the fact that the generation schedule in the planning phase is as decentralized as possible, i.e.~ $\bg^{(\text{planning})}=\frac{1}{n}\sum_{i=1}^n X_i \be$ (Section~\ref{ss:description}), due the strict convexity of the quadratic objective function. This property is partially retained with a convex piecewise linear function, since generation  will be optimally allocated to the break-points of the function, thus preserving the decentralization feature. The main difference is when an amount of generation greater than the largest break-point must be produced at a certain location, in which case it becomes optimal to produce as much as needed at that location. This, in turn, results in more local generation at the largest city, and thus smaller blackouts. As a results, a piecewise  linear cost function results in blackouts with smaller magnitudes than in the case of a quadratic function, while preserving the Pareto shape, as it can be appreciated in Fig.~\ref{fig:pwl} (a), (b).

\begin{figure}
\begin{subfigure}[t]{0.45\columnwidth}
  \centering
          \includegraphics[width=\textwidth]{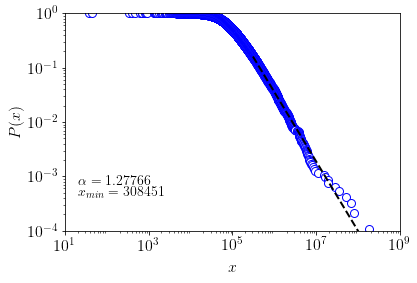}
          \caption{\small{PLFIT, quadratic cost function, 30-bus network.}}
\end{subfigure}\hfill
 \begin{subfigure}[t]{0.45\columnwidth}
  \centering
          \includegraphics[width=\textwidth]{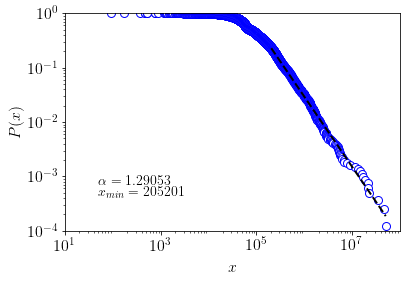}
          \caption{\small{PLFIT, pwl cost function, 30-bus network.}}
\end{subfigure}\hfill
 \begin{subfigure}[t]{0.45\columnwidth}
  \centering
                   \includegraphics[width=\textwidth]{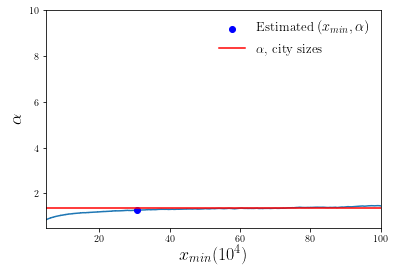}
          \caption{\small{Hill plot, quadratic cost function, 30-bus network.}}
\end{subfigure}\hfill
 \begin{subfigure}[t]{0.45\columnwidth}
  \centering
                   \includegraphics[width=\textwidth]{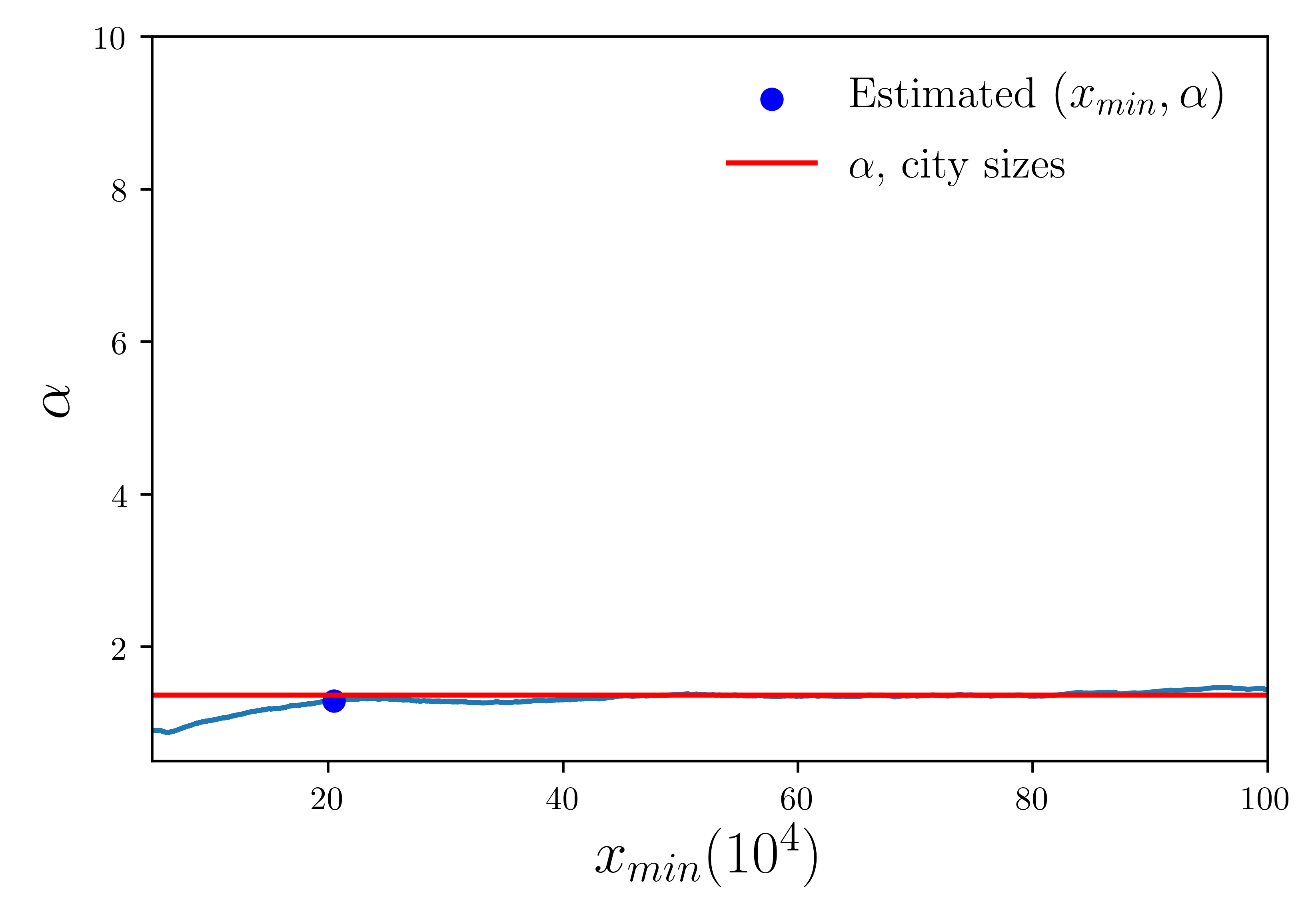}
          \caption{\small{Hill plot, pwl cost function, 30-bus network.}}
\end{subfigure}\hfill
\caption{\small Results for synthetically generated blackout data with quadratic and pwl generation cost function, for the IEEE 30-bus network, loading factor $\lambda=0.9$, Pareto city size distribution. a,b,): PLFIT results and log-log plot of CCDF; c,d) Hill plot: red line corresponds to the city sizes tail index 
$\alpha^{(\text{city})}=1.37$.\label{fig:pwl}}
\end{figure}

\clearpage
\section{SCIGRID case study}\label{sc:scigrid}

\begin{algorithm}[H]
\begin{small}
    \caption{Monte Carlo simulation - SciGRID German Network}
    \label{alg:montecarlo_scigrid}
    \begin{algorithmic}[1]
    	\Inputs{$\T=\{\text{hourly snapshots for the year 2011}\}$\\ $\lambda = $ line limits scaling factor
    	}
        \Initialize{Solve OPF $\forall t\in T$ with scaling factor $\lambda\in (0,1)$.\\
        Set $\T(\lambda)=\{\text{feasible OPF snapshots}\}$\\
        For all $t\in\T(\lambda)$, let $\G(t)$ be the corresponding network\\\vspace{0.5cm} }
        \Procedure{}{} 
        	\For{ $t\in \T(\lambda)$ }
					\State Remove $1$ line uniformly at random from $\G(t)$.
        			\State Set $\Gtc = $True
        				\While{$\Gtc = $True}
        					\State Shed load/generation within each component of $\G$ in order 						to achieve power balance (\textit{load shedding event})
        					\State Recompute normalized power flows $f_{\ell}$
        					\State Remove from $\G$ all lines exceeding the original line limit.
        					\State If at least one line was removed, let $\Gtc = $True; otherwise, set $\Gtc = $False.
        		\EndWhile
        		\State Store blackout realization
        	\EndFor
        \EndProcedure
    \end{algorithmic}
    \end{small}
\end{algorithm}
We perform our experiments using PyPSA, a free software toolbox for power system analysis~\cite{PyPSA2017}. We use the dataset described in~\cite{SCIGRID1}, which provides a model of the German electricity system based on SciGRID~\cite{SCIGRID00}.
The SciGRID model represents the actual German transmission network with $n = 585$ nodes, $489$ of which carry loads, and $m = 852$ lines.
Data for German city sizes are pulled from~\cite{wpr},
 while the population of German districts, together with the corresponding administrative borders, are taken from~\cite{eurostat} and~\cite{eurostatboundaries,githubboundaries}.

Since the aforementioned datasets do not include nodal demand data, we generate relative nodal demands by using population sizes and administrative borders of German NUTS3 districts, which are then rescaled with hourly nation-wide demand statistics. The procedure, based on~\cite{SCIGRID2}, is explained in detail below.

\subsection{Nodal demand}\label{ss:nodal_demand}
The SciGRID model of the German power grid contains $1423$ generators, $585$ nodes, $489$ demand nodes,  and $m=852$ transmission lines. Geographical coordinates of the demand nodes are denoted by $P_1,\ldots,P_{489}\in\R^2$. Moreover, Germany is partitioned into $402$ NUTS3 administrative districts: we denote by $\popd(j)\in\R$ and $\P_j\subseteq \R^2$, respectively, the population and the polygon describing the administrative borders of district $j$.

In order to attach the loads to the $489$ demand nodes, we proceed as follows.
First, we partition Germany using the Voronoi tessellation associated with the demand nodes. Since some of the nodes lie outside the border of Germany, we consider a bounding square $X$ that contains Germany and all the $P_i$-s, and we define the Voronoi cells:
\[V(P_i)=\{x\in X\,:\,||x-P_i|| \le ||x-P_j||\,\forall j\neq i\}.\]

Then, the population of a node $P_i$, denoted by $\popn(i)$, is taken to be proportional to the overlapping area between $V(P_i)$ and all the NUTS3 districts that intersect $V(P_i)$. Rigorously, if we define the \textit{transfer} matrix
 $\mathbf{T}\in\R^{489\times 402}$ as
\begin{equation}
T_{i,j}=\sum_{j=1}^{402}
\frac{\text{Area}\,(V(P_i)\cap \P_j)}
{\text{Area}(\P_j)},
\end{equation}
the nodal population can be calculated as the matrix-vector product $\popn=\mathbf{T}\,\popd$.

Table~\ref{tab:power_law_stats} and Fig.~\ref{fig:SciGRID} summarizes the key statistics for the power law fits of city, district and nodal population.

Fig.~\ref{fig:Voronoi} shows the different partitions of Germany in NUTS3 districts and Voronoi cells associated with SciGRID demand nodes.
Finally, the demand at node $i$ at time $t$, denoted by $d_i(t)$, is calculated by rescaling the country-wide demand $\dg(t)$ by a factor proportional to the nodal population, as shown in Eq.~\ref{eq:demand_nodes}.

 \begin{equation}\label{eq:demand_nodes}
d_i(t)=\dg(t)\cdot \frac{\popn_i}{\sum_i \popn_i}
\end{equation}

\begin{table}[ht!]
\centering
\begin{tabular}{c|c|c|c|c|c}
\hline
\hline
Quantity & $N$ & $\ntail$ &$\alpha$  & $\xmin\cdot 10^4$ & KS $\pv$  \\
\hline
Cities pop. & $400$ & $271$  & $1.29\pm 0.08$ &  $4.4 \pm 1 $ & $0.35$  \\
Districts pop. & $402$ & $107$  & $2.35\pm 0.34$ &  $22.9 \pm 3.8 $ & $0.65$  \\
Nodal pop. & $498$ & $51$  & $3.77\pm 1.07$ &  $35.7 \pm 7.8 $ & $0.76$  \\
\hline
\hline
\end{tabular}
\caption{PLFIT statistics for German cities, district and nodal population.}
\label{tab:power_law_stats}
\end{table}

\begin{figure}[hbt!]
     \centering
    \includegraphics[width=0.9\columnwidth]{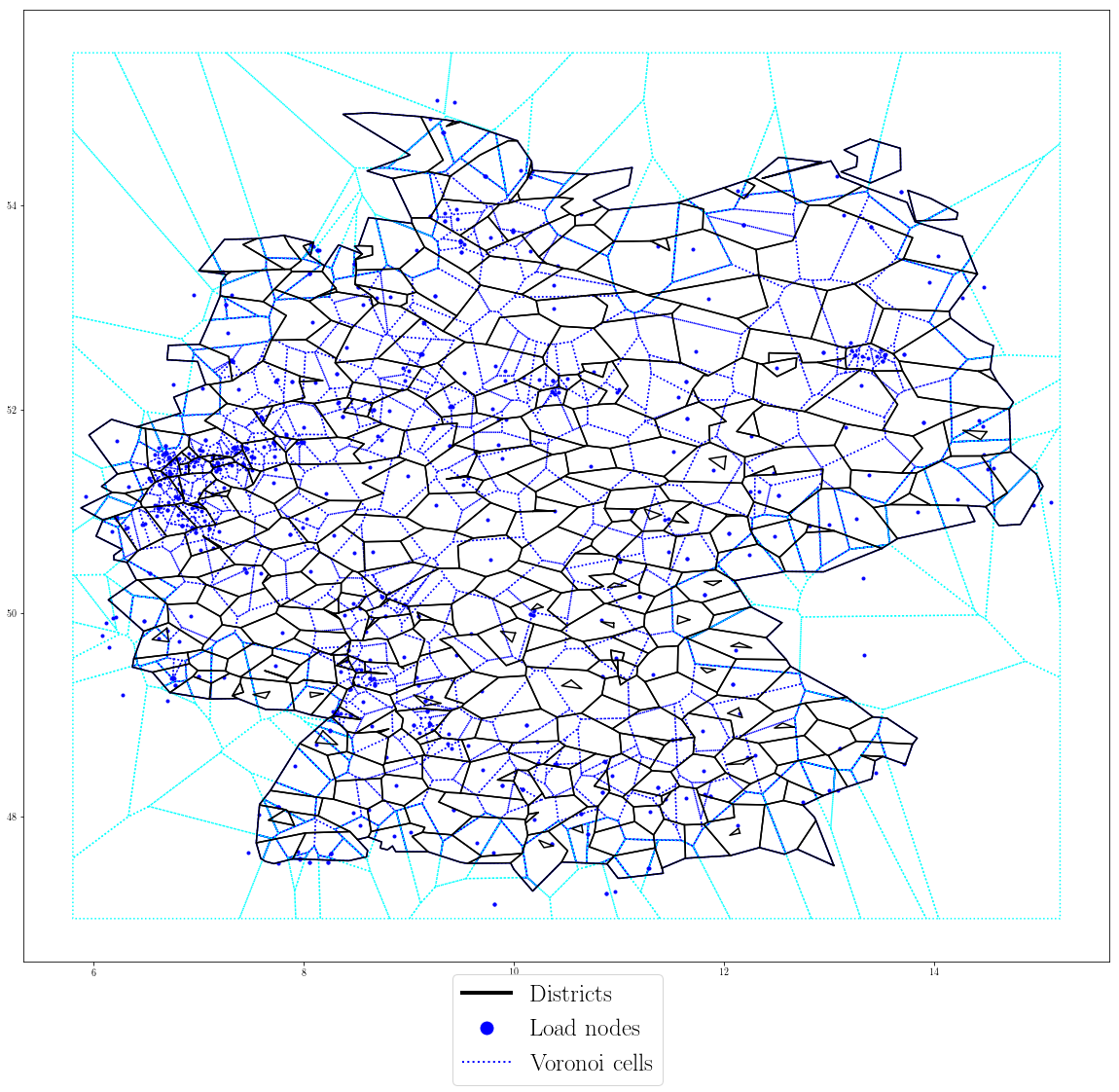}
       \caption{Subdivision of Germany according to NUTS3 districts and to Voronoi tessellation corresponding to demand SciGRID nodes.}
       \label{fig:Voronoi}
    \end{figure}

\subsection{Simulation setup}

The dataset described in~\cite{SCIGRID1} includes hourly nodal generation time series for the entire year 2011, together with data for power grid topology,  transmission lines limits, generation capacities and marginal costs. After augmenting it with the nodal demands generated as described in Subsection~\ref{ss:nodal_demand}, we are able to run realistic Optimal Power Flow (OPF) instances.
We generate blackout data via Monte Carlo simulation, as described in Algorithm~\ref{alg:montecarlo_scigrid}. First, for each of the $24\cdot 365=8760$ hourly snapshots in the year 2011, we solve the corresponding OPF using a safety factor $\lambda \in(0,1)$ (line 2 in Algorithm~\ref{alg:montecarlo_scigrid}). This corresponds to the \textit{operational} phase in our mathematical model. Note that there is no \textit{planning} phase in this simulation since we are using a model of a real-world grid.

Among the $8760$ hourly snapshots available, only a subset results in a feasible OPF, due to the introduction on the conservative parameter $\lambda$. Such snapshots are called \textit{feasible OPF snapshots}.
Then, for each feasible snapshot, we remove one line uniformly at random (line 5), and let the cascade evolve as explained in Section~\ref{s:CascadeModel} (lines 6-12). One \textit{stage} of the cascade is comprised of lines 8-11. Note that a load shedding event (line 8) may or may not happen during a given cascade stage, according to whether the previous stage line failures caused a network disconnection or not.
Finally, we store the resulting \textit{blackout realization} (line 13) expressed in terms of the total number of customers affected, obtained from the total amount of load shed via the relationship~\eqref{eq:demand_nodes}.
In general, only a subset of the feasible snapshots resulted in \textit{non-zero blackout realization}, i.e. a realization with a strictly positive blackout size, the others stopping without disconnecting the network, and thus without any load shedding.
Table~\ref{tab:scigrid_stats} reports statistics on the number of feasible OPF snapshots and non-zero blackout realizations based on Algorithm~\ref{alg:montecarlo_scigrid}.


\subsection{Results and analysis}
Given a cascade realization with $k$ stages, let $L_i$ be the cumulative load shed up to stage $i$, for $i=0,\ldots$, with the convention $L_0=0$, and let
$L_{i}-L_{i-1}$ denote the amount of load shed at stage $i$. The \textit{number of load shedding events}, in a blackout realization with $n$ stages is
\begin{equation*}
J=\left\vert \{i=1,\ldots,k\,:\, L_i-L_{i-1}>0\}\vert\right .
\end{equation*}
Fig.~\ref{fig:jumps} (corresponding to Figure $3$ in the main paper) reports the histogram and the CCDF of the total number of load shedding events in the SciGRID network, for different values of $\lambda$.
For a moderate loading
factor $\lambda = 0.7$, nearly $98\%$
of the blackouts involve just a single load shedding event, corresponding to a network disconnection. Even for a high loading factor $\lambda= 0.9$, $90\%$
of the blackouts involve just a single disconnection, and the fraction of blackouts with four or more disconnections remains below $4\%$ in all cases, as can be seen from Fig.~\ref{fig:jumps2}
These observations
are typical properties that follow from our framework, and sharply contrast the branching process
approximations where many small jumps take place.

We note that, due to the small dimension of the network and the fact that German city sizes are kept frozen (as opposed to our mathematical model where $X_1,\ldots,X_n$ are random variables). For a sufficiently large network, a frozen version of our model still leads to the correct power law behavior, as we show in Section~\ref{sec:fixed}.

%

\begin{figure}[hbt!]
     \centering
\begin{subfigure}[t]{0.55\columnwidth}
             \includegraphics[width=\textwidth]{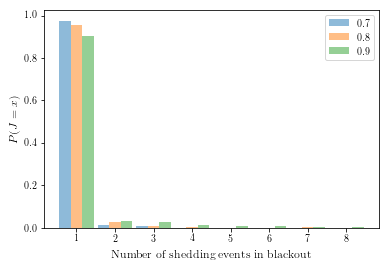}
     \caption{\small{Histogram of $J$.}}
\label{fig:jumps1}
  \end{subfigure}
\begin{subfigure}[t]{0.55\columnwidth}
                   \includegraphics[width=\textwidth]{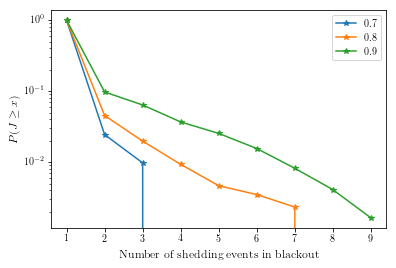}
          \caption{\small{CCDF of $J$.}}
          \label{fig:jumps2}
             \end{subfigure}
   \caption{Statistics for the total number of shedding events $J$ in the SciGRID simulation.\label{fig:jumps}}
\end{figure}
 \begin{table}[hbt!]
\centering
\begin{tabular}{c|c|c|c|c|c}
\hline
\hline
\small{loading factor $\lambda$} & \small{$\nfeas$} &
\small{$\nnonzero$}  \\
\hline
$0.7$ & $3718$ & $614$  \\
$0.8$ & $4988$ & $858$  \\
$0.9$ & $6127$ & $1220$  \\
\hline
\hline
\end{tabular}
\caption{Key statistics for the SciGRID case study. $\nfeas$ is the number of feasible OPF snapshots, and $\nnonzero$ is the number of nonzero blackout realizations.}
\label{tab:scigrid_stats}
\end{table}

\clearpage

\begin{figure*}[hbt!]
\centering
  \begin{subfigure}[t]{0.25\textwidth}
    \includegraphics[width=\textwidth]{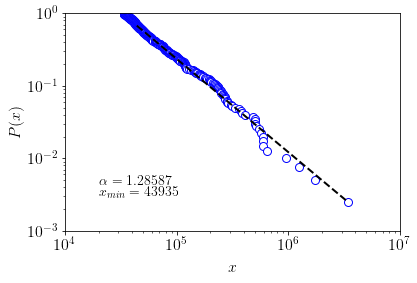}
     \caption{\small{PLFIT, German cities.}}
  \end{subfigure}
  \begin{subfigure}[t]{0.25\textwidth}
          \includegraphics[width=\textwidth]{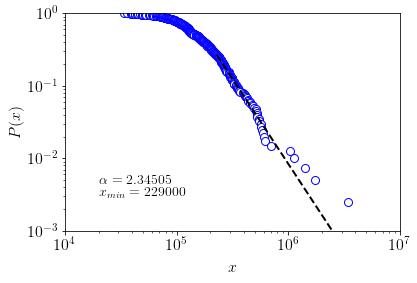}
          \caption{\small{PLFIT, German districts.}}
             \end{subfigure}
  \begin{subfigure}[t]{0.25\textwidth}
          \includegraphics[width=\textwidth]{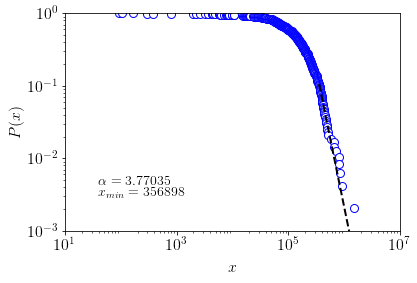}
               \caption{\small{PLFIT, SciGRID nodes.}
     }
      \end{subfigure}
  \begin{subfigure}[t]{0.25\textwidth}
    \includegraphics[width=\textwidth]{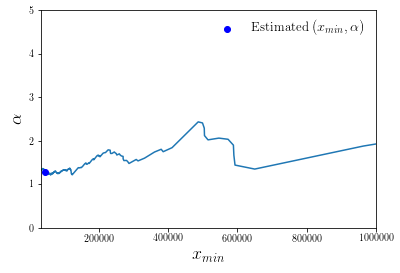}
          \caption{Hill plot, German cities.}
  \end{subfigure}
  \begin{subfigure}[t]{0.25\textwidth}
          \includegraphics[width=\textwidth]{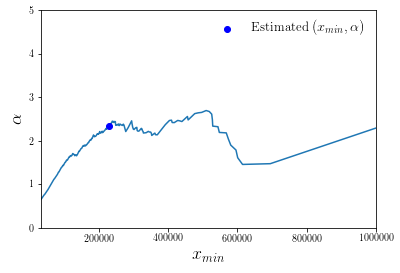}
          \caption{\small{Hill plot, German districts.}}
      \end{subfigure}
  \begin{subfigure}[t]{0.25\textwidth}
          \includegraphics[width=\textwidth]{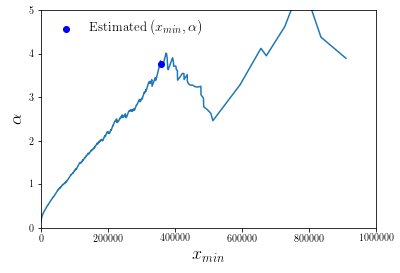}
          \caption{\small{Hill plot, SciGRID nodes.}}
      \end{subfigure}
      \caption{PLFIT results for German cities, districts and nodes population.\label{fig:SciGRID}}
\end{figure*}

\begin{figure*}[hbt!]
\centering
  \begin{subfigure}[t]{0.25\textwidth}
    \includegraphics[width=\textwidth]{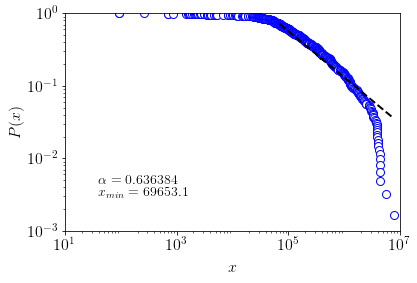}
     \caption{\small{PLFIT, $\lambda=0.7$.}}
  \end{subfigure}
  \begin{subfigure}[t]{0.25\textwidth}
          \includegraphics[width=\textwidth]{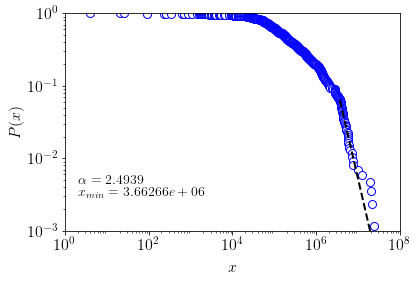}
          \caption{\small{PLFIT, $\lambda=0.8$}}
             \end{subfigure}
         \begin{subfigure}[t]{0.25\textwidth}
          \includegraphics[width=\textwidth]{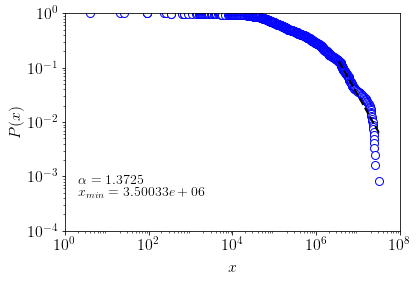}
               \caption{\small{PLFIT, $\lambda=0.9$.}
     }
      \end{subfigure}
  \begin{subfigure}[t]{0.25\textwidth}
    \includegraphics[width=\textwidth]{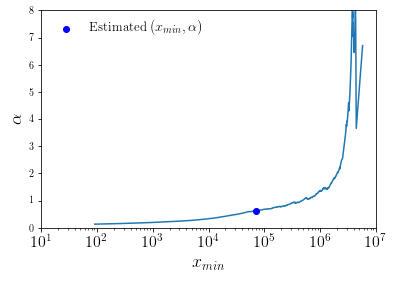}
          \caption{\small{Hill plot, $\lambda=0.7$.}}
  \end{subfigure}
        \begin{subfigure}[t]{0.25\textwidth}
          \includegraphics[width=\textwidth]{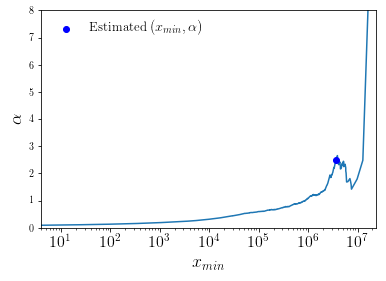}
          \caption{\small{Hill plot, $\lambda=0.8$.}}
      \end{subfigure}
        \begin{subfigure}[t]{0.25\textwidth}
          \includegraphics[width=\textwidth]{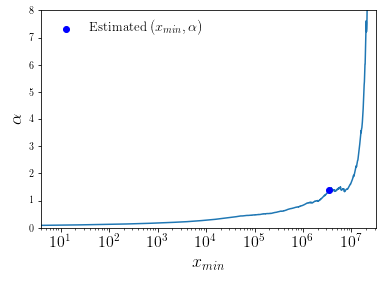}
          \caption{\small{Hill plot, $\lambda=0.9$.}}
      \end{subfigure}
      \caption{Results for SciGRID blackout simulation for different values of $\lambda$. a,b,c): PLFIT results and log-log plot of the ccdf of the number of customers affected; d,e,f) Hill plots.\label{fig:SciGRID_results}}
\end{figure*}

\end{document}